%% file: main.tex
\documentclass[11pt,a4paper]{article}

\newif\iflipics
\makeatletter
\@ifclassloaded{lipics-v2021}{\lipicstrue}{\lipicsfalse}
\makeatother
\newcommand{\LIPICS}[1]{\iflipics #1\fi}
\newcommand{\ARXIV}[1]{\iflipics\else #1\fi}
\newcommand{\LIPICSorARXIV}[2]{\iflipics #1\else #2\fi}

\ARXIV{\usepackage[a4paper,margin=3cm]{geometry}}
\usepackage{stmaryrd} 

\usepackage[utf8]{inputenc}
\usepackage{amsthm}
\usepackage[theorems=false,etoolbox=true,paralist=false,todo=hide,textwidth=1.5cm,amssymb=true]{generic} 
\usepackage{chngcntr} 
\usepackage[xcolor,hyperref,notion,quotation,paper]{knowledge}
\ARXIV{\usepackage[numbers,sort]{natbib}}
\usepackage{calc}
\usepackage{adjustbox} 
\usepackage{collcell} 

\input{appearence}

\input{macros}

\input{notions}
\input{frontmatter}

\begin{document}
    \maketitle
    \input{abstract}
    \ARXIV{\tableofcontents}
    
\section{Introduction}\label{sec:intro}

    \input{intro}

\section{Preliminaries}\label{sec:preliminaries}

    \input{functions}

    \input{terms}

\section{Minimality via category theory}\label{sec:category}

    \input{epis-monos}

    \input{factorization-systems}

\section{The transducer model}\label{sec:model}

    \input{data}

    \input{transducers}
    \input{morphisms}

\section{Simple criteria for transducer minimization}\label{sec:minimization}
    
    \input{images-of-morphisms}
    \input{initial-transducer}

    \input{final-transducer}

    \input{minimization-theorems}

\section{Applications to string-to-term transducers}\label{sec:applications}

    \input{downward-stt}

    \input{upward-stt}

\section{Conclusion}\label{sec:conclusion}
    \input{conclusion}

\bibliographystyle{plainurl}
\bibliography{biblio}


\end{document}

%% file: appearence.tex
\makeatletter
\def\@LIPICSplainheadfont{\bfseries} 
\makeatother

\LIPICSorARXIV{
\newcommand\myparagraph[1]{\medskip\noindent{\sffamily\bfseries #1.}\xspace\ignorepars}
}{
\newcommand\myparagraph[1]{\subsection{#1}}  
}

\usetikzlibrary{calc}
\usetikzlibrary{patterns}
\usetikzlibrary{decorations}
\usetikzlibrary{decorations.markings}
\usetikzlibrary{decorations.pathmorphing}
\usetikzlibrary{decorations.pathreplacing}
\usetikzlibrary{arrows}
\usetikzlibrary{arrows.meta}
\usetikzlibrary{bending}
\usetikzlibrary{positioning}
\usetikzlibrary{matrix}

\tikzset{every node/.style={font=\vphantom{Ag}, minimum size=1mm, inner sep=0pt, outer sep=0pt}}
\tikzset{every path/.style={line width=0.5pt, shorten >=3pt, shorten <=3pt}}

\pgfdeclarelayer{background}
\pgfsetlayers{background,main}

\tikzstyle{reverseclip}=[insert path={(current page.north east) --
  (current page.south east) --
  (current page.south west) --
  (current page.north west) --
  (current page.north east)}
]

\tikzstyle{dot}=[draw, shape=circle, fill, style={font=\vphantom{}}]
\tikzstyle{bluebox}=[shape=rectangle, fill=cyan!50, minimum size=4mm, inner sep=2pt, outer sep=0pt]
\tikzstyle{ubrace} = [draw, thick, decoration={brace, mirror, raise=0.0cm}, decorate, 
                      every node/.style={anchor=north, yshift=-0.1cm}]

\tikzset{ftip/.tip = {>[angle'=45, bend, round, scale=3]}}
\tikzset{fftip/.tip = {>[angle'=45, bend, round, scale=3, sep=-4pt]>[angle'=45, bend, round, scale=3]}}
\tikzset{btip/.tip = {<[angle'=45, bend, round, scale=3]}}
\tikzset{bbtip/.tip = {<[angle'=45, bend, round, scale=3, sep=-4pt]<[angle'=45, bend, round, scale=3]}}

\tikzstyle{outline} = [preaction={-, draw, shorten >=2pt, shorten <=2pt, line width=3pt, white}]
\tikzstyle{arrow} = [arrows = {-ftip}]

\tikzstyle{partial arrow} = [arrows = {-ftip}]

\tikzstyle{bold arrow} = [arrows = {-ftip[scale=1.5,line width=1.5]}, line width=1.5]

\tikzstyle{over} = [preaction={-, draw, shorten >=1pt, shorten <=1pt, line width=5pt, white}]

\tikzstyle{epi arrow} = [arrows = {-fftip}]
\tikzstyle{mono arrow} = [arrows = {btip-ftip}]
\tikzstyle{reg epi arrow} = [arrows = {-fftip}]
\tikzstyle{reg mono arrow} = [arrows = {bbtip-ftip}]

\tikzstyle{loop above} = [loop, looseness=10, in=60, out=120]
\tikzstyle{loop below} = [loop, looseness=10, in=-60, out=-120]

\pgfkeys{
  /triparams/.is family, /triparams,
  w/.store in=\triW,
  h/.store in=\triH,
  pattern/.store in=\triPat,
  color/.store in=\triCol,
  draw/.store in=\triDrw,
  w=1, h=1, color=gray, draw=black,
}

\pgfkeys{
  /triparams/.is family, /triparams,
  w/.store in=\triW,
  h/.store in=\triH,
  color/.store in=\triCol,
  draw/.store in=\triDrw,
  w=0.5, h=0.5, color=nicered, draw=black,
}

\tikzset{
  pics/tri/.default={}, 
  pics/tri/.style args={#1}{
    code={
      \pgfkeys{/triparams, #1}
      \pgfmathsetmacro{\triHalf}{0.5*\triW}
      \coordinate (-baseL) at (-\triHalf,-0.5*\triH);
      \coordinate (-baseR) at ( \triHalf,-0.5*\triH);
      \coordinate (-apex)  at (0,0.5*\triH);
      \path[draw=\triDrw, thick, fill=\triCol]
        (-\triHalf,-0.5*\triH) -- (\triHalf,-0.5*\triH) -- (0,0.5*\triH) -- cycle;
    }
  },
  pics/triOutline/.default={},
  pics/triOutline/.style args={#1}{
    code={
      \pgfkeys{/triparams, w=3.6, h=2.6, draw=black, #1}
      \pgfmathsetmacro{\triHalf}{0.5*\triW}
      \coordinate (-baseL) at (-\triHalf,-0.5*\triH);
      \coordinate (-baseR) at ( \triHalf,-0.5*\triH);
      \coordinate (-apex)  at (0,0.5*\triH);
      \draw[\triDrw, thick] (-\triHalf,-0.5*\triH) -- (\triHalf,-0.5*\triH) -- (0,0.5*\triH) -- cycle;
    }
  }
}

\makeatletter
\pgfdeclareshape{register}{
  \inheritsavedanchors[from=rectangle] 
  \inheritanchorborder[from=rectangle] 
  \inheritanchor[from=rectangle]{center}
  \inheritanchor[from=rectangle]{base}
  \inheritanchor[from=rectangle]{north}
  \inheritanchor[from=rectangle]{south}
  \inheritanchor[from=rectangle]{west}
  \inheritanchor[from=rectangle]{east}
  \backgroundpath{
    \southwest \pgf@xa=\pgf@x \pgf@ya=\pgf@y
    \northeast \pgf@xb=\pgf@x \pgf@yb=\pgf@y
    \pgfpathmoveto{\pgfpoint{\pgf@xa}{\pgf@ya}}
    \pgfpathlineto{\pgfpoint{\pgf@xa}{\pgf@yb}}
    \pgfpathlineto{\pgfpoint{(\pgf@xa+\pgf@xb)/2}{\pgf@yb+0.5*(\pgf@xb-\pgf@xa)}}
    \pgfpathlineto{\pgfpoint{\pgf@xb}{\pgf@yb}}
    \pgfpathlineto{\pgfpoint{\pgf@xb}{\pgf@ya}}
    \pgfpathlineto{\pgfpoint{(\pgf@xa+\pgf@xb)/2}{\pgf@ya-0.5*(\pgf@xb-\pgf@xa)}}
    \pgfpathclose
  }
}
\pgfdeclareshape{gap}{
  \inheritsavedanchors[from=rectangle] 
  \inheritanchorborder[from=rectangle] 
  \inheritanchor[from=rectangle]{center}
  \inheritanchor[from=rectangle]{base}
  \inheritanchor[from=rectangle]{north}
  \inheritanchor[from=rectangle]{south}
  \inheritanchor[from=rectangle]{west}
  \inheritanchor[from=rectangle]{east}
  \backgroundpath{
    \southwest \pgf@xa=\pgf@x \pgf@ya=\pgf@y
    \northeast \pgf@xb=\pgf@x \pgf@yb=\pgf@y
    \pgfpathmoveto{\pgfpoint{\pgf@xa}{\pgf@ya}}
    \pgfpathlineto{\pgfpoint{\pgf@xa}{\pgf@yb}}
    \pgfpathlineto{\pgfpoint{(\pgf@xa+\pgf@xb)/2}{\pgf@yb-0.5*(\pgf@xb-\pgf@xa)}}
    \pgfpathlineto{\pgfpoint{\pgf@xb}{\pgf@yb}}
    \pgfpathlineto{\pgfpoint{\pgf@xb}{\pgf@ya}}
    \pgfpathlineto{\pgfpoint{(\pgf@xa+\pgf@xb)/2}{\pgf@ya+0.5*(\pgf@xb-\pgf@xa)}}
    \pgfpathclose
  }
}
\pgfdeclareshape{top gap}{
  \inheritsavedanchors[from=rectangle] 
  \inheritanchorborder[from=rectangle] 
  \inheritanchor[from=rectangle]{center}
  \inheritanchor[from=rectangle]{base}
  \inheritanchor[from=rectangle]{north}
  \inheritanchor[from=rectangle]{south}
  \inheritanchor[from=rectangle]{west}
  \inheritanchor[from=rectangle]{east}
  \backgroundpath{
    \southwest \pgf@xa=\pgf@x \pgf@ya=\pgf@y
    \northeast \pgf@xb=\pgf@x \pgf@yb=\pgf@y
    \pgfpathmoveto{\pgfpoint{\pgf@xa}{\pgf@ya}}
    \pgfpathlineto{\pgfpoint{\pgf@xa}{\pgf@yb}}
    \pgfpathlineto{\pgfpoint{(\pgf@xa+\pgf@xb)/2}{\pgf@yb-0*(\pgf@xb-\pgf@xa)}}
    \pgfpathlineto{\pgfpoint{\pgf@xb}{\pgf@yb}}
    \pgfpathlineto{\pgfpoint{\pgf@xb}{\pgf@ya}}
    \pgfpathlineto{\pgfpoint{(\pgf@xa+\pgf@xb)/2}{\pgf@ya+0.5*(\pgf@xb-\pgf@xa)}}
    \pgfpathclose
  }
}
\pgfdeclareshape{bottom gap}{
  \inheritsavedanchors[from=rectangle] 
  \inheritanchorborder[from=rectangle] 
  \inheritanchor[from=rectangle]{center}
  \inheritanchor[from=rectangle]{base}
  \inheritanchor[from=rectangle]{north}
  \inheritanchor[from=rectangle]{south}
  \inheritanchor[from=rectangle]{west}
  \inheritanchor[from=rectangle]{east}
  \backgroundpath{
    \southwest \pgf@xa=\pgf@x \pgf@ya=\pgf@y
    \northeast \pgf@xb=\pgf@x \pgf@yb=\pgf@y
    \pgfpathmoveto{\pgfpoint{\pgf@xa}{\pgf@ya}}
    \pgfpathlineto{\pgfpoint{\pgf@xa}{\pgf@yb}}
    \pgfpathlineto{\pgfpoint{(\pgf@xa+\pgf@xb)/2}{\pgf@yb-0.5*(\pgf@xb-\pgf@xa)}}
    \pgfpathlineto{\pgfpoint{\pgf@xb}{\pgf@yb}}
    \pgfpathlineto{\pgfpoint{\pgf@xb}{\pgf@ya}}
    \pgfpathlineto{\pgfpoint{(\pgf@xa+\pgf@xb)/2}{\pgf@ya+0*(\pgf@xb-\pgf@xa)}}
    \pgfpathclose
  }
}
\makeatother
\tikzstyle{register} = [draw, nicecyan, fill=nicecyan!50!white, shape=register, inner sep=0pt, outer sep=0pt]
\tikzstyle{gap} = [draw, nicered, fill=nicered!50!white, shape=gap, inner sep=0pt, outer sep=0pt]
\tikzstyle{top gap} = [draw, nicered, fill=nicered!50!white, shape=top gap, inner sep=0pt, outer sep=0pt]
\tikzstyle{bottom gap} = [draw, nicered, fill=nicered!50!white, shape=bottom gap, inner sep=0pt, outer sep=0pt]

\definecolor{nicecyan}{HTML}{006165}
\definecolor{nicered}{HTML}{DB3A34}
\definecolor{nicegreen}{HTML}{6D972E}

\knowledgestyle*{intro notion}{color=blue, emphasize}
\knowledgestyle*{notion}{color=black}
\knowledgestyle*{unknown (cont)}{color=Orange}

\ARXIV{
\newtheorem{theorem}{Theorem}
\newtheorem{lemma}[theorem]{Lemma}
\newtheorem{proposition}[theorem]{Proposition}
\newtheorem{corollary}[theorem]{Corollary}
\newtheorem{example}[theorem]{Example}
}
\newtheorem{myclaim}[theorem]{Claim}

\newtheorem{assumption}{Assumption}

\newcolumntype{C}[1]{>{\hfil$}p{#1}<{$\hfil}}

%% file: macros.tex
\newcommand\gabriele[2][]{{\color{black}\todo[color=green!30,#1]{{\bf Gabriele:} #2}}\ignorespaces}
\ignorespaces

\ignorespaces

\newcommand{\qsep}{\mathord{\vcenter{\hbox{\LIPICSorARXIV{\scriptsize}{\footnotesize}$\,\blacksquare$}}}}

\newcommand{\usep}{\mathord{\vcenter{\hbox{\LIPICSorARXIV{\footnotesize}{\scriptsize}$\blacktriangle$}}}}
\newcommand{\dsep}{\mathord{\vcenter{\hbox{\LIPICSorARXIV{\footnotesize}{\scriptsize}$\blacktriangledown$}}}}

\NewDocumentCommand{\Dom}{s r()}{\IfBooleanTF{#1}{
	 \operatorname{Dom}(#2)}{
	"\operatorname{Dom}(#2)@domain"}}

\NewDocumentCommand{\Cod}{s r()}{\IfBooleanTF{#1}{
	 \operatorname{Cod}(#2)}{
	"\operatorname{Cod}(#2)@codomain"}}

\NewDocumentCommand{\Rng}{s r()}{\IfBooleanTF{#1}{
	 \operatorname{Rng}(#2)}{
	"\operatorname{Rng}(#2)@range"}}

\let\rightharpoonup\to

\NewDocumentCommand{\comp}{s}{\IfBooleanTF{#1}{
	 \mkern 1mu;\mkern 2mu}{
	 "\mkern 1mu;\mkern 2mu @composition"}}

\NewDocumentCommand{\vcomp}{s}{\IfBooleanTF{#1}{
	 \mkern 1mu;\mkern 2mu}{
	 "\mkern 1mu;\mkern 2mu @pointwise composition"}}

\NewDocumentCommand{\residual}{s m m}{\IfBooleanTF{#1}{
	 #2 \mkern 1mu\backslash\mkern 2mu {#3}}{
	 "#2 \mkern 1mu\backslash\mkern 2mu {#3} @residual"}}

\NewDocumentCommand{\vresidual}{s m m}{\IfBooleanTF{#1}{
	 #2 \mkern 1mu\backslash\mkern 2mu {#3}}{
	 "#2 \mkern 1mu\backslash\mkern 2mu {#3} @residual vector"}}

\NewDocumentCommand{\inverse}{s m}{\IfBooleanTF{#1}{
	 #2^{-1}}{
	 "#2^{-1} @inverse"}}

\NewDocumentCommand{\BottomUpTerms}{s O{\Gamma*}}{\IfBooleanTF{#1}{
	 \bbT^\uparrow[#2]}{
	"\bbT^\uparrow[#2] @bottom-up term algebra"}}

\NewDocumentCommand{\TopDownTerms}{s O{\Gamma*}}{\IfBooleanTF{#1}{
	 \bbT^\downarrow[#2]}{
	"\bbT^\downarrow[#2] @top-down term algebra"}}

\NewDocumentCommand{\epi}{s o}{\IfBooleanTF{#1}
    {\IfValueTF{#2}{\stackrel{#2}{\twoheadrightarrow}}{\twoheadrightarrow}}
    {\mathbin{"\IfValueTF{#2}{\stackrel{#2}{\twoheadrightarrow}}{\twoheadrightarrow} @epi"}}}

\NewDocumentCommand{\threeheadrightarrow}{}{\twoheadrightarrow\kern-0.45em\clipbox{{0.5\width} 0pt 0pt 0pt}{\ensuremath{\rightarrow}}}
\NewDocumentCommand{\sepi}{s o}{\IfBooleanTF{#1}
    {\mathbin{\IfValueTF{#2}{\stackrel{#2}{\threeheadrightarrow}}{\threeheadrightarrow}}}
    {\mathbin{"\IfValueTF{#2}{\stackrel{#2}{\threeheadrightarrow}}{\threeheadrightarrow} @strong epi"}}}

\NewDocumentCommand{\mono}{s o}{\IfBooleanTF{#1}
    {\IfValueTF{#2}{\stackrel{#2}{\rightarrowtail}}{\rightarrowtail}}
    {\mathbin{"\IfValueTF{#2}{\stackrel{#2}{\rightarrowtail}}{\rightarrowtail} @mono"}}}

\NewDocumentCommand{\rightarrowtailtail}{}{\clipbox{0pt 0pt {0.5\width} 0pt}{\ensuremath{\rightarrowtail}}\kern-0.63em\rightarrowtail}
\NewDocumentCommand{\smono}{s o}{\IfBooleanTF{#1}
    {\IfValueTF{#2}{\stackrel{#2}{\rightarrowtailtail}}{\rightarrowtailtail}}
    {\mathbin{"\IfValueTF{#2}{\stackrel{#2}{\rightarrowtailtail}}{\rightarrowtailtail} @strong mono"}}}

\NewDocumentCommand{\iso}{s o}{\IfBooleanTF{#1}
    {\IfValueTF{#2}{\stackrel{#2}{\simeq}}{\simeq}}
    {\mathrel{"\IfValueTF{#2}{\stackrel{#2}{\simeq}}{\simeq} @isomorphism"}}}

\NewDocumentCommand{\aiso}{s o}{\IfBooleanTF{#1}
    {\IfValueTF{#2}{\stackrel{#2}{\simeq}}{\simeq}}
    {\mathrel{"\IfValueTF{#2}{\stackrel{#2}{\simeq}}{\simeq} @arrow isomorphism"}}}

\NewDocumentCommand{\viso}{s o}{\IfBooleanTF{#1}
    {\IfValueTF{#2}{\stackrel{#2}{\cong}}{\cong}}
    {\mathrel{"\IfValueTF{#2}{\stackrel{#2}{\cong}}{\cong} @vector isomorphism"}}}

\NewDocumentCommand{\repr}{s m}{\IfBooleanTF{#1}
    {r_{#2}}
    {"r_{#2} @representative"}}

\NewDocumentCommand{\isov}{s m e{^}}{\IfBooleanTF{#1}
    {j_{#2}\IfValueT{#3}{^{#3}}}
    {"j_{#2}\IfValueT{#3}{^{#3}} @witnessing isomorphism"}}

\NewDocumentCommand{\emb}{s o}{\IfBooleanTF{#1}
    {\mathbin{\IfValueTF{#2}{\stackrel{#2}{\hookrightarrow}}{\hookrightarrow}}}
    {\mathbin{"\IfValueTF{#2}{\stackrel{#2}{\hookrightarrow}}{\hookrightarrow} @subterm embedding"}}}

\NewDocumentCommand{\multiplicity}{s m}{\IfBooleanTF{#1}
    {#2}
    {"#2 @multiplicity"}}
\NewDocumentCommand{\Subterms}{s m}{\IfBooleanTF{#1}
    {#2_\downarrow}
    {"#2_\downarrow @subterms closure"}}
\NewDocumentCommand{\capacity}{s m}{\IfBooleanTF{#1}
    {\multiplicity{\Subterms*{#2}}}
    {"\multiplicity{\Subterms*{#2}} @capacity"}}

\let\oldbbD\bbD
\RenewDocumentCommand{\bbD}{s e{_^}}{\IfBooleanTF{#1}{
	 \IfValueTF{#2}{\IfValueTF{#3}{\oldbbD_{#2 \to #3}}%
	 							  {\oldbbD_{#2}}}%
	 							  {\oldbbD}}{
	 \IfValueTF{#2}{\IfValueTF{#3}{"\oldbbD_{#2 \to #3} @updates"}%
	 							  {"\oldbbD_{#2} @domain"}}%
	 							  {"\oldbbD @data structure"}}}

\NewDocumentCommand{\type}{s e{_} r()}{\IfBooleanTF{#1}{
	 \tau\IfValueT{#2}{_{#2}}(#3)}{
	"\tau\IfValueT{#2}{_{#2}}(#3) @\type*_Q(q)"}}

\NewDocumentCommand{\Data}{s e{_} r()}{\IfBooleanTF{#1}{
	 \oldbbD_{\IfValueTF{#2}{\type*_{#2}(#3)}{\type*(#3)}}}{
	"\oldbbD_{\IfValueTF{#2}{\type*_{#2}(#3)}{\type*(#3)}} @domain"}}

\NewDocumentCommand{\tildeData}{s e{_} r()}{\IfBooleanTF{#1}{
	 \tilde\oldbbD_{\IfValueTF{#2}{\type*_{#2}(#3)}{\type*(#3)}}}{
	"\tilde\oldbbD_{\IfValueTF{#2}{\type*_{#2}(#3)}{\type*(#3)}} @\tildeData*(q)"}}

\NewDocumentCommand{\emptytuple}{s}{\IfBooleanTF{#1}{
	 ()}{
	"() @empty tuple"}}

\NewDocumentCommand{\nullvalue}{s}{\IfBooleanTF{#1}{
	 \bot}{
	"\bot @\nullvalue*"}}

\let\oldbbU\bbU
\RenewDocumentCommand{\bbU}{s e{_^}}{\IfBooleanTF{#1}{
	 \oldbbU\IfValueT{#2}{_{#2}}\IfValueT{#3}{^{#3}}}{
	"\oldbbU\IfValueT{#2}{_{#2}}\IfValueT{#3}{^{#3}} @typed monoid of updates"}}

\NewDocumentCommand{\speca}{s m}{\IfBooleanTF{#1}{
	 \hat #2}{
	"\hat #2 @specification"}}

\NewDocumentCommand{\specb}{s m}{\IfBooleanTF{#1}{
	 \check #2}{
	"\check #2 @specification"}}

\NewDocumentCommand{\States}{s m}{\IfBooleanTF{#1}{
	 \langle #2\rangle}{
	"\langle #2\rangle @configuration space"}}

\NewDocumentCommand{\CStates}{s m m}{\IfBooleanTF{#1}{
	 \langle #2|#3\rangle}{
	"\langle #2|#3\rangle @constrained configuration space"}}

\NewDocumentCommand{\init}{s}{\IfBooleanTF{#1}{
	 {\triangleright}}{
	{"\triangleright @initial type"}}}
\NewDocumentCommand{\halt}{s}{\IfBooleanTF{#1}{
	 {\triangleleft}}{
	{"\triangleleft @final type"}}}

\NewDocumentCommand{\initstate}{s}{\IfBooleanTF{#1}{
	 q_{\init*}}{
	"q_{\init*} @initial state"}}
\NewDocumentCommand{\haltstate}{s}{\IfBooleanTF{#1}{
	 q_{\halt*}}{
	"q_{\halt*} @final state"}}

\NewDocumentCommand{\initdata}{s}{\IfBooleanTF{#1}{
	 d_{\init*}}{
	"d_{\init*} @initial data"}}

\NewDocumentCommand{\initspace}{s}{\IfBooleanTF{#1}{
	 \States*{\initstate*}}{
	"\States*{\initstate*} @initial configuration space"}}

\NewDocumentCommand{\haltspace}{s}{\IfBooleanTF{#1}{
	 \States*{\haltstate*}}{
	"\States*{\haltstate*} @final configuration space"}}

\let\olddelta\delta
\let\oldkappa\kappa
\let\oldchi\chi
\RenewDocumentCommand{\delta}{s e{_^}}{\IfBooleanTF{#1}{
	"\olddelta\IfValueT{#2}{_{#2}}\IfValueT{#3}{^{#3}} @induced transformation"}{
	\olddelta\IfValueT{#2}{_{#2}}\IfValueT{#3}{^{#3}}}}
\RenewDocumentCommand{\kappa}{s e{_^}}{\IfBooleanTF{#1}{
	"\oldkappa\IfValueT{#2}{_{#2}}\IfValueT{#3}{^{#3}} @induced transformation"}{
	\oldkappa\IfValueT{#2}{_{#2}}\IfValueT{#3}{^{#3}}}}
\RenewDocumentCommand{\chi}{s e{_^}}{\IfBooleanTF{#1}{
	"\oldchi\IfValueT{#2}{_{#2}}\IfValueT{#3}{^{#3}} @induced transformation"}{
	\oldchi\IfValueT{#2}{_{#2}}\IfValueT{#3}{^{#3}}}}
\NewDocumentCommand{\deltainit}{}{\delta_{\init*}}
\NewDocumentCommand{\deltahalt}{}{\delta_{\halt*}}
\NewDocumentCommand{\kappainit}{}{\kappa_{\init*}}
\NewDocumentCommand{\kappahalt}{}{\kappa_{\halt*}}
\NewDocumentCommand{\chiinit}{}{\chi_{\init*}}
\NewDocumentCommand{\chihalt}{}{\chi_{\halt*}}

\NewDocumentCommand{\Sol}{s r()}{\IfBooleanTF{#1}{
	 \operatorname{Sol}(#2)}{
	"\operatorname{Sol}(#2) @solution set"}}

\NewDocumentCommand{\Solbig}{s r()}{\IfBooleanTF{#1}{
	 \operatorname{Sol}\big(#2\big)}{
	"\operatorname{Sol}\big(#2\big) @solution set"}}

\NewDocumentCommand{\eqdata}{s e{_}}{\IfBooleanTF{#1}{
	\dsim_{#2}}{
	\mathrel{"\dsim_{#2} @equivalence@data type"}}}

\NewDocumentCommand{\neqdata}{s e{_}}{\IfBooleanTF{#1}{
	\not\dsim_{#2}}{
	\mathrel{"\not\dsim_{#2} @equivalence@data type"}}}

\NewDocumentCommand{\eqdataclass}{s e{_} O{}}{\IfBooleanTF{#1}{
	 [#3]_{\eqdata*_{#2}}}{
	"[#3]_{\eqdata*_{#2}} @equivalence@data type"}}

\NewDocumentCommand{\quotient}{s m m}{\IfBooleanTF{#1}{
	 #2_{/#3}}{
	 #2_{/#3}}}

\NewDocumentCommand{\cl}{s r()}{\IfBooleanTF{#1}{
	 \operatorname{cl}(#2)}{
	"\operatorname{cl}(#2) @closure"}}
\NewDocumentCommand{\bigcl}{s r()}{\IfBooleanTF{#1}{
	 \operatorname{cl}\big(#2\big)}{
	"\operatorname{cl}\big(#2\big) @closure"}}

\NewDocumentCommand{\clS}{s r()}{\IfBooleanTF{#1}{
	 \operatorname{cl}(#2)}{
	"\operatorname{cl}(#2) @space closure"}}
\NewDocumentCommand{\bigclS}{s r()}{\IfBooleanTF{#1}{
	 \operatorname{cl}\big(#2\big)}{
	"\operatorname{cl}\big(#2\big) @space closure"}}

\NewDocumentCommand{\identity}{m}{\operatorname{id}_{#1}}

\NewDocumentCommand{\lef}{s}{\IfBooleanTF{#1}{
	 \mathrel{\sqsubseteq}}{
	 \mathrel{"\sqsubseteq @preorder on factorizations"}}}

\NewDocumentCommand{\Hankel}{s e{_}}{\IfBooleanTF{#1}
    {H_{\IfValueTF{#2}{#2}{}}}
    {"H_{\IfValueTF{#2}{#2}{}} @Hankel matrix"}}

\NewDocumentCommand{\Divisors}{s e{_}}{\IfBooleanTF{#1}
    {G_{\IfValueTF{#2}{#2}{}}}
    {"G_{\IfValueTF{#2}{#2}{}} @divisor vector"}}

\NewDocumentCommand{\Residuals}{s e{_}}{\IfBooleanTF{#1}
    {R_{\IfValueTF{#2}{#2}{}}}
    {"R_{\IfValueTF{#2}{#2}{}} @residual matrix"}}

\NewDocumentCommand{\Partials}{s e{_}}{\IfBooleanTF{#1}
    {\partial\IfValueT{#2}{_{#2}}}
    {"\partial\IfValueT{#2}{_{#2}} @derivative"}}

\NewDocumentCommand{\myhilleq}{s e{_}}{\IfBooleanTF{#1}
    {\sim\IfValueT{#2}{_{#2}}}
    {\mathrel{"\sim\IfValueT{#2}{_{#2}} @Myhill-Nerode equivalence"}}}

\NewDocumentCommand{\CData}{s m m}{\IfBooleanTF{#1}{
	 {#2/#3}}{
	 {"#2/#3 @induced constrained set"}}}

\NewDocumentCommand{\hattau}{s d()}{\IfBooleanTF{#1}{
	 \hat\tau\IfValueT{#2}{(#2)}}{
	 "\hat\tau\IfValueT{#2}{(#2)} @\hattau*"}}

\NewDocumentCommand{\bartau}{s d()}{\IfBooleanTF{#1}{
	 \bar\tau\IfValueT{#2}{(#2)}}{
	 "\bar\tau\IfValueT{#2}{(#2)} @\bartau*"}}

\NewDocumentCommand{\myhillclass}{s e{_} m}{\IfBooleanTF{#1}
    {[#3]_{\myhilleq*[\IfValueTF{#2}{#2}{\tau}]}}
    {"[#3]_{\myhilleq*[\IfValueTF{#2}{#2}{\tau}]} @Myhill-Nerode equivalence"}}

\NewDocumentCommand{\nerodeeq}{s e{_}}{\IfBooleanTF{#1}
    {\dsim_{\IfValueTF{#2}{#2}{A}}}
    {\mathrel{"\dsim_{\IfValueTF{#2}{#2}{A}} @Nerode equivalence"}}}

\NewDocumentCommand{\nerodeclass}{s e{_} m}{\IfBooleanTF{#1}
    {[#3]_{\nerodeeq*[\IfValueTF{#2}{#2}{A}]}}
    {"[#3]_{\nerodeeq*[\IfValueTF{#2}{#2}{A}]} @Nerode equivalence"}}

\NewDocumentCommand{\Min}{s e{_}}{\IfBooleanTF{#1}
    {M_{\IfValueTF{#2}{#2}{\tau}}}
    {"M_{\IfValueTF{#2}{#2}{\tau}} @minimal transducer for $\tau$"}}

\NewDocumentCommand{\controleq}{s}{\IfBooleanTF{#1}
    {\stackrel{\scalebox{0.5}{ctl}}{=}}
    {\mathrel{"\stackrel{\scalebox{0.5}{ctl}}{=} @\controleq*"}}}

\NewDocumentCommand{\Lower}{s r()}{\IfBooleanTF{#1}
    {\mathop{L}(#2)}
    {"\mathop{L}(#2) @\Lower*(U)"}}

\NewDocumentCommand{\divisor}{s e{_}}{\IfBooleanTF{#1}
    {\mathrel{\preceq\IfValueT{#2}{_{#2}}}}
    {\mathrel{"\preceq\IfValueT{#2}{_{#2}} @\divisor*"}}}

\NewDocumentCommand{\initialtransd}{s}{\IfBooleanTF{#1}
    {I}
    {"I @initial transducer"}}

\NewDocumentCommand{\initialmorph}{s m}{\IfBooleanTF{#1}
    {\mathop{!}\nolimits^{#2}}
    {"\mathop{!}\nolimits^{#2} @initial arrow"}}
\NewDocumentCommand{\initialmorphprime}{m}{\mathop{!!}\nolimits^{#1}}
\NewDocumentCommand{\specainitialmorph}{s m}{\IfBooleanTF{#1}
    {\mathop{\hat!}\nolimits^{#2}}
    {"\mathop{\hat!}\nolimits^{#2} @specification"}}
\NewDocumentCommand{\specbinitialmorph}{s m}{\IfBooleanTF{#1}
    {\mathop{\check!}\nolimits^{#2}}
    {"\mathop{\check!}\nolimits^{#2} @specification"}}

\NewDocumentCommand{\finaltransd}{s}{\IfBooleanTF{#1}
    {M}
    {"M @final transducer"}}

\NewDocumentCommand{\finalmorph}{s m}{\IfBooleanTF{#1}
    {\mathop{!}\nolimits_{#2}}
    {"\mathop{!}\nolimits_{#2} @final arrow"}}
\NewDocumentCommand{\finalmorphprime}{m}{\mathop{!!}\nolimits_{#1}}
\NewDocumentCommand{\specafinalmorph}{s m}{\IfBooleanTF{#1}
    {\mathop{\hat!}\nolimits_{#2}}
    {"\mathop{\hat!}\nolimits_{#2} @specification"}}
\NewDocumentCommand{\specbfinalmorph}{s m}{\IfBooleanTF{#1}
    {\mathop{\check!}\nolimits_{#2}}
    {"\mathop{\check!}\nolimits_{#2} @specification"}}

\NewDocumentCommand{\Reach}{s r()}{\IfBooleanTF{#1}
    {\operatorname{Reach}(#2)}
    {"\operatorname{Reach}(#2) @pruned transducer"}}

\NewDocumentCommand{\Obs}{s r()}{\IfBooleanTF{#1}
    {\operatorname{Obs}(#2)}
    {"\operatorname{Obs}(#2) @quotiented transducer"}}

\NewDocumentCommand{\equ}{s O{U}}{\IfBooleanTF{#1}{
	 \mathrel{=_{#2}}}{
	 \mathrel{"=_{#2} @\equ*[U]"}}}

\NewDocumentCommand{\nequ}{s O{U}}{\IfBooleanTF{#1}{
	 \mathrel{\neq_{#2}}}{
	 \mathrel{"\neq_{#2} @\equ*[U]"}}}

\NewDocumentCommand{\equclass}{s O{U} m}{\IfBooleanTF{#1}{
	 [#3]_{\equ*[#2]}}{
	"[#3]_{\equ*[#2]} @\equ*[U]"}}

\NewDocumentCommand{\equclassbig}{s O{U} m}{\IfBooleanTF{#1}{
	 \left[#3\right]_{\equ*[#2]}}{
	"\left[#3\right]_{\equ*[#2]} @\equ*[U]"}}

\NewDocumentCommand{\dotGamma}{s}{\IfBooleanTF{#1}{
	 \dot\Gamma}{
	"\dot\Gamma @alphabet with inverses"}}

\NewDocumentCommand{\eqg}{s}{\IfBooleanTF{#1}{
	 \mathrel{\dot=}}{
	 \mathrel{"\dot= @identity of the free group"}}}
\NewDocumentCommand{\neqg}{s}{\IfBooleanTF{#1}{
	 \mathrel{\dot\neq}}{
	 \mathrel{"\dot\neq @identity of the free group"}}}

\NewDocumentCommand{\starg}{s}{\IfBooleanTF{#1}{
	 \circledast}{
	{"\circledast @variant of the Kleene star"}}}

\let\oldmodels\models
\RenewDocumentCommand{\models}{s}{\IfBooleanTF{#1}{
	 \mathbin{\oldmodels}}{
	 \mathbin{"\oldmodels @models"}}}

\NewDocumentCommand{\rank}{s m}{\IfBooleanTF{#1}{
	 \#{#2}}{
	"\#{#2}@rank"}}

\NewDocumentCommand{\inter}{s m m}{\IfBooleanTF{#1}{
	 {#2} \mathop{\sslash} {#3}}{
	"{#2} \mathop{\sslash} {#3}@interleaving"}}

\NewDocumentCommand{\effect}{s m}{\IfBooleanTF{#1}{
	 \operatorname{effect}(#2)}{
	"\operatorname{effect}(#2)@effect"}}

\NewDocumentCommand{\dual}{s m}{\IfBooleanTF{#1}{
	 #2^\circ}{
	"#2^\circ @dual update"}}

\NewDocumentCommand{\forwardcomp}{s}{\IfBooleanTF{#1}{
	 \normalfont\textbf{(F)}}{
	"\normalfont\textbf{(F)}@forward compatibility"}}

\NewDocumentCommand{\backwardcomp}{s}{\IfBooleanTF{#1}{
	 \normalfont\textbf{(B)}}{
	"\normalfont\textbf{(B)}@backward compatibility"}}

\NewDocumentCommand{\compatibility}{s}{\IfBooleanTF{#1}{
	 \normalfont\textbf{(F-B)}}{
	"\normalfont\textbf{(F-B)}@compatibility"}}

\newcommand\inv[1][\cA]{"S_{#1}@\inv"}
\WithSuffix\newcommand\inv*[1][\cA]{S_{#1}}

\newcommand\idupdate[1][k]{"\mathrm{id}_{#1}@\idupdate"}
\WithSuffix\newcommand\idupdate*[1][k]{\mathrm{id}_{#1}}

\NewDocumentCommand{\subst}{s o m}{%
  \IfValueTF{#2}{%
    \IfBooleanTF{#1}{ [#2/#3]}{%
    				 "[#2/#3] @substitution"}%
  }{%
    \IfBooleanTF{#1}{ [#3]}{%
    				 "[#3] @linear substitution"}%
  }%
}
\NewDocumentCommand{\bigsubst}{s o m}{%
  \IfValueTF{#2}{%
    \IfBooleanTF{#1}{ \left[#2/#3\right]}{%
    				 "\left[#2/#3\right] @substitution"}%
  }{%
    \IfBooleanTF{#1}{ \left[#3\right]}{%
    				 "\left[#3\right] @linear substitution"}%
  }%
}

\NewDocumentCommand{\opentag}{m}{\underline{#1}}
\NewDocumentCommand{\closetag}{m}{\overline{#1}}

\newcommand\enc{\mathrm{enc}}

%% file: notions.tex
\knowledge{notion}
    | domain@function
    | domains@function
    | domain of definition
    | \Dom*(f)

\knowledge{notion}
    | source
    | sources
    | source set
    | source sets

\knowledge{notion}
    | codomain
    | codomains
    | \Cod*(f)

\knowledge{notion}
    | range
    | ranges
    | \Rng*(f)

\knowledge{notion}
    | functional composition
    | composition
    | compositions
    | composing
    | composed
    | compose
    | f \comp* g

\knowledge{notion}
    | factorization
    | factorizations
    | factorizes
    | factorize
    | factorized
    | factorizing

\knowledge{notion}
    | input alphabet
    | input alphabets
    | \Sigma
    | \Sigma*

\knowledge{notion}
    | output alphabet
    | output alphabets
    | \Gamma
    | \Gamma*

\knowledge{notion}
    | word substitution
    | word substitutions
    | substitution@word
    | substitutions@word

\knowledge{notion}
    | term
    | terms
    | Terms
    | free term
    | free terms
    | Free terms
\knowledge{notion}
    | linear term
    | linear terms
    | Linear terms
    | linear@term
\knowledge{notion}
    | ground term
    | ground terms
    | Ground terms
    | term@ground
    | terms@ground
\knowledge{notion}
    | first-order term
    | first-order terms
    | First-order terms
    | term@first-order
    | terms@first-order
\knowledge{notion}
    | second-order term
    | second-order terms
    | Second-order terms
    | term@second-order
    | terms@second-order
\knowledge{notion}
    | hedge
    | hedges

\knowledge{notion}
    | linear context
    | linear contexts
    | Linear contexts

\knowledge{notion}
    | t\subst*[\bar x]{\bar u}
    | first-order substitution
    | first-order substitutions
    | substitution
    | substitutions

\knowledge{notion}
    | t\subst*{\bar u}
    | linear substitution
    | linear substitutions
    | substitution@linear
    | substitutions@linear

\knowledge{notion}
    | category
    | categories
\knowledge{notion}
    | arrow
    | arrows
\knowledge{notion}
    | object
    | objects
\knowledge{notion}
    | identity
    | identities
\knowledge{notion}
    | epi 
    | epis
    | Epis
    | epi-morphism
    | epi-morphisms
    | Epi-morphisms
    | \epi*
    | \epi*[f]
\knowledge{notion}
    | mono
    | monos
    | Monos
    | mono-morphism
    | mono-morphisms
    | Mono-morphisms
    | \mono*
    | \mono*[g]
\knowledge{notion}
    | strong mono
    | strong monos
    | Strong monos
    | strong mono-morphism
    | strong mono-morphisms
    | Strong mono-morphisms
    | strong@mono
    | \smono*
    | \smono*[g]
\knowledge{notion}
    | object isomorphism
    | object isomorphisms
    | isomorphism
    | isomorphisms
    | Isomorphisms
    | isomorphic
    | invertible
    | iso
    | isos
    | \iso*
    | \iso*[i]
\knowledge{notion}
    | inverse
    | inverses
    | \inverse*{i}
\knowledge{notion}
    | arrow isomorphism
    | arrow isomorphisms
    | isomorphism@arrow
    | isomorphisms@arrow
    | isomorphic@arrow
    | \aiso*
    | \aiso*[i]
    | \aiso*[i,j]
\knowledge{notion}
    | vector isomorphism
    | isomorphism@vector
    | \viso*
    | \viso*[j]
\knowledge{notion}
    | non-trivial
\knowledge{notion}
    | representative
    | representatives
    | \repr*{w}
\knowledge{notion}
    | witnessing isomorphism
    | \isov*{w}
\knowledge{notion}
    | diagonal fill-in property

\knowledge{notion}
    | quotient
    | quotients
    | quotiented
\knowledge{notion}
    | subobject
    | subobjects
\knowledge{notion}
    | subquotient
    | subquotients
\knowledge{notion}
    | algebraically minimal
    | algebraic minimization
    | algebraic minimality
    | minimal
    | minimality
    | minimization
    | minimizing

\knowledge{notion}
    | strong factorization
    | strong factorizations
    | strong@factorization
\knowledge{notion}
    | image
    | images
    | Images
\knowledge{notion}
    | initial image
    | initial images
    | Initial images

\knowledge{notion}
    | initial
    | initial object
    | Initial
    | Initial objects
\knowledge{notion}
    | initial arrow
    | initial arrows
    | initial@arrow
    | \initialmorph*{A}
    | initial morphism
\knowledge{notion}
    | final
    | final object
    | Final
    | Final objects
\knowledge{notion}
    | final arrow
    | final arrows
    | final@arrow
    | \finalmorph*{A}
    | final morphism
\knowledge{notion}
    | initial transducer
    | \initialtransd*
\knowledge{notion}
    | final transducer
    | \finaltransd*
\knowledge{notion}
    | \Reach*(A)
    | pruned transducer
    | prune
    | pruned
    | pruning
\knowledge{notion}
    | \Obs*(A)
    | \Obs*(\Reach*(A))
    | quotiented transducer
    | quotient@transducer
    | quotiented@transducer
    | quotienting@transducer

\knowledge{notion}
    | data structure
    | data structures
    | \bbD*

\knowledge{notion}
    | data type
    | data types
    | type 
    | types

\knowledge{notion}
    | domain
    | domains
    | data domain
    | data domains
    | \bbD*_\tau

\knowledge{notion}
    | update
    | updates
    | Updates
    | updated
    | updating
    | \bbD*_\alpha^\beta

\knowledge{notion}
    | \nullvalue*
\knowledge{notion}
    | empty tuple
    | \emptytuple*

\knowledge{notion}
    | string register algebra
    | String register algebra
\knowledge{notion}
    | polynomial register algebra
    | Polynomial register algebra
\knowledge{notion}
    | affine register algebra
    | Affine register algebra
\knowledge{notion}
    | linear register algebra
    | Linear register algebra
\knowledge{notion}
    | leaf substitution algebra
    | Leaf substitution algebra
\knowledge{notion}
    | free term algebra
    | Free term algebra

\knowledge{notion}
    | copyful update
    | copyful updates
    | copyful
\knowledge{notion}
    | copyless update
    | copyless updates
    | copyless
\knowledge{notion}
    | non-erasing update
    | non-erasing updates
    | non-erasing
    | erasing update
    | erasing updates
\knowledge{notion}
    | linear update
    | linear updates
    | linear string updates
    | linear

\knowledge{notion}
    | bottom-up term algebra
    | \BottomUpTerms*

\knowledge{notion}
    | top-down term algebra
    | \TopDownTerms*

\knowledge{notion}
    | arity
    | arities
\knowledge{notion}
    | assignment
    | assignments
    | Assignments
    | memory
    | memories

\knowledge{notion}
    | \bbU*
    | \bbU*_n
    | \bbU*^m
    | \bbU*_n^m
    | typed monoid of updates
    | typed monoid
    | monoid of updates

\knowledge{notion}
    | register
    | registers
\knowledge{notion}
    | control state
    | control states
    | state
    | states
\knowledge{notion}
    | \type*_Q(q)
\knowledge{notion}
    | \Data*_Q(q)
    | \Data*(q)
    | \Data*(q')
\knowledge{notion}
    | \tildeData*{q}
    | \tildeData*{q'}
\knowledge{notion}
    | configuration space
    | configuration spaces
    | configuration space generated by $Q$
    | generated configuration space
    | generated@configuration space
    | generating@configuration space
    | space
    | spaces
    | unconstrained configuration space
    | unconstrained configuration spaces
    | \States*{Q}

\knowledge{notion}
    | configuration
    | configurations
\knowledge{notion}
    | transformation specified by
    | transformation specified
    | transformations specified by
    | transformations specified 
    | specification of the transformation
    | specifications of the transformations
    | specified by
    | specified
    | specify
    | specification
    | Specification
    | specifications
    | specification of
    | specifications of
    | transformation
    | transformations
    | Transformations
    | transform
    | transforms
    | transformed
    | \speca* f
    | \specb* f

\knowledge{notion}
    | initial type
    | \init*
\knowledge{notion}
    | final type
    | \halt*
\knowledge{notion}
    | initial data
    | \initdata*
\knowledge{notion}
    | initial control state
    | initial state
    | \initstate*
\knowledge{notion}
    | final control state
    | final state
    | \haltstate*
\knowledge{notion}
    | initial configuration space
    | \initspace*
\knowledge{notion}
    | final configuration space
    | \haltspace*
    
\knowledge{notion}
    | streaming transducer
    | streaming transducers
    | Streaming transducers
    | deterministic streaming transducer
    | deterministic streaming transducers
    | Deterministic streaming transducers
    | transducer
    | transducers
    | Transducers
\knowledge{notion}
    | subsequential transducer
    | subsequential transducers
    | Subsequential transducers
\knowledge{notion}
    | finite transducer
    | finite transducers
    | Finite transducers
    | finite@transducer
\knowledge{notion}
    | regular look-ahead

\knowledge{notion}
    | initial transformation
    | initial transformations
\knowledge{notion}
    | internal transformation
    | internal transformations
\knowledge{notion}
    | final transformation
    | final transformations
\knowledge{notion}
    | finite constraint
    | finite@constraint
\knowledge{notion}
    | induced transformation
    | induced transformations
    | transformation induced 
    | transformations induced
    | transformation@induced
\knowledge{notion}
    | output
    | outputs
    | output value
    | output values
\knowledge{notion}
    | transduction realized
    | transduction
    | transductions
    | realizable transduction
    | realizable transductions
    | realizes
    | realizes the transduction
    | realize
    | realized
    | realizing
    | realizable
    | \out
\knowledge{notion}
    | equivalent
    | equivalence
\knowledge{notion}
    | empty at

\knowledge{notion}
    | polynomial automaton
    | polynomial automata
    | Polynomial automata
\knowledge{notion}
    | polynomial map
    | polynomial maps
    | polynomial@maps
\knowledge{notion}
    | linear map
    | linear maps
\knowledge{notion}
    | Hybrid-set-vector automata
    | hybrid-set-vector automata
    | hybrid-set-vector automaton
\knowledge{notion}
    | Streaming string-to-string transducers
    | streaming string-to-string transducers
    | streaming string-to-string transducer
    | string-to-string transducers
    | string-to-string transducer
    | SSTs
    | SST
\knowledge{notion}
    | Sequential transducers
    | sequential transducers
    | sequential transducer
\knowledge{notion}
    | One-way transducers
    | one-way transducers
    | one-way transducer
\knowledge{notion}
    | Streaming string-to-term transducers
    | streaming string-to-term transducers
    | streaming string-to-term transducer
    | Streaming string-to-tree transducers
    | streaming string-to-tree transducers
    | streaming string-to-tree transducer
    | string-to-term transducers
    | string-to-term transducer
    | string-to-tree transducers
    | string-to-tree transducer
    | STTs
    | STT
\knowledge{notion}
    | downward string-to-term transducer
    | downward string-to-term transducers
    | Downward string-to-term transducers
    | Downward STTs
    | downward STTs
    | downward STT
    | downward@STT
\knowledge{notion}
    | upward string-to-term transducer
    | upward string-to-term transducers
    | Upward string-to-term transducers
    | Upward STTs
    | upward STTs
    | upward STT
    | upward@STT

\knowledge{notion}
    | transducer morphism
    | transducer morphisms
    | Transducer morphisms
    | transducer morphism from $A$ to $B$
    | morphism
    | morphisms

\knowledge{notion}
    | constraint
    | constraints
    | Constraints
\knowledge{notion}
    | \Sol*(\gamma)
    | \bbD*_\gamma
    | solution set
    | solution sets
    | Solution set
    | solution
    | solutions
\knowledge{notion}
    | constrained type
    | constrained types
    | Constrained types
    | type@constrained
    | types@constrained
\knowledge{notion}
    | constrained domain
    | constrained domains
    | Constrained domains
    | domain@constrained
    | domains@constrained
    | constrained codomain
    | constrained codomains
\knowledge{notion}
    | finitary
    | finitary constrained domains
    | finitary constrained domain
\knowledge{notion}
    | presentation@constraint

\knowledge{notion}
    | equivalence@data type
    | \eqdata*_{E'}
\knowledge{notion}
    | \cl*(D)
    | closure of $D$
    | closure
    | closures
    | closure operator
    | closed
\knowledge{notion}
    | extensive
    | extensiveness
\knowledge{notion}
    | monotone
    | monotonicity
\knowledge{notion}
    | idempotent
    | idempotency
\knowledge{notion}
    | \clS*(S)
    | space closure
    | space closures
    | closure@space
    | closures@space
    | closed@space
    | closure operator@space
    | closure of $S$

\knowledge{notion}
    | update@constrained
    | constrained update
    | constrained updates
\knowledge{notion}
    | presentation
    | presentations
    | presented by
    | presented
    | presenting
\knowledge{notion}
    | \forwardcomp*
    | forward compatibility
    | forward compatible
    | forward compatibility between register constraints
    | compatibility
\knowledge{notion}
    | \backwardcomp*
    | backward compatibility
    | backward compatible
    | backward compatibility between gap constraints
\knowledge{notion}
    | assignment@constrained
    | constrained assignment
    | constrained assignments
    | Constrained assignments
\knowledge{notion}
    | configuration space@constrained 
    | space@constrained 
    | spaces@constrained
    | constrained configuration space
    | constrained configuration spaces
    | Constrained configuration spaces
    | constrained space
\knowledge{notion}
    | configuration@constrained
    | configurations@constrained
    | constrained configuration
    | constrained configurations
    | Constrained configurations
\knowledge{notion}
    | transformation@constrained
    | constrained transformation
    | constrained transformations
    | Constrained transformations
    | specifies@constrained transformation
    | specified by@constrained transformation
    | specified@constrained transformation
    | specification@constrained transformation
    | specifications@constrained tranformation
\knowledge{notion}
    | transducer@constrained
    | transducers@constrained
    | constrained transducer
    | constrained transducers
    | Constrained transducers
\knowledge{notion}
    | constrained transducer morphism
    | transducer morphism@constrained
    | transducer morphisms@constrained
    | Transducer morphisms@constrained
\knowledge{notion}
    | initial configuration space@constrained
\knowledge{notion}
    | final configuration space@constrained
\knowledge{notion}
    | initial transformation@constrained
    | constrained initial transformation
    | constrained initial transformations
\knowledge{notion}
    | internal transformation@internal
    | constrained internal transformation
    | constrained internal transformations
    | transition
    | transitions
\knowledge{notion}
    | final transformation@constrained
    | constrained final transformation
    | constrained final transformations

\knowledge{notion}
    | inner divisor
    | inner divisors
    | (inner) divisor
    | divisor
    | divisors
    | divide
    | divides
    | divided
    | dividing
    | factor through
    | factors through
\knowledge{notion}
    | (outer) residual
    | (outer) residuals
    | outer residual
    | outer residuals
    | residual
    | residuals
    | \residual*{g}{u}
\knowledge{notion}
    | common divisor
    | common divisors
    | common inner divisor
    | common (inner) divisor
\knowledge{notion}
    | epi divisor
    | epi divisors
\knowledge{notion}
    | epi common divisor
    | epi common divisors
    | (epi) common divisor
    | (epi) common divisors
\knowledge{notion}
    | residual vector
    | residual vectors
    | \vresidual*{g}{U}
\knowledge{notion}
    | greatest common divisor
    | greatest common divisors
    | Greatest common divisors
    | greatest
    | GCD
    | GCDs
\knowledge{notion}
    | epi greatest common divisor
    | epi greatest common divisors
    | Epi greatest common divisors
    | EGCD
    | EGCDs
\knowledge{notion}
    | standard data structure
    | standard
\knowledge{notion}
    | jointly coprime
    | joint coprimality
    | coprimality

\knowledge{notion}
    | copy
    | copies
\knowledge{notion}
    | disjoint
\knowledge{notion}
    | subterm embedding
    | subterm embeddings
    | embedding@subterm
    | embeddings@subterm
    | embeds@subterm
    | embedded@subterm
    | \emb*
\knowledge{notion}
    | origin
    | origins
    | originates
    | originating
\knowledge{notion}
    | \Lower*(\cU)
\knowledge{notion}
    | ideal
\knowledge{notion}
    | downward closed
\knowledge{notion}
    | upward directed
\knowledge{notion}
    | conflict
    | conflicts
    | violation
    | violating pair
    | violating pairs
    | pair@violation
    | pairs@violation
\knowledge{notion}
    | capacity
    | capacities
\knowledge{notion}
    | subterms closure
\knowledge{notion}
    | multiplicity
    | multiplicities

\knowledge{notion}
    | vector
    | vectors
\knowledge{notion}
    | sub-vector
    | sub-vectors
    | U[a-]
    | V[a-]
\knowledge{notion}
    | matrix
    | matrices
\knowledge{notion}
    | column vector
    | column vectors
\knowledge{notion}
    | complete
    | completeness
\knowledge{notion}
    | Hankel matrix
    | \Hankel*

\knowledge{notion}
    | pointwise composition
    | pointwise compositions
    | composition@pointwise
    | compositions@pointwise
    | composes@pointwise
    | compose@pointwise
    | composed@pointwise
    | composing@pointwise
\knowledge{notion}
    | common factorization
    | common factorizations
    | commonly factorize
    | commonly factorizes
    | commonly factorized
    | commonly factorizing
    | factorization@common
    | factorizations@common
    | factorize@common
    | factorizes@common
    | factorized@common
    | factorizing@common

\knowledge{notion}
    | preorder on common factorizations
    | \lef*
\knowledge{notion}
    | witness
    | witnessed
    | witnesses
    | witnessing
\knowledge{notion}
    | equivalent common factorizations
    | equivalent@common factorization
    | equivalence@common factorization
\knowledge{notion}
    | greatest common factorization
    | greatest common factorizations
    | Greatest common factorizations
    | GCF
    | GCFs

\knowledge{notion}
    | divisor vector
    | \Divisors*
    | \Divisors*[w]
\knowledge{notion}
    | residual matrix
    | \Residuals*
    | \Residuals*[w,-]
\knowledge{notion}
    | derivative
    | \Partials*_a
    | \Partials*_a[w]
\knowledge{notion}
    | Myhill-Nerode equivalence
    | \myhilleq*
    | equivalence@Myhill-Nerode
    | equivalent@Myhill-Nerode

\knowledge{notion}
    | divisor-equivalent
    | divisor-equivalence
    | equivalent@divisor
    | equivalence@divisor

\knowledge{notion}
    | \hattau*
    | \hattau*(w)
\knowledge{notion}
    | \bartau*
    | \bartau*(w,i)
\knowledge{notion}
    | \Min*
    | minimal transducer for $\tau$

\knowledge{notion}
    | \controleq*
\knowledge{notion}
    | Nerode equivalence of $A$
    | Nerode equivalence
    | \nerodeeq*
    | equivalence@Nerode
    | equivalent@Nerode

\knowledge{notion}
    | generalize
    | generalizes
    | generalized
    | generalization
    | generalizations

\knowledge{notion}
    | least general generalization
    | least general generalizations
    | anti-unifier
    | anti-unifiers
    | anti-unification

\knowledge{notion}
    | unifier
    | unifiers
\knowledge{notion}
    | parameter
    | parameters
\knowledge{notion}
    | most general
    | most general unifier
    | most general unifiers
    | MGU
    | MGUs
\knowledge{notion}
    | unified
    | unifying
    | unifies
\knowledge{notion}
    | maximally deconstructed

\knowledge{notion}
    | monoidal transducer
    | monoidal transducers
    | Monoidal transducers

\knowledge{notion}
    | cost register automaton
    | cost register automata
    | Cost register automata
    | CRA

%% file: frontmatter.tex
\LIPICSorARXIV{
    \title{Minimization of Streaming Transducers}
    
    \author{Christian Bianchini}{Università degli Studi di Udine, Italy}{bianchini.christian@spes.uniud.it}{https://orcid.org/0009-0006-2848-2517}{}
    \author{Gabriele Puppis}{Università degli Studi di Udine, Italy}{gabriele.puppis@uniud.it}{https://orcid.org/0000-0001-9831-3264}{}
    
    \authorrunning{G. Puppis and C. Bianchini}

    \Copyright{Gabriele Puppis and Christian Bianchini}
    \ccsdesc[500]{Theory of computation~Transducers}
    \keywords{Streaming transducers, Minimization}

    \funding{This work was partially supported by INdAM - GNCS Project (CUP: E53C25002010001).}

    \EventEditors{Claudia Faggian and Joost-Pieter Katoen}
    \EventNoEds{2}
    \EventLongTitle{41st Annual Symposium on Logic in Computer Science (LICS 2026)}
    \EventShortTitle{LICS 2026}
    \EventAcronym{LICS}
    \EventYear{2026}
    \EventDate{July 20--23, 2026}
    \EventLocation{Lisbon, Portugal}
    \EventLogo{}
    \SeriesVolume{380}
    \ArticleNo{27}

    \relatedversiondetails{Full Version}{http://arxiv.org/abs/2605.11190}

    \nolinenumbers
}{
    \title{{\bf Minimization of streaming transducers}%
           \thanks{This work was partially supported by INdAM - GNCS Project (CUP: E53C25002010001).}
           \\[0.5em]
           \large\bf Extended Abstract}
    
    \author{
      Christian Bianchini\\
      Udine University\\
      \texttt{bianchini.christian@spes.uniud.it}
      \and
      Gabriele Puppis\\
      Udine University\\
      \texttt{gabriele.puppis@uniud.it}
    }
    
    \date{}
}

%% file: abstract.tex

\begin{abstract}
We provide general criteria for the existence of "minimal" models of
"streaming transducers", namely devices that read an input word and produce an
"output value" by iteratively updating an internal memory. This abstract model
subsumes classical (sub)"sequential transducers" (Sch{\"u}tzenberger),
"streaming string-to-string transducers" (Alur--\v{C}ern{\'y}), 
"polynomial automata" (Benedikt et al.), and variants of 
"streaming string-to-tree transducers" (Alur--D'Antoni). 
We then instantiate these criteria to obtain effective "minimization" results 
for variants of the latter model, where outputs are "terms" constructed 
incrementally by extending (tuples of) "terms" either at the leaves or at the roots.
\end{abstract}

%% file: intro.tex

We study the existence of "minimal" models of "streaming transducers".
By "streaming transducer" we generically mean a device that consumes an
input word and deterministically produces a corresponding "output" of a 
certain type by updating data in its internal memory. 
Examples of such models are 
the "sequential transducers" \cite{Schutzenberger77},
the "streaming string-to-string transducers" \cite{AlurCerny10}, 
the "polynomial automata" \cite{BenediktDSW17}, and
the "streaming string-to-tree transducers" \cite{AlurDAntoni17,BoiretPS18fsttcs}.

A classical notion of memory that is often used, especially in the previously
mentioned models, consists of a finite number of registers that can be assigned 
numbers, strings, terms, etc. 
Registers can be manipulated using specific operations, that we call 
"updates" (e.g.~sums and products of numbers, concatenations of strings, 
substitutions of terms).
An important aspect is that 
\emph{the memory of our "transducer" model is ``write-only''}: it serves 
exclusively as a mechanism for generating "outputs" and is separated from the
machine's control flow.
\ARXIV{
For example, the registers of
finite-memory automata \cite{KaminskiFrancez94} do not fit this memory model.
}

We will assume that the set of "updates" available for manipulating the internal memory
of a "transducer" is closed under "composition"; in particular, "updates" will 
naturally form a monoid structure.
\AP
In this sense, our model of "transducer" is also similar to 
the ""monoidal transducer"" of \cite{Aristote25}. 
The main difference is that in a "monoidal transducer" the data domain 
itself is required to be a monoid, and the outputs are produced by
aggregating, via the monoid product, the values emitted by the transitions, 
whereas in our model only the "updates" carry a monoid structure. 
Another minor difference is that our data is typed.
More generally, one can see "monoidal transducers" as a syntactic
fragment of "streaming transducers".

The main goal of the paper is to establish general criteria for the
existence of "minimal" "streaming transducers", in the same spirit 
of the celebrated Myhill-Nerode theorem \cite{UllmanHopcroft79},
and then apply these criteria to obtain new "minimization" results
for concrete examples of "transducer" classes.
Drawing inspiration from a categorical approach to automata minimization \cite{Gougen72}, 
and from more recent work that views automata as functors \cite{ColcombetPetrisan,ColcombetPetrisanStabile21}, 
we adopt an \emph{algebraic notion of minimality} for "transducers". 
Informally, a "transducer" $A$ is "algebraically minimal" if, for every 
"equivalent" "transducer" $B$, $A$ can be obtained as a "quotient" of a "subobject" of $B$.
The ``quotient-of-a-subobject'' relation ("subquotient" for short) is formalized 
in terms of suitable pairs of functions that transform the "configuration space" of a
"transducer" and that generalize surjective functions (for inducing "quotients") 
and injective functions (for inducing "subobjects"). In "category" theory,
such functions are called "epis" and "monos", and we
shall adopt the same terminology here.
However, we will keep category-theoretic jargon to a minimum and use it only 
when it provides a clear and convenient shorthand.

The mere existence of "algebraically minimal" "transducers" 
is quite significant: a "minimal" "transducer", being a "subquotient" 
of all "equivalent" "transducers", inherits from those objects any 
property that is closed under taking "quotients" and "subobjects"
(cf.~Eilenberg's theory of varieties/pseudovarieties \cite{Eilenberg76,EilenbergSchutzenberger76}).
For example, it often happens that classes of "transductions" 
definable by logical formalisms (notably, the class of first-order definable transductions) 
are characterized in terms of forbidden patterns at the structural level of "transducers"
(e.g., aperiodicity of the transformation semigroup). 
In these cases, the properties that characterize definability in 
logical formalisms are naturally closed under "quotients" and "subobjects". 
In particular, if some equivalent "transducer" avoids a forbidden pattern 
(e.g.~a non-trivial permutation of the "configuration space"), then so does a "minimal" 
"transducer"; conversely, any obstruction present in a "minimal" "transducer" 
also appears in every equivalent "transducer". 
When "algebraically minimal" "transducers" can be computed, these characterizations 
become effective: membership in a "subquotient"-closed fragment can 
be decided by inspecting the structure of a "minimal" "transducer".
We refer the reader to 
\cite{Schutzenberger65,FiliotKrishnaTrivedi14,ColcombetLeyPuppis15,
      FiliotGauwinLhote16,SaarikiviVeanes17,CartonColcombetPuppis18,MuschollPuppisSTACS19,ColcombetVanGoolMorvan22}
for examples of applications of these principles in the theory of automata and "transducers".

Another important application of "minimization" is automata learning: in active 
learning settings, the availability of canonical/minimal representatives is used 
to maintain compact hypotheses and justify correctness of the learned model 
up to "equivalence". 
See, for instance, 
\cite{Angluin87,RivestSchapire93,Vilar96,ShahbazGroz09,DeLaHiguera10,
      boj14icalp,Laurence2014,CasselHowarJonssonSteffen16,
      MoermanSammartinoSilvaKlinSzynwelski17,ColcombetPetrisanStabile21,Aristote25}.

\medskip
The results in the first part of the paper give sufficient and necessary conditions 
for the existence of "algebraically minimal" "transducers", and can be summarized as follows:

\begin{theorem}[name={\textbf{rephrasing of Theorems \ref{thm:minimal-transducer}--\ref{thm:gcd-necessary}}}]%
               \label{thm:minimal-transducer-rephrasing}
A class of "transducers" admits "algebraically minimal" models if and only if%
    \,\footnote{The only-if direction actually relies on mild assumptions on the sets of
                data and "updates" available to the "transducers" (cf.~the definition of
                "standard data structure" preceding Theorem~\ref{thm:gcd-necessary}). These
                assumptions are satisfied by essentially all "data structures" underlying the
                models considered in the literature.}
the underlying "data structure" admits "constrained domains" and 
"greatest common divisors" for every non-empty set of "updates".
In addition, if every "constrained domain" is also "finitary", then 
"algebraically minimal" models also exist within the subclass of "finite transducers".
\end{theorem}

In the second part of the paper, we give examples of applications of the above theorem 
for "minimizing" "streaming transducers" that produce "terms" as outputs.
More precisely, the outputs of these "transducers" are constructed incrementally
while reading the input, by extending (tuples of) "terms" 
either at the leaves (called "downward string-to-term transducers") 
or at the roots (called "upward string-to-term transducers"). 
For both models, "minimization" is effective, meaning that, given a 
"finite transducer", one can compute a "minimal" one "equivalent" to it.
To the best of our knowledge, these minimization results are new, and can be seen as 
non-trivial generalizations of Choffrut's "minimization" of "sequential transducers" 
\cite{cho03}, one of the few core "minimization" theorems for "transducers". 
Even if some proofs are technically difficult, the most interesting part of 
this contribution lies in a connection between "minimization" of 
"string-to-term transducers", the theory of "unifiers" 
of Robinson, Martelli and Montanari \cite{Robinson65,MartelliMontanari82},
and the theory of "anti-unifiers" of Plotkin and Reynolds \cite{Plotkin70,Reynolds70}.
Specifically, the former theory is used to compute a "subobject"
of an "upward string-to-term transducer", while the latter theory is used
to compute a "quotient" of a "downward string-to-term transducer".

\myparagraph{Related work}
There are several ways to define and construct a minimal "transducer". The two
most prominent approaches in the literature either minimize concrete resources,
such as the number of "control states" and/or registers
(cf.~\cite{AlurRaghothaman2013Decision,
           DaviaudReynierTalbot2016Twinning,BaschenisGauwinMuschollPuppis2016,
           BenaliouaIdrissLhoteReynier2024,BenaliouaLhoteReynier26}),
or even the amount of non-determinism
(cf.~\cite{BellSmertnig2021Polya,BellSmertnig2023LinearHull}), or adopt an
algebraic notion of "minimality", such as the one discussed above, expressed
through a universal "subquotient" property relative to all "equivalent" objects
(cf.~\cite{Gougen72,Gerdjikov18,ColcombetPetrisan,ColcombetPetrisanStabile21,Aristote25}).
Unlike resource minimization, whose goal is to optimize a chosen concrete
measure of an implementation, "algebraic minimality" aims at identifying a
canonical representative of an equivalence class. 

In some settings, however, the two approaches coincide: an object that is
"minimal" in the algebraic sense is also one that uses the fewest concrete
resources. The paradigmatic example is finite-state automata minimization, 
but similar phenomena also occur for automata enhanced with registers
(cf.~\cite{BenediktLeyPuppis10,BojanczykKlinLasota14,MoermanSammartinoSilvaKlinSzynwelski17,BalachanderFiliotGentiliniTzevelekos25}).
In these cases, the automaton with the fewest states is unique and can be
obtained by first pruning unreachable states, thereby passing to a "subobject",
and then merging states according to the coarsest equivalence compatible with
the recognized language. The resulting automaton is therefore a "quotient" of
every other equivalent pruned automaton, and satisfies precisely the universal
"subquotient" property.

For "transducers", the relationship between "algebraic minimality" and resource
minimization is more subtle. On the one hand, "algebraically minimal" 
"transducers" do have a minimum number of "control states" among all 
"equivalent" "transducers" (this follows from standard arguments). 
On the other hand, they are not necessarily optimal with respect to the number 
of registers used. 
In fact, it is even difficult to define a uniform notion of register that 
applies to all generic models of "transducers". When such a notion exists, the
best one can generally expect is to minimize locally, subject to the already
minimized control structure, the number of registers available at each
"control state" that are not fixed by a single assignment. 
We will see in Section~\ref{subsec:upwardSTT} an example of this phenomenon
for the class of "upward string-to-term transducers".

The "algebraic minimization" approach is also closely related to the classical
linear-algebraic theory of weighted automata, going back at least to the
realization results of Carlyle and Paz~\cite{CarlylePaz1971} and to the
minimization theorem of Fliess~\cite{Fliess1974}. Indeed, a weighted automaton
can be transformed into an equivalent "cost register automaton" ("CRA") with a
single "control state" and linear register "updates", where the states of the
weighted automaton become the registers of the one-state "CRA". Thus,
minimizing states in the former amounts to minimizing registers in the latter.
From our perspective, Fliess' minimization can be seen as a concrete instance,
in a linear-algebraic setting, of the same ``subobject followed by quotient''
approach underlying "algebraic minimization": one first restricts to the
smallest subspace generated by reachable values, corresponding to a "subobject",
and then quotients by the equivalence induced by the remaining linear
observations.

\myparagraph{Organization of the paper}
Section \ref{sec:preliminaries} fixes basic notation for functions
and "terms". Section \ref{sec:category} introduces the "category"-theoretic 
concepts underlying "algebraic minimality" and the existence of "minimal" objects. 
Section \ref{sec:model} defines the abstract model of "streaming transducer", 
while Section \ref{sec:minimization} proves "minimization" results under
simple assumptions. Section \ref{sec:applications} gives effective "minimization" 
results for two concrete classes of "transducers", namely 
"downward string-to-term transducers" and "upward string-to-term transducers". 
Finally, Section \ref{sec:conclusion} concludes and discusses a few research 
directions.

\LIPICSorARXIV{
Due to space constraints, most proofs are only sketched or omitted; full details
are available in the extended version at \url{http://arxiv.org/abs/2605.11190}.
}{
This manuscript is the extended version of a paper presented at the conference LICS'26.   
}
For convenience, technical terms and notations in the electronic version of 
this manuscript are hyper-linked to their definitions (cf.~\url{https://ctan.org/pkg/knowledge}).

\LIPICS{
}

%% file: functions.tex

\myparagraph{Functions}

\AP
Given a function $f$, we denote by
$""\Dom*(f)""$, $""\Cod*(f)""$, and $""\Rng*(f)""$,
respectively, its 
\reintro[domain of definition]{domain}, \reintro{codomain}, 
and \reintro{range}%
\footnote{We prefer to not use the term `image' since this is reserved 
          for another concept, to be introduced later.}.
In particular, we have $\Rng(f)\subseteq\Cod(f)$ and 
this inclusion might be strict.
\emph{We assume that the "domain" $A$ and "codomain" $B$ 
are always part of the definition of a function, and they
are usually specified within the notation $f: A\to B$.}
\AP
We use the same notation for a partial function 
$f: A \to B$, with a bit of care for its "domain@@function": 
we distinguish the ""source set"" $A$ from the 
\reintro{domain of definition} of $f$, defined 
as the subset $\Dom(f)$ of $A$ that contains the 
elements $x$ where $f(x)$ is defined.
When equating two partial functions, like in $f=g$,
we require that (1) $f$ and $g$ have the same
"source", "domain of definition", and "codomain",
and (2) for all $x\in\Dom(f)$, $f(x)=g(x)$.

\AP
Given two (partial) functions $f: A \to B$ and $g: B \to C$, 
we denote by $""f \comp* g @composition""$ their \reintro{composition}, 
interpreted according to the \emph{diagrammatic order}%
\footnote{We are aware that functional composition 
	is often written with $\circ$ and in the opposite order;
	we believe that the adopted notation eases readability, 
	especially here since we often deal with long sequences of compositions.},
namely, $f\comp g: A \to C$ is defined on those
$x$'s such that $x\in\Dom(f)$ and $f(x)\in\Dom(g)$, and maps
any such $x$ to $g(f(x))\in C$.
When we write a "composition" $f\comp g$ we always tacitly 
assume that the "codomain" of $f$ is the same as the "source" of $g$; 
however, $\Rng(f)$ may be strictly contained in the "source" of $g$, and
$\Dom(g)$ may be strictly contained in $\Cod(f)$.
\AP
A ""factorization"" of a (partial) function $h: A \to B$ is
a pair of (partial) functions $f: A\to C$ and $g: C\to B$ such that $h=f\comp g$.

%% file: terms.tex

\myparagraph{Terms}

\AP
A ""term"" is a non-empty, finite tree, 
labelled over a ranked alphabet $\Gamma$, where the number 
of children of each node coincides with the rank of its label. 
An example is the "term" $a(b(c),c)$, provided that 
$a,b,c\in\Gamma$ have ranks $2,1,0$, respectively.
\AP
We shall also work with "terms" over an alphabet expanded with variables,
say $\Gamma\uplus X$, with $X = \{x_1,x_2,\dots,y,\dots\}$, where the variables
are assumed to have rank $0$, and thus necessarily occur at the leaves. 

\AP
We mostly manipulate "terms" using (\reintro[substitution]{first-order}) \reintro{substitutions}:
given a "term" $t$ over an $n$-tuple of variables $\bar x=(x_1,\dots,x_n)$
and given an $n$-tuple of "terms" $\bar u=(u_1,\dots,u_n)$ 
(possibly over a different tuple of variables), 
we denote by $""t\subst*[\bar x]{\bar u}""$ the result of substituting in 
$t$ every occurrence of variable $x_i$, simultaneously for all $i=1,\dots,n$,
with the corresponding "term" $u_i$.
For example, if $t=a(x_1,x_2,x_1)$, $\bar u= (b, c(x_3))$, 
then  $t\subst[\bar x]{\bar u} = a(b,c(x_3),b)$.

%% file: epis-monos.tex

%

We focus on an algebraic notion of "minimality" for "transducers",
characterized by a universal property formalized in terms of "quotient" and
"subobject" relations. Although these relations can be tailored to specific
models of "transducers", such definitions tend to obscure the key properties
behind model-dependent details. Since we aim to establish "minimization" results
for several classes of "transducers", we work at a higher level of abstraction,
where "quotients" and "subobjects" can be defined once, uniformly and
independently of the model.
For this, we use a simplified category-theoretic framework. 
\AP
We call ""arrow"" any member of a restricted class of functions or partial functions, 
assumed to be closed under "composition" and to contain an ""identity"" on each set. 
\AP
The sets forming the "domains" and "codomains" of these (partial) functions may
carry additional structure, and are generically called ""objects""; accordingly,
an "arrow" can be seen as a structure-preserving (partial) function.
Examples of "arrows" are the "transformations" within a "transducer"'s
"configuration space", representing transition functions, and the
"transformations" between "transducers" themselves, representing "transducer
morphisms". 
\AP
A ""category"" is precisely a class of "objects" connected by "arrows".

\emph{All definitions in this section are understood relative to a fixed "category",
which we keep implicit whenever no ambiguity arises.}

\LIPICSorARXIV{
\myparagraph{Epis and monos}
}{
\subsection{Epis and monos}\label{secub:epis-monos}
}
We begin by discussing special types of "arrows", called "epis" and "monos", 
which capture "quotients" and "subobjects", respectively, and generalize surjective
and injective maps between sets. Such generalizations are common in "category"
theory and come in several strengths. For instance, plain "epis" lie at the
bottom of the standard hierarchy
\[
	\text{"epis"} ~\supseteq~ \text{extremal epis} ~\supseteq~
	\text{strong epis} ~\supseteq~ \text{strict epis} ~\supseteq~
	\text{regular epis}.
\]
In this work, we mostly use "epis" and "strong monos", a choice tailored to our
"minimization" results. In the "category" of unrestricted functions between sets,
the hierarchies of "epis" and "monos" collapse to the usual surjective and
injective functions. In more structured "categories", where "arrows" are
restricted functions ---as happens for the transition functions of our
"transducers"--- these notions may diverge. For example, in the "category" 
of "polynomial maps" over algebraic varieties, the function
$f: \bbR^2 \to \bbR^2$ defined by $f(x,y)=(x,xy)$ is "epi" but not surjective,
since $(0,1)\nin\Rng(f)$.

Below are the definitions of the important types of "arrows"
that we are going to use:
\begin{itemize}
\item \AP
	An "arrow" $f: A \to B$ is an ""epi"", 
	denoted $A \mathrel{\reintro{\epi*[f]}} B$,
	if for all "arrows" $h,h': B\to C$,
	$f\comp h = f\comp h'$ implies $h=h'$.
\item \AP
	An "arrow" $g: B\to C$ is a ""mono"", 
	denoted $B \mathrel{\reintro{\mono*[g]}} C$,
	if for all "arrows" $h,h': A\to B$,
	$h\comp g = h'\comp g$ implies $h=h'$.
\item \AP
	An "arrow" $i: A\to B$ is an ""isomorphism"",
	denoted $A \mathrel{\reintro{\iso*[i]}} B$,
	if there is an arrow $\inverse*{i}: B\to A$,
	called ""inverse"" of $i$,
	such that $i\comp i^{-1}$ and $i^{-1}\comp i$ 
	are "identities" on $A$ and $B$, respectively. 
	\AP
	We write $A \mathrel{\reintro{\iso*}} B$,
	without the annotation $i$,
	to state the existence of such an "isomorphism" between $A$ and $B$.
	\AP
	In a similar way, by thinking of "arrows" as "objects" 
	in a new "category" (the arrow "category"), we define an
	""isomorphism@@arrow"" between two "arrows" 
	$f: A\to B$ and $f': A'\to B'$ as a pair of 
	"object isomorphisms" $i: A\to A'$ and $j: B\to B'$
	such that $f\comp j = i\comp f'$.
	\AP
	We denote this by $f \mathrel{\reintro{\aiso*[i,j]}} f'$,
	or simply by $f \mathrel{\reintro{\aiso*}} f'$.
\item \AP
	An "arrow" $g: C\to D$ is a ""strong mono"", 
	denoted $C \mathrel{\reintro{\smono*[g]}} D$,
	if for all "epi" $f: A \epi B$ and "arrows"
	$h: A\to C$ and $h': B \to D$, 
	if $f\comp h' = h\comp g$, 
	then there is a unique "arrow"
	$d: B\to C$ such that $h = f\comp d$ and $h' = d\comp g$.
	In other words, if the left diagram commutes, 
	then also the right diagram commutes for some unique $d$:
	\begin{align*}
	\input{fig-diamond-fillin}
	\end{align*}
\AP
The above condition is called ""diagonal fill-in property"", and is
crucial for reasoning with diagrams involving "epis" and "strong monos". Two
standard consequences will be used repeatedly. First, every "strong mono" is a
"mono": indeed, by taking $f$ to be an identity and by considering 
two arrows $h,h'': A\to C$ such that $h\comp g = h''\comp g'$, we get
$h = d = h'$. Second, every "arrow" that is both "epi" and "strong mono" 
is an "object isomorphism", obtained by applying the same "diagonal fill-in property" 
to the square where $h$ and $h'$ are identities.
\end{itemize}
It is easy to see that the classes of "epis", "monos", 
"isomorphisms", and "strong monos" are all closed under "composition".

%% file: fig-diamond-fillin.tex
\begin{tikzpicture}[baseline=(current bounding box.center)]
  \node (A) at (0,0) {$A$};
  \node (B) at (20mm,0) {$B$};
  \node (C) at (0,-14mm) {$C$};
  \node (D) at (20mm,-14mm) {$D$};

  \path
    (A) edge [epi arrow]
      node [above=1mm] {\scriptsize $f$}
    (B)
    (C) edge [reg mono arrow]
      node [below=1mm] {\scriptsize $g$}
    (D)
    (A) edge [arrow]
      node [left=1mm,pos=0.4] {\scriptsize $h$}
    (C)
    (B) edge [arrow]
      node [right=1mm,pos=0.4] {\scriptsize $h'$}
    (D);
\end{tikzpicture}
\qquad\quad
\text{implies}
\qquad\quad
\begin{tikzpicture}[baseline=(current bounding box.center)]
  \node (A) at (0,0) {$A$};
  \node (B) at (20mm,0) {$B$};
  \node (C) at (0,-14mm) {$C$};
  \node (D) at (20mm,-14mm) {$D$};

  \path
    (A) edge [epi arrow]
      node [above=1mm] {\scriptsize $f$}
    (B)
    (C) edge [reg mono arrow]
      node [below=1mm] {\scriptsize $g$}
    (D)
    (A) edge [arrow]
      node [left=1mm,pos=0.4] {\scriptsize $h$}
    (C)
    (B) edge [arrow]
      node [right=1mm,pos=0.4] {\scriptsize $h'$}
    (D)
    (B) edge [arrow, dashed]
      node [above left=0mm] {\scriptsize $\exists! d$}
    (C);
\end{tikzpicture}

%% file: factorization-systems.tex

\LIPICSorARXIV{
\myparagraph{Factorization systems}
}{
\subsection{Factorization systems}\label{subsec:factorization-systems}
}
\AP
A "factorization" $h = f\comp g$ is called ""strong@@factorization"" 
if $f$ is an "epi" and $g$ is a "strong mono", and in this case the 
intermediate "object" $C=\Cod(f)$ is called the ""image"" of $h$. 
The uniqueness of the "image" of $h$ follows directly from the 
"diagonal fill-in property":

\begin{lemma}[restate=StrongImage,
						  name=\textbf{{\cite[Propositions 14.4]{TheJoyOfCats}}}]\label{lem:strong-image}
The "image" of an "arrow", if a "strong factorization" exists,
is unique up to "isomorphism".
\end{lemma}

\AP
An "object" $I$ is ""initial"" if for every "object" $A$,  
there is a unique "arrow" $""\initialmorph*{A}"": I\to A$, 
called the \reintro{initial arrow} for $A$.
Symmetrically, an "object" $F$ is ""final"" if for every "object" $A$,  
there is a unique "arrow" $""\finalmorph*{A}"": A\to F$, 
called the \reintro{final arrow} for $A$.

\begin{lemma}[restate=InitialFinal, name=\textbf{{\cite[Propositions 7.3, 7.6]{TheJoyOfCats}}}]%
	\label{lem:initial-final}
"Initial" (resp.~"final") "objects", 
if exist in the considered "category", 
are unique up to "isomorphism".
\end{lemma}

\medskip
Below, we define "algebraic minimality" and we outline some general conditions 
for the existence of "algebraically minimal" "objects", based on a reworking of \cite{ColcombetPetrisan}.
\begin{itemize}
\item \AP
	$A$ is a ""quotient"" of $B$ if $B \epi A$, i.e.~there is an "epi" from $B$ to $A$,%
\item \AP
	$A$ is a ""subobject"" of $B$ if $A \smono B$, i.e.~there is a "strong mono" from $A$ to $B$
\item \AP
	$A$ is a ""subquotient"" of $B$ 
	if $A \mathrel{\scalebox{-1}[1]{$\epi$}} C \smono B$ for some $C$,
	i.e.~$A$ is a "quotient" of some "subobject" $C$ of $B$,
\item \AP
	$A$ is ""algebraically minimal"" if it is a "subquotient" of every "object".
\end{itemize}

\noindent
Next, assume that the "category" under consideration admits: 
\begin{enumerate}
\item a "strong factorization" of every "arrow", together with the corresponding "image",
\item an "initial object" $I$,
\item \AP
	an "object" $M$ that is "final" in the \emph{sub-category} consisting
	of the "images" of the "initial arrows" (""initial images"" for short).
\end{enumerate}
Recall that, by Lemmas \ref{lem:strong-image}
and \ref{lem:initial-final}, the above "objects"
are unique up to "isomorphism".
We also remark that the above assumptions differ slightly 
w.r.t.~the presentation of \cite{ColcombetPetrisan}, specifically
because we do not require the existence of a "final object"
in the full category, but rather in the sub-category of
"initial images". The reason for this difference is
that it simplifies the construction of the "minimal" 
object, especially for the case of "transducers".

For every "object" $A$, we can then define:
\begin{itemize}
\item \AP
	$""\Reach*(A)""$ as the "image" 
	of the "initial arrow" $I \to A$,
\item \AP 
	$""\Obs*(\Reach*(A))""$ as the "image" 
	of the "final arrow" $\Reach(A) \to M$
	in the sub-category of "initial images".
\end{itemize}
Note that the "initial object" $I$ is "isomorphic" to its "image"
$\Reach(I)$, since, by definition, there is a unique "arrow" from $I$ to itself, 
and this is both the identity and the "composition" $I \epi \Reach(I) \smono I$.

The notations $\Reach(A)$ and $\Obs(\Reach(A))$ are inspired by 
the operations performed to minimize deterministic finite automata. 
Specifically, $\Reach(A)$ reflects the idea that constructing a 
"strong factorization" of the "initial arrow" $I \to A$ 
amounts to pruning $A$ by retaining only its reachable states. 
Likewise, $\Obs(\Reach(A))$ is suggestive of the fact that the 
"final arrow" $\Reach(A) \to M$ induces an observational 
equivalence that groups reachable states inducing the
same residual languages. 

\begin{proposition}[restate=AlgebraicallyMinimal,
	name=\textbf{{variation of \cite[{Lemma 3.5}]{ColcombetPetrisan}}}]%
	\label{prop:minimal}
Under the previous assumptions that enable the constructions 
$\Reach(-)$ and $\Obs(\Reach(-))$, we have that 
\begin{enumerate}
\item $I \iso \Reach(I)$,
\item $M \iso \Obs(I)$,
\item $M \iso \Obs(\Reach(A))$ for every "object" $A$.
\end{enumerate}
In particular, $\Obs(\Reach(A))$ is a "subquotient" of $A$,
and because $M \iso \Obs(\Reach(A))$, 
$M$ is a "subquotient" of every $A$, hence an "algebraically minimal" "object".
\end{proposition}

\begin{proof}
Although this result is a variation of \cite[{Lemma 3.5}]{ColcombetPetrisan}, we include a proof to keep the presentation self-contained.
The first property $I \iso \Reach(I)$ follows from considering
a "strong factorization" of the "initial arrow" 
$\initialmorph{I}: I \to I$, which is also an identity "morphism".
This gives the "initial image" $\Reach(I)$ together with an "epi" 
$f: I \epi \Reach(I)$ and a "strong mono" $g: \Reach(I) \smono I$. 
One can then apply twice the "diagonal fill-in property" 
to the emerging diagram:
\begin{align*}
\input{fig-initial-is-image}
\end{align*}

Similarly, the second property $M \iso \Obs(I)$ follows from considering a
"strong factorization" of the unique "morphism" $h$ from $I$ to $M$, together
with the induced "image".
Since the latter "morphism" $h$ can be interpreted both as the "initial arrow" 
for $M$ and as the "final arrow" for $I$ ($\iso \Reach(I))$, we obtain that the 
"images" $\Reach(M)$ and $\Obs(I)$ are "isomorphic". Moreover, we also have
$M \iso \Reach(M)$, simply because $M$ is a "final object" in the sub-category of "initial images",
and hence $M \iso \Reach(M) \iso \Obs(I)$. 
These arguments can be equally seen as consequences of the "diagonal fill-in property" 
applied to the following diagram:
\begin{align*}
\input{fig-minimal-is-initial-quotient}
\end{align*}
Finally, for an arbitrary "object" $A$, the property $M \iso \Obs(\Reach*(A))$
follows from considering, first, a "strong factorization" $I \epi[f] \Reach(A) \smono[g] A$
for the unique "initial arrow" $\initialmorph{A}: I \to A$, and then another
"strong factorization" $\Reach(A) \epi[f'] \Obs(\Reach*(A)) \smono[g'] M$ 
for the unique "final arrow" $\finalmorph{\Reach*(A)}: \Reach(A) \to M$. 
By recalling that there is also an "epi" $h: I \iso \Reach(I) \epi \Obs(\Reach*(I)) \iso M$,
we get the diagram below, where the "diagonal fill-in property" implies $M \iso \Obs(\Reach*(A))$:
\begin{align*}
\input{fig-minimal-is-unique}
\end{align*}
\qedhere
\end{proof}

We conclude by observing that, in general, unless the "category" is particularly 
well-behaved, the "subquotient" relation does not need to be transitive or 
anti-symmetric. In particular, there may exist non-"isomorphic", "minimal" 
"objects" that are one a "subquotient" of another. However, if the
objects are of the form $\Obs(\Reach(A))$ and $\Obs(\Reach(A'))$,
then they are necessarily also "isomorphic".

%% file: fig-initial-is-image.tex
\begin{tikzpicture}[baseline=(current bounding box.center)]
\matrix (M) [matrix of math nodes, column sep=20mm, row sep=18mm, nodes={inner sep=1mm}] {
  |[name=I, anchor=center]| I & 
  |[name=ReachI, anchor=west]| \Reach(I) 
  \\
  |[name=Ip, anchor=center]| I & 
  |[name=Ipp, anchor=west]| I 
  \\
};
\path (I) edge [epi arrow] node [above=1mm,pos=0.4] {$f$} (ReachI);
\path (I) edge [epi arrow] node [left=2mm,pos=0.4] {$\initialmorph{I}$} (Ip);
\path (Ip) edge [reg mono arrow] node [below=1mm] {$\initialmorph{I}$} (Ipp);
\path ([xshift=2mm]ReachI.south west) edge [reg mono arrow] node [right=2mm] {$g$} ([xshift=2mm]Ipp.north west);
\path ([xshift=-1.3mm]ReachI.south west) edge [arrow,dashed] node [above left=1.3mm] {$\iso*$} ([xshift=-1mm]Ip.north east);
\path ([xshift=1.3mm]Ip.north east) edge [arrow,dashed] ([xshift=1.3mm]ReachI.south west);
\end{tikzpicture}

%% file: fig-minimal-is-initial-quotient.tex
\begin{tikzpicture}[baseline=(current bounding box.center)]
\matrix (Mat) [matrix of math nodes, column sep=20mm, row sep=18mm, nodes={inner sep=1mm}] {
  |[name=I, anchor=west]| I \iso \Reach(I) & 
  |[name=ObsReachI, anchor=west]| \Obs(\Reach*(I)) \iso \Obs(I) 
  \\
  |[name=ReachM, anchor=west]| \Reach(M) \iso M & 
  |[name=M, anchor=west]| M 
  \\
};
\path (I.east) edge [epi arrow] (ObsReachI.west);
\path ([xshift=2mm]I.south west) edge [epi arrow] ([xshift=2mm]ReachM.north west);
\path (ReachM.east) edge [reg mono arrow] (M.west);
\path ([xshift=3mm]ObsReachI.south west) edge [reg mono arrow] ([xshift=3mm]M.north west);
\path ([xshift=-1mm]ObsReachI.south west) edge [arrow,dashed] node [above left=1mm] {$\iso*$} ([xshift=-1mm]ReachM.north east);
\path ([xshift=1mm]ReachM.north east) edge [arrow,dashed] ([xshift=1mm]ObsReachI.south west);
\end{tikzpicture}

%% file: fig-minimal-is-unique.tex
\begin{tikzpicture}[baseline=(current bounding box.center)]
\matrix (Mat) [matrix of math nodes, column sep=20mm, row sep=18mm, nodes={inner sep=1mm}] {
  & 
  |[name=A, anchor=center]| A & 
  \\
  |[name=I, anchor=west]| I & 
  |[name=Reach, anchor=center]| \Reach(A) & 
  |[name=ObsReach, anchor=east]| \Obs(\Reach*(A)) 
  \\
  |[name=M, anchor=west]| M & 
  & 
  |[name=Mp, anchor=east]| M 
  \\
};
\draw [fill=gray, opacity=0.1, rounded corners] ([xshift=7mm,yshift=7mm]ObsReach.north east) rectangle ([xshift=-7mm,yshift=-7mm]M.south west);
\draw ([xshift=1mm,yshift=3mm]ObsReach.north east) node [above=2mm, gray, anchor=west] {sub-category of};
\draw ([xshift=1mm,yshift=3mm]ObsReach.north east) node [below=2mm, gray, anchor=west] {initial images};
\path ([xshift=2mm]I.south west) edge [epi arrow] node [left=2mm,pos=0.4] {$h$} ([xshift=2mm]M.north west);
\path (M) edge [reg mono arrow] node [below=2mm] {$\identity{M}$} (Mp);
\path (I) edge [epi arrow] node [above=2mm,pos=0.4] {$f$} (Reach);
\path (Reach) edge [reg mono arrow] node [right=2mm,pos=0.6] {$g$} (A);
\path (Reach) edge [epi arrow] node [above=2mm,pos=0.4] {$f'$} (ObsReach);
\path ([xshift=-3mm]ObsReach.south east) edge [reg mono arrow] node [right=1mm] {$g'$} ([xshift=-3mm]Mp.north east);
\path ([xshift=-1mm]ObsReach.south) edge [arrow,dashed] node [above left=1mm] {$\iso*$} ([xshift=-1mm]M.north east);
\path ([xshift=5mm]M.north east) edge [arrow,dashed] ([xshift=5mm]ObsReach.south);
\end{tikzpicture}

%% file: data.tex

Our "transducers" can store data from an arbitrary, fixed universe 
(e.g.~real numbers, words, or trees), and can manipulate it
via a basic set of operations, called "updates" 
(e.g. sums and products of numbers, concatenations of strings, 
"term" "substitutions").

\subsection{Data structures}\label{subsec:data}

\AP
We define a \reintro{data structure} as an algebra $""\bbD*""$ of typed data. 
In most cases, the reader can assume ""types"" to be natural numbers,
although later we will introduce more abstract "types"
--- even with this extension, we shall always assume 
that the set of "types" is countable.
The \reintro{domain} of a "type" $\tau$, denoted $""\bbD*_\tau""$,
is the subset of $\bbD$ consisting of all data of "type" $\tau$.
For example, if $\tau$ is a number, then its "domain" may include
$\tau$-tuples of "terms".
\AP
The "data structure" is also equipped with a set of operations, called ""updates"",
and this set is assumed to be closed under "composition" and to contain "identities"
on each "domain". 
Each "update" is associated with an input "type" $\alpha$ and an "output" "type" $\beta$, 
and accordingly maps any data in $\bbD_\alpha$ to a data in $\bbD_\beta$. In particular,
the "domain of definition" of an "update" coincides with the "domain" of its input type.
\AP
We denote by $\reintro[updates]{\bbD*_\alpha^\beta}$ the set of "updates" 
with input "type" $\alpha$ and "output" "type" $\beta$.

The following examples of "data structures" with numeric "types" will be used later:
\begin{itemize}
\item {\bf ""Free term algebra"".}
      \AP
      The "domain" associated with each  
      "type" $\tau$ contains $\tau$-tuples of ""ground terms"", namely, "terms" without variables.
      \AP
      Note that there is only one data of "type" $\tau=0$, namely, the "empty tuple" $\intro{\emptytuple*}$.
      An "update" in $\bbD_\alpha^\beta$ is specified by a $\beta$-tuple of "terms" 
      $\bar c=(c_1,\dots,c_\beta)$ over an $\alpha$-tuple of variables $\bar x=(x_1,\dots,x_\alpha)$, 
      and maps any $\alpha$-tuple of "ground terms" $\bar t=(t_1,\dots,t_\alpha)$ 
      to the $\beta$-tuple $\big(c_1\subst[\bar x]{\bar t},\dots,c_\beta\subst[\bar x]{\bar t}\big)$.
      Intuitively, such an "update" extends the input "terms" $t_1,\dots,t_\alpha$ upward,
      at their roots.
      One can further restrict the types of "updates" in $\bbD_\alpha^\beta$ by 
      requiring that the specifications $\bar c$ satisfy any combination of the following conditions:
      \LIPICSorARXIV{%
      (1) every variable $x_i$ occurs \emph{at most once} in $\bar c$
            (""copyless updates"", as opposed to ""copyful updates"");
      (2) every variable $x_i$ occurs \emph{at least once} in $\bar c$
            (""non-erasing updates"");
      (3) every variable $x_i$ occurs \emph{exactly once} in $\bar c$
            and when the variables of $\bar c$ are listed from left to right 
            they appear in the precise order $x_1,\dots,x_\alpha$
            (""linear updates"").      
      }{
      \begin{enumerate}
      \item every variable $x_i$ occurs \emph{at most once} in $\bar c$
            (""copyless updates"", as opposed to ""copyful updates"");
      \item every variable $x_i$ occurs \emph{at least once} in $\bar c$
            (""non-erasing updates"");
      \item every variable $x_i$ occurs \emph{exactly once} in $\bar c$
            and when the variables of $\bar c$ are listed from left to right 
            they appear in the precise order $x_1,\dots,x_\alpha$
            (""linear updates"").
      \end{enumerate}
      }
\item {\bf ""Leaf substitution algebra"".}
      The "domain" of each "type" $\tau$ contains "terms" with exactly $\tau$
      variable occurrences, without repetitions and in a fixed left-to-right order
      $x_1,\dots,x_\tau$. 
      \AP
      Since variable names are immaterial in this setting, we
      replace them by placeholders and write the "substitution" simply as
      $""t\subst*{\bar u}""$.
      We call these objects ""linear terms"". 
      An "update" in $\bbD_\alpha^\beta$ is specified by an $\alpha$-tuple 
      $\bar u=(u_1,\dots,u_\alpha)$ of "linear terms"
      with precisely $\beta$ placeholders, and maps every "linear term" $t$ of
      "type" $\alpha$ to $t\subst{\bar u}$ of "type" $\beta$. 
      Intuitively, such an "update" extends $t$ downward, at the leaves.
\item \AP
      {\bf ""Polynomial register algebra"".}
      The "domain" of each "type" $\tau$ consists of $\tau$-tuples of elements from
      a given field, such as $\bbQ$, or more generally from a polynomial ring, such
      as $\bbQ[z]$.
      \AP
      The "updates" are ""polynomial maps"" of the form
      $(x_1,\dots,x_\alpha) \mapsto (p_1(\bar x),\dots,p_\beta(\bar x))$,
      where each $p_i$ is a polynomial in $\bar x=(x_1,\dots,x_\alpha)$. As before,
      one may restrict "updates", for example, to affine or linear maps.
\item \AP 
      {\bf ""String register algebra"".} 
      Here the data consists of tuples of strings, whose "types" are again their arities.
      \AP
      The "updates" are generated by atomic operations that prepend or append a letter
      to a component, concatenate two components, insert a new empty component,
      duplicate, delete, or swap components. As with terms, one can restrict these
      "updates", for instance by forbidding duplication ("copyless updates") and/or
      deletion ("non-erasing updates") \cite{AlurCerny10,MuschollPuppisICALP19}.
\end{itemize}
The "free term algebra" is particularly important for us, because many 
"data structures" can be presented as its quotients. For example, 
strings over $\Gamma$ can be represented by "terms" over the ranked 
alphabet $\Gamma\uplus\{\emptystr,\odot\}$, where symbols in $\Gamma$ 
and $\emptystr$ have rank $0$, while $\odot$ has rank $2$ and represents concatenation. 
To enforce associativity and identity laws for $\emptystr$, one quotients 
"terms" by the finest congruence $\sim$ satisfying 
$\odot(x,\odot(y,z)) \sim \odot(\odot(x,y),z)$ and $\odot(x,\emptystr) \sim x \sim \odot(\emptystr,x)$. 
Accordingly, "updates" on tuples of strings, such as those of the "string register algebra", 
can be simulated by "term" "substitutions": for example, concatenating $w_1$ and $w_2$ 
amounts to applying the "term" "update" $(w_1,w_2) \mapsto t\subst[x_i]{w_i}_{i=1,2}$, 
with $t=\odot(x_1,x_2)$.

%% file: transducers.tex

\subsection{Streaming transducers}\label{subsec:streaming-transducers}

In the following, we fix an arbitrary "data structure" $\bbD$.
While different "transducers" from the same class will share the same "data structure" $\bbD$, 
it is convenient to allow different "types" of data at different "control states" of a "transducer". 
For example, a "transducer" can store tuples of different arities at different "states". 

\AP
We denote by $Q$ a set of ""control states"". 
Each "state" $q\in Q$ is assigned a "data type" $""\type*_Q(q)""$, 
and hence a "domain" $\Data_Q(q)$.
Note that the assignment may depend on $Q$,
and so the same "state" $q$ may be assigned different "types" in
different state sets. When $Q$ is clear from the context, we drop the 
subscript $Q$ from the notations $\type_Q(q)$ and $\Data_Q(q)$.
\AP
The ""configuration space generated by $Q$"" is the set 
\[
  \reintro{\States*{Q}}
  ~=~ \big\{ (q,d) \::\: q\in Q, ~ d\in\Data(q) \big\}.
\]
\AP
A ""configuration"" is a "state"-data pair $(q,d) \in \States{Q}$.

\AP
A ""transformation"" is a partial function
$f: \States{Q} \rightharpoonup \States{Q'}$ between "configuration spaces"
that is \reintro{specified} by a pair $(\speca f,\specb f)$, where:
\begin{itemize}
\item $\reintro{\speca* f}: Q \rightharpoonup Q'$ is a partial function from $Q$ to $Q'$,
			describing the possible transitions between control states;
\item $\reintro{\specb* f}$ is a partial function, 
		  with the same "domain of definition" as $\hat f$ 
      (namely, $\hat f$ is defined on $q$ iff $\check f$ is defined on $q$)
      that maps every "control state" $q\in\Dom(\check f)$ 
			to a corresponding "update" $\check f(q)\in \bbD_{\type(q)}^{\type(r)}$,
			where $r=\hat f(q)$.
\end{itemize}
The transformation $f$ "specified" by $(\speca f,\specb f)$ 
has for "domain" the set 
$\States{\Dom(\speca f)} \subseteq \States{Q}$ 
and maps any "configuration" $(q,d)\in\Dom(f)$ 
to the "configuration" $\big(\speca f(q), \, \specb f(q)(d)\big) \in \States{Q'}$.
Note that the notion of "transformation" between "configuration spaces" 
is a restriction of that of a partial function, in two respects. 
First, the "transformation" 
relativized to a particular "state" must correspond to an "update", 
which might already be a restricted form of function between data.
Second, there might be a dependency of $\specb f$ from $\speca f$, 
but not the other way around.
For example, we can injectively map "configurations" with different 
"states" to "configurations" with the same "state" 
and different data, but not vice versa.

"Transformations" inherit from "updates" their closure under "composition"
(proof omitted):

\begin{lemma}[name=\textbf{composition of transformations},
			  restate=Compositiontransformations]
			  \label{lem:composition-transformations}
Given two "transformations" 
$f: \States{Q} \rightharpoonup \States{Q'}$ and 
$g: \States{Q'} \rightharpoonup \States{Q''}$,
"specified" respectively by 
$(\speca f,\specb f)$ and $(\speca g,\specb g)$, 
their "composition" is the "transformation"
$h: \States{Q} \rightharpoonup \States{Q''}$ 
"specified" by $(\speca h,\specb h)$, where 
$\speca h = \speca f\comp \speca g$ and 
$\specb h(q) = \specb f(q) \,\comp \specb g(\speca f(q))$ 
for all $q\in Q$.
\end{lemma}

\AP
To model the "initial@initial transformation" and "final transformations"
of a "transducer" it is convenient to introduce two special "states",
$""\initstate*""$ and $""\haltstate*""$,
and assume they are fixed as part of the "transducer" model, just as the 
underlying "data structure" $\bbD$ is.
The "state" $\initstate$ is assigned a special "data type" $""\init*""$ 
with a singleton "domain" $\bbD_{\init*} = \{""\initdata*""\}$. 
For example, sometimes we identify $\init$ with the numeric "type" $0$, 
and $\initdata$ with the "empty tuple" $\emptytuple$.
\AP
Symmetrically, the "state" $\haltstate$ is assigned "type" $""\halt*""$ 
(e.g.~the number $1$),
but its "domain" usually contain multiple data, that is, the possible output values.

The "states" $\initstate$ and $\haltstate$
determine, respectively, the \reintro{initial configuration space} 
$""\initspace*"" = \{(\initstate,\initdata)\}$ and 
the \reintro{final configuration space}
$""\haltspace*"" = \{(\haltstate,d) \::\: d\in\bbD_{\halt*}\}$.
Both "spaces" are part of the underlying setup and hence are fixed 
for the entire class of "transducers" under consideration.

\AP
A ""deterministic streaming transducer"" (hereafter simply \reintro{transducer})
over an input alphabet $\Sigma$ is a tuple 
$A=(Q,\deltainit,(\delta_a)_{a\in\Sigma},\deltahalt)$, where:
\begin{itemize}
\item $Q$ is a typed set of "control states", generating $\States{Q}$,
\item \AP
	  $\deltainit: \initspace \rightharpoonup \States{Q}$ 
	  is an ""initial transformation"",
\item \AP 
	  $\delta_a: \States{Q} \rightharpoonup \States{Q}$ 
	  is an ""internal transformation"", for each $a\in\Sigma$,
\item \AP
	  $\deltahalt: \States{Q} \rightharpoonup \haltspace$ 
	  is a ""final transformation"". 
\end{itemize}
Recall that "transformations" are allowed to be partial functions.
Intuitively, $\deltainit(\initstate,\initdata)$, if defined, determines 
the first "configuration" of the run of the "transducer";
$\delta_a$ describes how the "transducer" moves from one 
"configuration" to another when reading an input letter $a$;
and, $\deltahalt$ describes how the final output is 
constructed from the last "configuration" of the run.
\AP
We remark that the "state" set $Q$ of a "transducer" does not need to be finite; 
we will explicitly say ""finite transducer"" when we want $Q$ to be finite.

The general structure of a "transducer" model is shown on the left of
Figure~\ref{fig:transducer-diagram}. For a specific "transducer", however, it
is often useful to annotate "control states" with their "data domains", and
"transitions" with their input symbols and "updates". The right-hand side of
Figure~\ref{fig:transducer-diagram} illustrates this representation with an
example of a "sequential transducer" (the model is defined further below), 
which copies the input to the output while inserting a separator $\#$ before
each pair of consumed input symbols. 
Note that the "data domain" associated with each "control state" is $\Gamma^*$,
where $\Gamma = \Sigma\uplus\{\#\}$, and every "update" appends a constant 
string to the source data.

\begin{figure*}[t]
\begin{align*}
\input{fig-transducer-diagram}
\end{align*}
\caption{Diagrams of a "transducer".}
\label{fig:transducer-diagram}
\end{figure*}
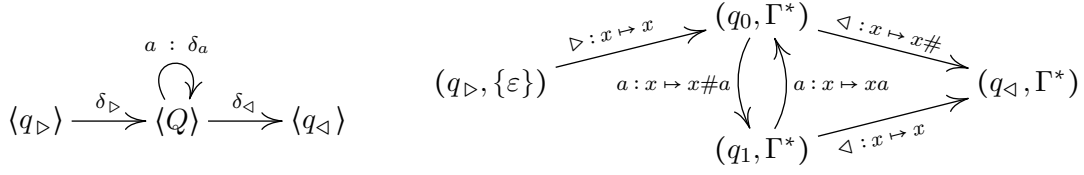

\AP
Given a "transducer" $A=(Q,\deltainit,(\delta_a)_{a\in\Sigma},\deltahalt)$
and some input $w=a_1\dots a_n \in \Sigma^*$,
the ""transformation induced"" by $w$ is the "composition" 
$\delta*_w = \delta_{a_1}\comp\dots\comp\delta_{a_n}$ of the 
"internal transformations" induced by the letters in $w$.
This "composition" can be prolonged to include the
"initial@initial transformation" and/or the "final transformation",
resulting for instance in the "transformation" 
$\delta*_{\init w\halt} = \deltainit\comp\delta*_w\comp\deltahalt:
 \initspace \rightharpoonup \haltspace$. 
\AP
In particular, 
$\delta*_{\init* w\halt*}(\initstate,\initdata) = \deltahalt(\delta*_w(\deltainit(\initstate,\initdata)))$, 
and if defined, results in a "configuration" of the form $(\haltstate,d_w)$, for some $d_w\in\bbD_{\halt*}$.
In this case, $d_w$ is called the ""output"" of $A$ on $w$.
\AP
The ""transduction realized"" by $A$ is the partial function 
$\varphi: \Sigma^* \rightharpoonup \bbD_{\halt*}$ that maps any $w\in\Dom(\varphi)$ 
to the corresponding output $d_w$, as described above.
\AP
We say that two "transducers" are ""equivalent"" if they "realize" 
the same "transduction".

We conclude the section by showing how different models of "transducers"
proposed in the literature can be seen as instances of our model,
for a suitable choice of the "data structure" $\bbD$: 
\begin{itemize}
	\item \AP
				{\bf ""Sequential transducers""} \cite{Schutzenberger77}. 
				These are "transducers" over a "data structure" whose data 
				are words over an output alphabet $\Gamma$, and whose "updates" 
				are of the form $f_w:x\mapsto xw$, appending a fixed string $w$ 
				to the input string $x$.
				Figure~\ref{fig:transducer-diagram} shows an instance of this model. 
				Besides the special "type" $\init$, with singleton "domain"
				$\bbD_{\init*}=\{\emptystr\}$, the "data structure" has a single 
				numeric "type" $\tau=1$, whose "domain" is $\bbD_1=\Gamma^*$.
	\item \AP 
				{\bf ""Polynomial automata""} \cite{BenediktDSW17}.
		  	These are "transducers" over the "polynomial register algebra".
		  	\AP
		  	A slight generalization of this class, called ""cost register automata"",
		  	was introduced in \cite{AlurDDRY13} as a "transducer" that manipulates numeric 
		  	registers using generic arithmetic operations. The latter model is also very similar 
		  	in spirit to our notion of "transducer" over an abstract "data structure".
	\item \AP 
				{\bf ""Streaming string-to-string transducers"" (\reintro{SSTs})} \cite{AlurCerny10}.
		    These are "transducers" over the "string register algebra".
			  For this model, it is worth recalling that there is a particularly interesting 
			  subclass of "SSTs" with "copyless updates", which captures precisely the 
			  string-to-string "transductions" definable in monadic second-order logic \cite{Cou97}. 
	\item \AP 
				{\bf ""Streaming string-to-term transducers"" (\reintro{STTs})}.
		  	These are "transducers" that consume input words, manipulate "terms" by
		  	"first-order substitutions", and eventually output "ground terms". They
		  	are essentially the model of \cite[Section 3.6]{AlurDAntoni17}, except
		  	that our output trees are ranked, whereas those of \cite{AlurDAntoni17}
		  	are unranked. Related models over unranked, unordered trees have also
		  	been studied in \cite{BoiretPS18fsttcs}.
		  	As discussed in Section~\ref{subsec:data}, "substitutions" yield at
		  	least two natural "data structures" for defining "updates" on "terms",
		  	and hence two variants of "STTs".
		  	\AP
		  	The first, called ""downward STT"", uses the "leaf substitution algebra"
		  	and constructs outputs by extending "linear terms" at the leaves. The
		  	second, called ""upward STT"", uses the "free term algebra" and extends
		  	tuples of "ground terms" at the roots. When "updates" of the latter
		  	variant are restricted to be "copyless" (as for "SSTs", and as enforced
		  	in \cite{AlurDAntoni17}), the "downward@@STT" and "upward@@STT"
		  	variants, both enhanced with ""regular look-ahead""\,%
		  	\footnote{Regular look-ahead is the ability of a "transducer"
		  					  to annotate its input with a co-deterministic Mealy machine, 
		  					  prior to consuming it.},
		  	are equivalent in expressive power and capture precisely the
		  	string-to-"term" transductions definable in monadic second-order logic
		  	\cite[Theorem 16]{AlurDAntoni17}.
\end{itemize}

%% file: fig-transducer-diagram.tex
\begin{tikzpicture}[baseline=(current bounding box.center)]
\matrix (M) [matrix of nodes, column sep=12mm, row sep=0mm, nodes={anchor=center}] {
  |[name=init]| $\initspace$ & |[name=space]| $\States{Q}$ & |[name=halt]| $\haltspace$ \\
};
\path (init) edge [partial arrow] node [above=1mm] {\scriptsize $\deltainit$} (space);
\path (space) edge [partial arrow, loop above] node [above=1mm] {\scriptsize $a ~:~ \delta_a$} (space);
\path (space) edge [partial arrow] node [above=1mm] {\scriptsize $\deltahalt$} (halt);
\end{tikzpicture}
\qquad\quad
\begin{tikzpicture}[baseline=(current bounding box.center)]
\matrix (N) [matrix of nodes, column sep=22mm, row sep=5mm, nodes={anchor=center}] {
  & |[name=q0]| $(q_0,\Gamma^*)$ & \\
  |[name=initp]| $(q_{\init},\{\emptystr\})$ & & |[name=haltp]| $(q_{\halt},\Gamma^*)$ \\
  & |[name=q1]| $(q_1,\Gamma^*)$ & \\
};
\path (initp) edge [arrow] node [sloped, above=1mm, pos=0.4] {\scriptsize $\init: x\mapsto x$} (q0);
\path (q0) edge [arrow, bend right=30] node [left=1mm] {\scriptsize $a: x\mapsto x \# a$} (q1);
\path (q1) edge [arrow, bend right=30] node [right=1mm] {\scriptsize $a: x\mapsto x a$} (q0);
\path (q0) edge [arrow] node [sloped, above=0.5mm, pos=0.45] {\scriptsize $\halt: x\mapsto x\#$} (haltp);
\path (q1) edge [arrow] node [sloped, below=0.5mm, pos=0.4] {\scriptsize $\halt: x\mapsto x$} (haltp);
\end{tikzpicture}

%% file: morphisms.tex

\subsection{Transducer morphisms}\label{subsec:transducer-morphisms}

Let $A=(Q_A,\deltainit,(\delta_a)_{a\in\Sigma},\deltahalt)$
and $B=(Q_B,\kappainit,(\kappa_a)_{a\in\Sigma},\kappahalt)$ 
be two "transducers" from the same class, that is, with the same 
input alphabet $\Sigma$ and the same "data structure" $\bbD$.
\AP
A ""transducer morphism"" from $A$ to $B$ is a "transformation" 
$h: \States{Q_A} \rightharpoonup \States{Q_B}$ such that
\begin{itemize}
\item $h$ is a total function,
\item $\kappainit = \deltainit \comp h$,
\item $h \comp \kappa_a = \delta_a \comp h$, for every $a\in\Sigma$,
\item $h \comp \kappahalt = \deltahalt$.
\end{itemize}
The above identities are understood as equalities between the
respective "domains" and "codomains" of the partial functions, 
conjoined with the pointwise equalities of the images, that is, 
$f=g$ means $\Dom(f)=\Dom(g)$, $\Cod(f)=\Cod(g)$, and $f(x)=g(x)$ for all $x\in\Dom(f)$.
In particular, a "morphism" $h: A \rightharpoonup B$ can be equally defined 
as a total "transformation" that makes
the diagrams in Figure \ref{fig:transducer-morphism-diagrams} commute.
Of course, from the commutativity of these diagrams it also 
follows that the "transducers" $A$ and $B$ are "equivalent".

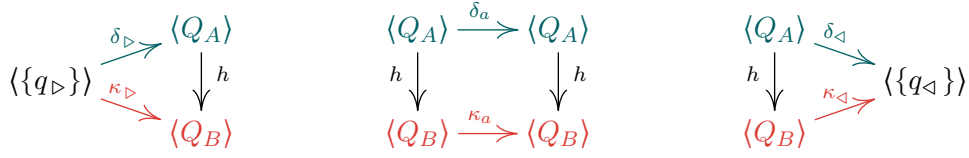
\begin{figure*}[t]
\begin{align*}
\input{fig-transducer-morphism-diagrams}
\end{align*}
\caption{Commutative diagrams for a "transducer morphism" $h$.}
\label{fig:transducer-morphism-diagrams}
\end{figure*}

To conclude the section and support intuition, Figure \ref{fig:transducer-morphism}
presents an example of a "morphism" between two instances of the "sequential transducer" model.
The "transducer" at the top is the same as the one shown in Figure \ref{fig:transducer-diagram}:
it inserts a separator symbol $\#$ before each pair of consumed input symbols. 
The "transducer" at the bottom is an "equivalent" one that anticipates the insertion 
at the previous "state". 
Notice that, in order to satisfy the commutativity conditions, the "morphism" transforms 
the data component of the former "transducer" into that of the latter in a non-trivial way. 
We also observe that no inverse "morphism" exists from the bottom "transducer" to the top one, 
since "updates" appending non-empty strings are not reversible in the underlying "data structure".

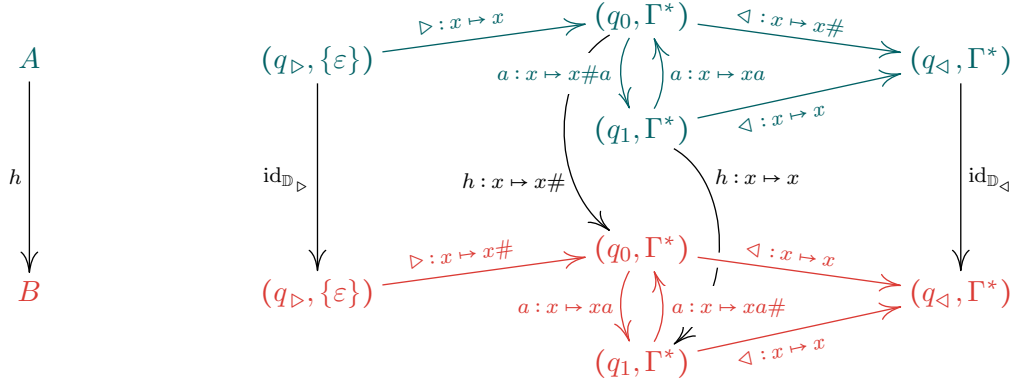
\begin{figure*}[t]
\begin{align*}
\input{fig-transducer-morphism}
\end{align*}
\caption{Example of "morphism" between two "transducers".}
\label{fig:transducer-morphism}
\end{figure*}

%% file: fig-transducer-morphism-diagrams.tex
\begin{tikzpicture}[baseline=(current bounding box.center)]
\matrix (M) [matrix of nodes, column sep=10mm, row sep=3mm, nodes={anchor=center}] {
  & |[nicecyan,name=space1]| $\States*{Q_A}$ && |[nicecyan,name=space2]| $\States*{Q_A}$ & |[nicecyan,name=space3]| $\States*{Q_A}$ && |[nicecyan,name=space4]| $\States*{Q_A}$ & \\
  |[name=init]| $\States*{\{\initstate*\}}$ & && & && & |[name=halt]| $\States*{\{\haltstate*\}}$ \\
  & |[nicered,name=space1p]| $\States*{Q_B}$ && |[nicered,name=space2p]| $\States*{Q_B}$ & |[nicered,name=space3p]| $\States*{Q_B}$ && |[nicered,name=space4p]| $\States*{Q_B}$ & \\
};

\path[nicecyan] (init) edge [partial arrow] node [above=1mm, pos=0.4] {\scriptsize $\deltainit$} (space1);
\path[nicecyan] (space2) edge [partial arrow] node [above=1mm, pos=0.4] {\scriptsize $\delta_a$} (space3);
\path[nicecyan] (space4) edge [partial arrow] node [above=1mm, pos=0.4] {\scriptsize $\deltahalt$} (halt);
\path[nicered] (init) edge [partial arrow] node [above=1mm, pos=0.4] {\scriptsize $\kappainit$} (space1p);
\path[nicered] (space2p) edge [partial arrow] node [above=1mm, pos=0.4] {\scriptsize $\kappa_a$} (space3p);
\path[nicered] (space4p) edge [partial arrow] node [above=1mm, pos=0.4] {\scriptsize $\kappahalt$} (halt);

\path (space1) edge [partial arrow] node [right=2mm,pos=0.4] {\scriptsize $h$} (space1p);
\path (space2) edge [partial arrow] node [left=2mm,pos=0.4] {\scriptsize $h$} (space2p);
\path (space3) edge [partial arrow] node [right=2mm,pos=0.4] {\scriptsize $h$} (space3p);
\path (space4) edge [partial arrow] node [left=2mm,pos=0.4] {\scriptsize $h$} (space4p);
\end{tikzpicture}

%% file: fig-transducer-morphism.tex
\begin{tikzpicture}[baseline=(current bounding box.center)]
\matrix (M) [matrix of nodes, column sep=29mm, row sep=2mm, nodes={anchor=center}] {
  & & |[nicecyan,name=q0]| $(q_0,\Gamma^*)$ & \\
  |[nicecyan,name=A]| $A$ & |[nicecyan,name=init]| $(q_{\init*},\{\emptystr\})$ & & |[nicecyan,name=halt]| $(q_{\halt*},\Gamma^*)$ \\[2ex]
  & & |[nicecyan,name=q1]| $(q_1,\Gamma^*)$ & \\[6ex]
  & & |[nicered,name=q0p]| $(q_0,\Gamma^*)$ & \\
  |[nicered,name=B]| $B$ & |[nicered,name=initp]| $(q_{\init*},\{\emptystr\})$ & & |[nicered,name=haltp]| $(q_{\halt*},\Gamma^*)$ \\[2ex]
  & & |[nicered,name=q1p]| $(q_1,\Gamma^*)$ & \\
};

\path (A) edge [arrow] node [left=1mm] {\scriptsize $h$} (B);
\path (init) edge [arrow] node [left=1mm] {\scriptsize $\identity{\bbD_{\init*}}$} (initp);
\path (q0) edge [arrow,bend right=60] node [left=1mm,pos=0.7] {\scriptsize $h: x\mapsto x\#$} (q0p);
\path (q1) edge [arrow,bend left=60] node [right=1mm,pos=0.2] {\scriptsize $~h: x\mapsto x$} (q1p);
\path (halt) edge [arrow] node [right=1mm] {\scriptsize $\identity{\bbD_{\halt*}}$} (haltp);

\path[nicecyan] (init) edge [arrow] node [sloped,above=1mm,pos=0.4] {\scriptsize $\init*: x\mapsto x$} (q0);
\path[nicecyan] (q0) edge [arrow,bend right=30] node [shape=rectangle,inner sep=1mm,fill=white,left=0.25mm] {\scriptsize $a: x\mapsto x \# a$} (q1);
\path[nicecyan] (q1) edge [arrow,bend right=30] node [right=1mm] {\scriptsize $a: x\mapsto x a$} (q0);
\path[nicecyan] (q0) edge [arrow] node [sloped,above=0.5mm,pos=0.45] {\scriptsize $\halt*: x\mapsto x\#$} (halt);
\path[nicecyan] (q1) edge [arrow] node [sloped,below=0.5mm,pos=0.4] {\scriptsize $\halt*: x\mapsto x$} (halt);

\path[nicered] (initp) edge [arrow] node [sloped,above=1mm,pos=0.4] {\scriptsize $\init*: x\mapsto x\#$} (q0p);
\path[nicered] (q0p) edge [arrow,bend right=30] node [left=1mm,fill=white] {\scriptsize $a: x\mapsto x a$} (q1p);
\path[nicered] (q1p) edge [arrow,bend right=30] node [shape=rectangle,inner sep=1mm,fill=white,right=0.25mm] {\scriptsize $a: x\mapsto x a \#$} (q0p);

\path[white] (q0p) edge [arrow,line width=2mm] (haltp);

\path[nicered] (q0p) edge [arrow] node [sloped,above=0.5mm,pos=0.45] {\scriptsize $\halt*: x\mapsto x$} (haltp);
\path[nicered] (q1p) edge [arrow] node [sloped,below=0.5mm,pos=0.4] {\scriptsize $\halt*: x\mapsto x$} (haltp);
\end{tikzpicture}

%% file: images-of-morphisms.tex

We fix again a "data structure" $\bbD$ and the induced class 
of "transducers".
We also fix a "transduction" $\varphi$ "realizable" by "transducers"  
in this class, and think of the sub-class of "transducers"
that "realize" $\varphi$ as a "category", with
arrows denoting the possible "morphisms" between "transducers".

To establish the existence of an "algebraically minimal"
"transducer" that computes $\varphi$, we use the approach 
outlined in Section \ref{sec:category}, which 
boils down to showing that the "category" of "transducers"
that compute $\varphi$ admits "strong factorizations" of 
"transducer morphisms", together with induced 
"images", an "initial object", and also a "final object" 
in the sub-category of "initial images". 
In the next subsections we provide precisely the constructions 
of such objects.

\subsection{Images of transducer morphisms}\label{sec:factorizations}

Recall that the "image" of a "morphism" $h: A \to B$ is an object $C$
together with an "epi" $f: A\epi C$ and a "strong mono" $g: C\smono B$
such that $h=f\comp g$. For an ordinary function between sets, this is just
the usual "range" construction: $C=\Rng(h)$, the map $f$ is obtained from
$h$ by replacing its "codomain" with $C$, and $g$ is the inclusion map 
of $C$ into $B$, as depicted below:
\begin{align*}
\input{fig-strong-factorization}%
\end{align*}
In our setting, where $A$ and $B$ are "transducers", a similar idea
applies: taking the "image" of $h$ amounts at restricting the
"configuration space" of $B$ to those elements that can be reached 
via $h$ from the "configuration space" of $A$.
This restriction acts at two levels: on the "control states", 
by pruning the unreachable ones, and on the "data domains" associated 
with them, by removing (most of the) unreachable data.
While pruning "control states" is easy, it is not clear at the
moment how one can restrict a "data domain".
Our approach is to \emph{approximate} the set of reachable data as 
the "solution set" of a suitable system of equations.
Indeed, an exact representation of the set of reachable data is not necessary here, 
provided that "updates" cannot distinguish the exact set from its approximation.
To formalize this idea, we shall expand our "data structure"
$\bbD$ with abstract "types", each associated with a "domain@@constrained"
that is defined by a system of equations.

\AP
We define a ""constraint"" of type $\tau$ as a (possibly infinite) set 
$\gamma$ of pairs $(f,g)$ of "updates" with $\tau$ as input "type".
Any such pair $(f,g)$ represents an equation%
\footnote{This idea is similar to the use of equalizers in category theory;
          in particular, a "constraint" can be seen as a special case 
          of a multi-equalizer.}
of the form $f(x) = g(x)$, 
where $x$ is a variable denoting arbitrary data from the unrestricted
"domain" $\bbD_\tau$.
\AP
The ""solution set"" of the "constraint" $\gamma$ is defined as
\[
    \reintro{\bbD*_\gamma} ~=~ 
    \big\{ d\in\bbD_\tau \::\: \forall (f,g)\in\gamma ~~ f(d) = g(d) \big\}.
\]
The "constraint" $\gamma$ can be added to the "data structure" $\bbD$
as a new "type" (actually a sub-type of $\tau$), 
and its "solution set" $"\bbD*_\gamma"$ can be thought of as 
the "data domain" associated with it.
\AP
We call these new "types" and "domains", respectively, 
""constrained types"" and ""constrained domains"".

The "updates" from a "constrained type" $\gamma$ are
nothing but the "updates" from the underlying basic "type" $\tau$ 
restricted to $\bbD_\gamma$. Of course, also the "codomain" of
an "update" can be "constrained@constrained domain", provided 
that it still covers all the images of the "update".
\AP
\phantomintro{constrained configuration space}%
\phantomintro{constrained transformation}%
\phantomintro{constrained initial transformation}%
\phantomintro{constrained internal transformation}%
\phantomintro{constrained final transformation}%
\phantomintro{constrained transducer}%
\phantomintro{constrained transducer morphism}%
The other definitions of "configuration space", 
"transformation", "transducer", and "transducer morphism" 
are adapted, almost verbatim, in presence of "constrained domains". 

\begin{example}\label{ex:constraint-polynomial}
Let $\bbD$ be the "polynomial register algebra" and consider the "updates" $f,g:\bbD_2\to\bbD_1$ 
defined by $f(x,y)=x^2+y^2$ and $g(x,y)=1$. The "constraint" $f(x,y)=g(x,y)$ induces a new
"domain@@constrained" $\bbD_{\{(f,g)\}} \subseteq \bbD_2$ that coincides
with the unit circle.
\end{example}

By design, "constrained domains" are closed under inverses of "updates"
and intersections:

\begin{lemma}[restate=Constraints, name=\textbf{closures of constrained domains}]%
    \label{lem:constraint-closure}
"Constrained domains" are effectively closed under 
inverses of "updates" and under finite intersections. 
They are also closed (non-effectively) under infinite intersections.
\end{lemma}

\begin{proof}
To prove closure of "constrained domains" under inverses of "updates", 
consider a set $E\subseteq \bbD_\beta^\gamma\times\bbD_\beta^\gamma$ 
representing a system of equations inducing a "constrained domain" 
$\Sol(E)$ over the basic "type" $\beta$.
Next, consider an "update" $f\in\bbD_\alpha^\beta$. 
For an arbitrary data $d$ of "type" $\alpha$, we have 
$d \in f^{-1}(\Sol(E))$ iff $f(d) \in \Sol(E)$.
This means that the set $f^{-1}(\Sol(E))$ contains all and only
the "solutions" of the following system of equations:
\[
      E' = \{(f\comp u, f\comp v) \::\: (u,v)\in E \}.
\]
In particular, $f^{-1}(\Sol(E))$ is a "constrained domain"
over the basic type $\alpha$ represented by system $E'$,
which can be computed from $E$ and $f$.

Closure of "constraints" under intersections is straightforward, 
as taking the intersection of "constraints" corresponds
to taking the conjunction of corresponding systems of equations.
Moreover, if the intersection is finite, then the conjunction 
is also finite and hence computable.
\end{proof}

\AP
The above result is useful for deriving a "closure operator" 
that over-approximates any set $D$ of data of uniform "type" $\tau$
by the intersection of all the possible "constrained domains" that contain $D$. 
We call this over-approximation the \reintro{closure of $D$},
and denote it by $""\cl*(D)""$. 
Equivalently, one can define $\cl(D)$ as the "solution set" $\bbD_{\gamma_D}$
of the "constraint" $\gamma_D = \big\{ (f,g) \::\: \forall d\in D ~ f(d)=g(d) \big\}$ 
that contains all equations that are valid over $D$.

\begin{example}\label{ex:closure-polynomials}
Let $\bbD$ be again the "polynomial register algebra" and 
$D = \big\{ \big(\frac{1-n^2}{1+n^2}, \frac{2n}{1+n^2}\big) \::\: n\in\bbN \big\} \subseteq \bbD_2$.
\LIPICSorARXIV{%
This set contains infinitely many points on the unit circle
from Example \ref{ex:constraint-polynomial},
and its "closure" coincides with circle itself.
}{%
This set contains infinitely many points that can 
be graphically represented as follows:
\[
\begin{tikzpicture}[scale=0.4]
  \draw[gray] (0,0) circle (1);

  \foreach \n in {0,1,...,10}{
    \pgfmathsetmacro{\x}{(1-\n*\n)/(1+\n*\n)}
    \pgfmathsetmacro{\y}{2*\n/(1+\n*\n)}
    \fill (\x,\y) circle ({8mm * 1/(\n+5)});
  }

  \node at (-1,0) {$\times$};
\end{tikzpicture} 
\]
Its "closure" is the unit circle from Example \ref{ex:constraint-polynomial}.
}%
We also observe that every finite system of polynomial equations, and hence
every "constrained domain" of the "polynomial register algebra", can be
represented by a single "constraint" of the form $f(x)=g(x)$, for suitable
"updates" $f,g: \bbD_n \to \bbD_m$. Moreover, using the fact that the set of
roots of a product of polynomials is the union of the sets of roots of its 
factors, one can show that "constrained domains" are closed under finite unions. 
In particular, every finite set is its own "closure". 
Thus, over the "polynomial register algebra", the "closure" operator behaves
like a topological closure --- as a matter of fact, this is known as Zariski 
topology \cite{CoxLittleOShea15}. This behavior, however, is specific to the
polynomial setting and should not be expected for arbitrary "data structures".
For example, it fails for the "data structure" with linear or affine "updates"
and also for the "free term algebra", as shown in the next example.
\end{example}

\begin{example}\label{ex:closure-terms}
Let $\bbD$ be the "free term algebra" over a ranked alphabet $\Gamma$.
In this "data structure", the only "constrained domains" contained in
$\bbD_1$ are $\emptyset$, the singletons, and $\bbD_1$ itself; this 
follows from Lemma~\ref{lem:non-overlap}, presented later in 
Section~\ref{subsec:downwardSTT}.
In particular, if $D \subseteq \bbD_1$ contains at least two distinct "terms", 
then $\cl(D)=\bbD_1$.
Non-trivial "closures" may nevertheless arise at higher numerical "types".
For instance, let $D=\{(a,a),(b,b)\}\subseteq\bbD_2$, where $a,b$ are distinct
rank-$0$ symbols of $\Gamma$. 
If $\Gamma$ also contains a binary symbol $c$, then
$\cl(D)=\big\{(t,t) \::\: t\in\bbD_1\big\}$,
as witnessed by the equation $f(x,y)=g(x,y)$, with
$f:(x,y)\mapsto c(x,y)$ and $g:(x,y)\mapsto c(y,x)$.
\end{example}

\AP
It is easy to see that the "closure" operator $\cl(-)$ is 
""extensive"" (i.e.~$D\subseteq\cl(D)$), 
""monotone"" (i.e.~$D\subseteq D'$ implies $\cl(D)\subseteq\cl(D')$), and 
""idempotent"" (i.e.~$\cl(\cl(D))=\cl(D)$), and hence the following lemma
holds:

\begin{lemma}[restate=ClosureUnionPush,
              name={\textbf{closure vs union}}]\label{lem:closure-union-push}
For all sets $D,D'\subseteq\bbD_\tau$, 
$\cl(D\cup D') = \bigcl(\cl(D) \cup \cl(D'))$.
\end{lemma}

\begin{proof}
By "extensiveness", we have $D \cup D' \subseteq \cl(D) \cup \cl(D')$, and
hence, by "monotonicity", $\cl(D \cup D') \subseteq \cl(\cl(D) \cup \cl(D'))$.
Conversely, $D \subseteq D\cup D'$ and $D'\subseteq D\cup D'$ imply, by "monotonicity",
$\cl(D) \cup \cl(D') \subseteq \cl(D \cup D')$, and hence 
$\cl(\cl(D) \cup \cl(D')) \subseteq \cl(\cl(D\cup D'))$.
Finally, "idempotency" implies $\cl(\cl(D\cup D')) = \cl(D \cup D')$.
\end{proof}

\AP
It is also convenient to generalize the "closure operator" 
from sets of data to sets of "configurations". Given $S \subseteq \States{Q}$, 
we define $""\clS*(S)"" = \big\{(q,d) \::\: q\in Q, \: d\in \cl(D_{S,q})\big\}$,
where $D_{S,q} = \big\{d'\in\Data(q) \::\: (q,d')\in S\big\}$.
Note that $\clS(S)$ can be seen as a new "configuration space"
$\States{\tilde Q}$, where $\tilde Q$ contains a copy $\tilde q$ 
of each "state" $q\in Q$ associated with a "constrained domain" 
$\Data(\tilde{q})=\cl(S_q)$.


\medskip
We will soon see that the "image" of a "transformation" $h$ between
"configuration spaces", and hence in particular of a "transducer morphism", 
is obtained as the "closure" of the "range" of $h$. The current definition of
"constrained domain", however, may lack a finite presentation, since it allows
infinite systems of equations. This is harmless for proving the existence of
"images" of "transducer morphisms", but it raises the question of whether the
same result holds when all components of a "transducer" are required to be
finitely presentable.
\AP
A natural setting is obtained by restricting to ""finitary constrained domains"", 
namely "constrained domains" that are "solution sets" of finite systems of equations. 
Since finite systems are not, by themselves, closed under arbitrary intersections, 
we rely on the following compactness assumption for the "data structure" $\bbD$ 
under consideration:

\begin{assumption}[restate=Compactness,
                   name={\textbf{compactness}}]\label{ass:compactness}
For every numeric "type" $\tau$ of $\bbD$ and for every "constraint" 
$\gamma$ of "type" $\tau$, there is a finite subset $\gamma'$ of $\gamma$ 
that is a "constrained type" of $\bbD$ and has the same "solution set" as $\gamma$. 
Therefore, all "constrained domains" are allowed in $\bbD$ and they are "finitary".
\end{assumption}

The canonical example of a "data structure" satisfying this assumption 
is, again, the "polynomial register algebra". For this structure, Hilbert's finite 
basis theorem \cite{Hilbert1890} implies that every "constraint" is equivalent 
to a finite subset of it. This result serves as a foundation for deriving 
similar compactness properties for other "data structures", such as the 
"string register algebra" or the "free term algebra". 
We shall use Hilbert's theorem later, when we will present concrete
classes of "transducers" that admit "algebraically minimal" objects.

\medskip
Now, turning to the proof of existence of "images" of "transducer morphisms", 
we are going to exploit closure properties of "constrained domains" and the derived "closure" operator.
We first introduce a few technical lemmas, one showing a continuity property of
"transformations" w.r.t.~"closures@@space", another characterizing "transformations" that 
are "epis", and a third one giving sufficient conditions for "transformations" to be 
"strong mono".

\begin{lemma}[restate=ClosureContinuity, 
              name={\textbf{closure-continuity of transformations}}]
              \label{lem:closure-continuity}
For every "transformation" $f: \States{Q} \to \States{Q'}$
and every subset $S$ of $\States{Q}$, we have
$f(\clS(S)) \subseteq \clS(f(S))$.
\end{lemma}

\begin{proof}
We first prove closure-continuity for "updates", namely, that
for every $f:\bbD_\alpha\to\bbD_\beta$ and every $D \subseteq \bbD_\alpha$,
\[
    f(\cl(D)) \subseteq \cl(f(D)).
\]
For this, we are going to exploit the definition of "closure" of a set as the intersection of all "constrained domains"
that contain that set. So, consider an arbitrary "constrained domain" $E \subseteq \bbD_\beta$ that contains $f(D)$.
By Lemma \ref{lem:constraint-closure}, $f^{-1}(E)$ ($\subseteq \bbD_\alpha$) is also a "constrained domain".
Moreover, $D\subseteq f^{-1}(E)$, because $d\in D$ implies $f(d) \in f(D) \subseteq E$, and hence $d\in f^{-1}(E)$.
By definition of $\cl(D)$ as the intersection of all "constrained domains" that contain $D$, we get
$\cl(D) \subseteq f^{-1}(E)$.
Applying $f$ and using the general fact that $f(f^{-1}(E)) \subseteq E$, we obtain 
$f(\cl(D)) \subseteq f(f^{-1}(E)) \subseteq E$.
Since the latter containment holds for all "constrained domains" $E \supseteq f(D)$, it also holds for their intersection,
and hence $f(\cl(D)) \subseteq \cl(f(D))$.

Next, we lift the above property to "transformations". More precisely, 
given $f:\States{Q} \to \States{Q'}$, "specified" by a 
pair $(\speca f,\specb f)$, and given $S\subseteq \States{Q}$, we let 
$D_{S,q} = \{d\in\Data(q) \::\: (q,d)\in S\}$ and derive:
\begin{align*}
    f\big(\clS(S)\big) 
    &~=~ \bigcup_{q\in Q} f\big(\{q\}\times\cl(D_{S,q}))
    \tag{by definition of "closure@@space"} \\
    &~=~ \bigcup_{q\in Q} \big\{\speca f(q)\big\} \times \specb f\big(\cl(D_{S,q})\big)
    \tag{by distributivity of $\cup$ and $\times$} \\
    &~\subseteq~ \bigcup_{q\in Q} \big\{\speca f(q)\big\} \times \bigcl(\specb f(D_{S,q}))
    \tag{by closure-continuity of "updates"} \\
    &~\subseteq~ \bigclS(\bigcup_{q\in Q} \{\speca f(q)\} \times \specb f(D_{S,q}))
    \tag{by subadditivity of "closure"} \\
    &~=~ \bigclS(f(S)).
    \tag{by definition of "closure@@space"}
\end{align*}
Note that the opposite containment does not hold in general, since the "closure" operator 
does not commute with union.
\end{proof}

\begin{lemma}[restate=Epis, name={\textbf{epi transformations}}]\label{lem:epis}
Let $f: \States{Q} \to \States{Q'}$ be a "transformation" "specified" by $(\speca f,\specb f)$. 
The following conditions are equivalent:
\begin{itemize}
\item $f$ is "epi",
\item $\Cod(f) \subseteq \clS(\Rng(f))$,
\item $\speca f: Q\to Q'$ surjective and 
      $\Data(r) \subseteq \bigcl(\bigcup_{q\in\speca f^{-1}(r)} \Rng(\specb f(q)))$ for every $r\in Q'$.
\end{itemize}
\end{lemma}

\begin{proof}
We start by observing that 
$$\Cod(f) = \bigcup_{r\in Q'} \bigcup_{q\in \speca f^{-1}(r)} \Cod(\specb f(q))$$ and 
$$\Rng(f) = \bigcup_{r\in Q'} \bigcup_{q\in \speca f^{-1}(r)} \Rng(\specb f(q)).$$
This implies the equivalence between the latter two conditions.

We now prove the implication from the first condition to the third.
Suppose that $f$ is "epi" and consider an arbitrary "state" $r\in Q'$.
Let $h,h'$ be the "transformations" over $\States{Q'}$
specified, respectively, by $(\speca h,\specb h)$ and $(\speca h',\specb h')$,
where $\speca h$ and $\specb h$ are both undefined on the entire set $Q'$,
$\speca h'$ (resp.~$\specb h'$) is defined only on $r$ and maps it to the
same "control state" $r$ (resp.~to the identity "update" over $\Data(r)$).
If $r\nin\Rng(\speca f)$, we would have 
$f\comp h = f\comp h'$, and hence $h=h'$ because $f$ is "epi".
Since this would contradict the definitions of $h$ and $h'$ and
since $r$ was chosen arbitrarily, 
we must conclude that $\speca f$ is surjective.
To complete the proof in one direction, we show that 
$\Data(r) \subseteq \bigcl(\bigcup_{q\in\speca f^{-1}(r)} \Rng(\specb f(q)))$.
For the sake of brevity, let $D$ be the "constrained domain" 
$\bigcl(\bigcup_{q\in\speca f^{-1}(r)} \Rng(\specb f(q)))$.
Suppose, by way of contradiction, that $\Data(r) \nsubseteq D$,
and let $d\in \Data(r) \setminus D$.
Since $d\nin D$, there exist a "constrained domain" $D'$ 
and a pair of "updates" $u,v: \Data(r) \to D'$ such that
$u(d)\neq v(d)$, and yet $u(d')=v(d')$ 
for all $d'\in \bigcup_{q\in\speca f^{-1}(r)}\Rng(\specb f(q))$.
We construct from this two "transformations" $h,h': \States{Q'} \to \States{Q''}$
such that $f\comp h = f\comp h'$. Formally, we let $Q''$ be the set that consists
of a single "control state" $q''$ with associated "data type" $\Data(q'')=D'$, and
"specify" $h$ and $h'$ respectively by the pairs $(\speca h,\specb h)$ and 
$(\speca h',\specb h')$, where
\begin{itemize}
    \item $\speca h$ and $\speca h'$ are defined only on $r$ and they both map 
          $r$ to $q''$,
    \item $\specb h$ and $\specb h'$ are defined only on $r$ and they map $r$
          to the "update" $u$ and $v$, respectively.
\end{itemize}
Note that $h$ and $h'$ agree on the images of the "configurations" $(r,d')\in\States{Q'}$, 
for all $d'\in \bigcup_{q\in\speca f^{-1}(r)}\Rng(\specb f(q))$, 
and hence, because the latter "configurations" are precisely the images of $f$,
we have $f\comp h = f\comp h'$. 
On the other hand, by construction, we have $h(q,d)=(q'',u(d)) \neq (q'',v(d))=h'(q,d)$,
and because $(q,d)$ is a "configuration" of $\States{Q'}$ where both $h$ and $h'$ are defined,
we have $h\neq h'$.
This contradicts the fact that $f$ is "epi", and thus proves that 
$\Data(r)$ is contained in $D=\bigcl(\bigcup_{q\in\speca f^{-1}(r)} \Rng(\specb f(q)))$.

We now prove the converse direction, from the third condition to the first.
Suppose that $\speca f$ is surjective and that
$\Data(r)$ is contained in the "constrained domain"
$D=\bigcl(\bigcup_{q\in\speca f^{-1}(r)} \Rng(\specb f(q)))$.
We head towards proving that $f$ is "epi".
Consider two "transformations" $h,h': \States{Q'} \to \States{Q''}$,
for an arbitrary "configuration space" $\States{Q''}$,
and suppose $f\comp h=f\comp h'$.
Using Lemma \ref{lem:composition-transformations}, 
the latter equality can be rephrased in terms of the "specifications"
$(\speca h,\specb h)$ and $(\speca h',\specb h')$ of $h$ and $h'$:
\begin{itemize}
\item $\speca f\comp \speca h = \speca f\comp \speca h'$,
\item $\specb f(q)\comp \specb h(\speca f(q)) = \specb f(q)\comp \specb h'(\speca f(q))$
      for all $q\in Q$.
\end{itemize}
Now, since $\speca f$ is surjective, the first item above implies 
$\speca h=\speca h'$.
Similarly, the second item implies that, for all $r\in Q'$, 
the two "updates" $\specb h(r)$ and $\specb h'(r)$ coincide 
on $\bigcup_{q\in\speca f^{-1}(r)} \Rng(\specb f(q))$.
Moreover, because any two "updates" that agree on a set also agree 
on its "closure", $\specb h(r)$ and $\specb h'(r)$ coincide
also on $D=\bigcl(\bigcup_{q\in\speca f^{-1}(r)} \Rng(\specb f(q)))$.
Putting all together and recalling that $\Data(r) \subseteq D$, 
we get $\specb h(r) = \specb h'(r)$ for all $r\in Q'$, and hence $h=h'$.
Finally, because $h,h'$ were chosen arbitrarily, we conclude that $f$ is "epi".
\end{proof}

\begin{lemma}[restate=StrongMonos, name={\textbf{strong mono transformations}}]\label{lem:strong-monos}
Every inclusion map $g: \States{Q} \to \States{Q'}$, where $\States{Q} \subseteq \States{Q'}$
and $g(q,d) = (q,d)$ for all $(q,d)\in\States{Q}$, is a "strong mono".
\end{lemma}

\begin{proof}
To prove that the inclusion map $g: \States{Q} \to \States{Q'}$ is a "strong mono", 
we consider some "transformations" $f: \States{R} \epi \States{R'}$, with $f$ "epi", 
$h:\States{R} \to \States{Q}$, and $h': \States{R'} \to \States{Q'}$ such that 
$f\comp h' = h\comp g$,
and we construct a "transformation" $d: \States{R'} \to \States{Q}$ such that
$h = f\comp d$ and $h' = d\comp g$, as in the following diagram:
\begin{align*}
\input{fig-inclusion-is-strong-mono}
\end{align*}

Intuitively, $d$ is obtained from $h'$ by restricting its "codomain" to $\States{Q}$
(recall that $\States{Q}\subseteq\States{Q'}$).
For this to make sense, we need to verify that $\Rng(h') \subseteq \States{Q}$.
Towards this, we observe that
\begin{enumerate}
    \item since $f\comp h'=h\comp g$, we have $h'(f(\States{R})) = g(h(\States{R}))$;
    \item since $g$ is an inclusion map, we have $g(h(\States{R})) = h(\States{R}) \subseteq \States{Q}$,
          and hence, by "monotonicity", $\clS(g(h(\States{R}))) \subseteq \cl(\States{Q}) = \States{Q}$;
    \item since $f$ is "epi", by Lemma \ref{lem:epis} we have
          $\States{R'} \subseteq \clS(f(\States{R}))$.
\end{enumerate}
Putting all together and using Lemma \ref{lem:closure-continuity}
to pull out the "closure@@space", we get
\[
  \Rng(h') \:=\: 
  h'(\States{R'}) \:\subseteq\: 
  h'(\clS(f(\States{R}))) \:\subseteq\:
  \clS(h'(f(\States{R}))) \:=\:
  \clS(g(h(\States{R}))) \:\subseteq\:
  \States{Q}.
\]
We can now define $d$ as $h'$ with the "codomain" restricted to $\States{Q}$.
This can implemented at the level of "specifications", as follows.
Given a "specification" $(\speca h',\specb h')$ of $h'$, 
the "specification" of $d$ is defined as $(\speca d,\specb d)$, where
\begin{itemize}
    \item $\speca d: R' \to Q$ is defined by $\speca d(r)=\speca h'(r)$ for all $r\in R'$,
    \item given $r\in\Dom(\speca h')$ and $q=\speca h'(r)$, 
          $\specb d(r): \Data_{R'}(r) \to \Data_{Q}(q)$ is the "update"
          obtained from $\specb h'(r): \Data_{R'}(r) \to \Data_{Q'}(q)$ 
          by restricting the "codomain" $\Data_{Q'}(q)$ to the set 
          $\Data_{Q}(q) = \bigcl(\bigcup_{q'\in{\speca h'}^{-1}(r)} \Rng(\specb h'(q')))$
          (note that the "state" $q$ is assigned the "data domain"
           $\Data_{Q'}(q)$ or, equally, the "constrained domain" $\Data_{Q}(q)$, 
           depending on whether $q$ is seen as an element of $Q'$ or $Q$).
\end{itemize}
By construction, we have $h = f\comp d$ and $h' = d\comp g$.

To conclude the proof, it remains to show that the defined "transformation" $d$ 
is the only possible one that makes the above diagram commute.
This follows easily from the fact that $f$ is an "epi"; indeed, for every
other "transformation" $d': \States{R'} \to \States{Q}$ such that
$h = f\comp d'$, we have $f\comp d = f\comp d'$ and hence $d=d'$.
\end{proof}

Since "transducer morphisms" are particular forms of "transformations",
the two lemmas above imply similar sufficient conditions for a "transducer morphism" 
to be "epi", resp.~"strong mono".%
\footnote{In fact, the existential quantification in the "diagonal fill-in property" 
          for the definition of "strong mono" requires to check that the resulting
          diagonal $d$ is a "transducer morphism". This is readily done by
          inspecting the proof of Lemma \ref{lem:strong-monos}%
          \LIPICS{ in the long version of the paper}, %
          which constructs $d$ from a given "transformation" $h'$ by restricting the "codomain".
          Clearly, this construction is also applicable to "transducer morphisms".}
We remark, however, that these results do not give exact characterizations of
"epis" and "strong monos" among "transducer morphisms". For instance, because
the definition of "epi" involves a universal quantification over the available
"arrows", a "transducer morphism" may be "epi" relative to the class of
"transducer morphisms", while failing to be "epi" relative to the larger 
class of "transformations".

Using the above results, we can now prove the existence of "strong factorizations" 
of "transducer morphisms".  

\begin{proposition}[restate=Image,name={\textbf{images}}]
                    \label{prop:images}
Every "transducer morphism" $h: A\to B$ has a "strong factorization" 
$A \epi[f] C \smono[g] B$, 
with $f$ "epi" and $g$ "strong mono".
\end{proposition}


\begin{proof}
Let $A=(Q_A,\deltainit,(\delta_a)_{a\in\Sigma},\deltahalt)$,
$B=(Q_B,\kappainit,(\kappa_a)_{a\in\Sigma},\kappahalt)$,
and let $(\speca h,\specb h)$ be "specification" of the "morphism" $h$.
We define the "transducer" 
$C=(Q_C,\chiinit,(\chi_a)_{a\in\Sigma},\chihalt)$
as a ``pruning'' of $B$:
\begin{itemize}
\item $Q_C=\Rng(\speca h)$, which is a subset of $Q_B$.
      Each "control state" $q$ in $Q_C$ is assigned a new "constrained type" 
      with the induced "domain"
      $\bigcl(\bigcup_{r\in\speca h^{-1}(q)} \Rng(\specb h(r)))$.
      Note that the same "control state" $q$ is assigned a possibly
      larger "domain" in $Q_B$, so we use explicit notation
      to distinguish the two "domains": 
      $\Data_{Q_B}(q)$ and $\Data_{Q_C}(q)$.
      Also observe that, by using the "generalization of the closure operator@closure@space"
      to sets of "configurations", we have $\States{Q_C} = \bigclS(\Rng(h))$.
\item $\chiinit$, $\chi_a$, and $\chihalt$ are obtained from
      the corresponding "transformations" $\kappainit$, $\kappa_a$,
      $\kappahalt$ of $B$ by replacing the "domain"/"codomain" 
      $\States{Q_B}$ with $\States{Q_C}$.
      These replacements can be done at the level of "specifications".
      For example, if $(\speca\kappa_a,\specb\kappa_a)$
      is the "specification" of the "internal transformation" $\kappa_a$
      of $B$, then the "specification" of the corresponding "transformation"
      $\chi_a$ of $C$ is $(\speca\chi_a,\specb\chi_a)$,
      where $\speca\chi_a$ is the restriction of $\speca\kappa_a$
      to the set $Q_C$ and 
      $\specb\chi_a$ maps every $q\in Q_C$
      to the "update" $\specb\kappa_a(q)$, with its "domain" and "codomain"
      restricted to $\Data_{Q_B}(q)$ and $\Data_{Q_B}(\speca\kappa_a(q))$, 
      respectively.
      This definition guarantees not only 
      $\Dom(\chi_a)\subseteq\States{Q_C}$, but also 
      $\Rng(\chi_a)\subseteq\States{Q_C}$:
      \begin{align*}
        \Rng(\chi_a) 
        &~=~ \chi_a(\States{Q_C}) 
            \tag{by definition of "range"} \\
        &~=~ \chi_a\big(\clS(h(\States{Q_A}))\big) 
            \tag{by definition of $Q_C$} \\
        &~\subseteq~ \bigclS(\chi_a(h(\States{Q_A}))) 
            \tag{by Lemma \ref{lem:closure-continuity}} \\
        &~=~ \bigclS(h(\delta_a(\States{Q_A})))
            \tag{since $h$ is a "morphism"} \\
        &~\subseteq~ \bigclS(h(\States{Q_A}))
            \tag{since $\Rng(\delta_a)\subseteq\States{Q_A}$} \\
        &~=~ \States{Q_C}.
            \tag{again by definition of $Q_C$}
      \end{align*}
\end{itemize}
We also define the "transformation" $f:\States{Q_A} \to \States{Q_C}$ 
as $h$, but with its "codomain" restricted to $\States{Q_C}$.
Lemma \ref{lem:epis} immediately implies that $f$ is an "epi".
Moreover, $f$ is a "transducer morphism" from $A$ to $C$:
indeed, because $h$ is a "morphism" from $A$ to $B$ and 
the "internal transformation" of $C$ is a restriction 
of that of $B$, we have 
$f(\delta_a(q,d)) = h(\delta_a(q,d)) = \kappa_a(h(q,d)) = \chi_a(f(q,d))$,
and similarly for $\deltainit$ and $\deltahalt$.

To conclude, we define the other part of the "factorization" of $h$, 
namely, the "strong mono" "morphism" $g$ from $C$ to $B$,
so that $h=f\comp g$. 
The "morphism" $g$ is simply the inclusion map from 
$\States{Q_C}=\bigclS(\Rng(h))$ to $\States{Q_B}$. 
Formally, it is specified by the pair $(\speca g,\specb g)$, where 
$\speca g: Q_C \to Q_B$ maps every $q\in\Rng(\speca h)$ 
to $q$ itself, and $\specb g$ maps every $q\in\Rng(\speca h)$ 
to the identity "update" $\specb g(q): \Data_{Q_C}(q) \to \Data_{Q_B}(q)$.
By Lemma \ref{lem:strong-monos}, we know that $g$ is a "strong mono".
\end{proof}

%% file: fig-strong-factorization.tex
\\[-4ex]
\begin{tikzpicture}
\draw (40mm,0) ellipse (4mm and 8mm);
\node (B) at (40mm,10mm) {$B$};

\fill [gray!20] (0,8mm) -- (20mm,3mm) -- (20mm,-3mm) -- (0,-8mm) -- cycle;
\fill [gray!20] (20mm,3mm) -- (40mm,3mm) -- (40mm,-3mm) -- (20mm,-3mm) -- cycle;
\draw (0,8mm) -- (20mm,3mm);
\draw (20mm,-3mm) -- (0,-8mm);
\draw (20mm,3mm) -- (40mm,3mm);
\draw (20mm,-3mm) -- (40mm,-3mm);

\draw [fill=white] (0,0) ellipse (4mm and 8mm);
\node (A) at (0,10mm) {$A$};

\draw [fill=white] (20mm,0) ellipse (2mm and 3mm);
\fill [pattern=north east lines, pattern color=gray!50] (20mm,0) ellipse (2mm and 3mm);
\node (C) at (20mm,5mm) {$C$};

\draw [dashed, fill=white] (40mm,0) ellipse (2mm and 3mm);
\fill [pattern=north east lines, pattern color=gray!50] (40mm,0) ellipse (2mm and 3mm);

\draw (A) edge [arrow, bend left=15] node [above=1mm] {\scriptsize $h$} (B);
\draw (A) edge [epi arrow] node [above=1mm] {\scriptsize $f$} (C);
\draw (C) edge [reg mono arrow] node [above=1mm] {\scriptsize $g$} (B);
\end{tikzpicture}

%% file: fig-inclusion-is-strong-mono.tex
\begin{tikzpicture}[baseline=(current bounding box.center)]
\matrix (M) [matrix of math nodes, row sep=16mm, column sep=18mm] {
  |[name=A]| \States{R} & |[name=B]| \States{R'} \\
  |[name=C]| \States{Q} & |[name=D]| \States{Q'} \\
};
\path (A) edge [epi arrow] node [above=1mm] {$f$} (B);
\path (C) edge [arrow] node [below=1mm] {$g$} (D);
\path (A) edge [arrow] node [left=1mm] {$h$} (C);
\path (B) edge [arrow] node [right=1mm] {$h'$} (D);
\path (B) edge [arrow,dashed] node [above left=1mm] {$d$} (C);
\end{tikzpicture}

%% file: initial-transducer.tex

\subsection{The initial transducer}\label{sec:initial-transducer}

The construction of an "initial transducer" does not require additional assumptions
and it is a straightforward generalization of the construction of the "initial object"
in the "category" of finite state automata. 
Intuitively, the "states" of the "initial transducer" are identified with
the input words over $\Sigma$ and there is no internal memory. The "output"
is produced by simply mapping each state $w$ to the corresponding value
$\varphi(w)$, if defined.
As a matter of fact, any partial function $\varphi: \Sigma^* \to \bbD_{\halt*}$
can be "realized" by such a "transducer".

\phantomintro{\initialtransd*}%
\begin{proposition}[restate=InitialTransducer, 
			  name={\textbf{initial transducer}}]\label{prop:initial-transducer}
There is an (infinite) "transducer" $\reintro{\initialtransd*}$,
with words over $\Sigma$ as "control states",
that is "initial", namely, that admits a unique 
"morphism" $\initialmorph{A}: \initialtransd \to A$
towards each "transducer" $A$ that "realizes" $\varphi$.
\end{proposition}

\begin{proof}
As claimed in the proposition, the set "control states" of 
$\initialtransd$ is $\Sigma^*$. 
All "control states" $w\in\Sigma^*$ are associated with the
"data type" $\init$, i.e.~the same that is associated with $\initstate$,
and the corresponding "domain" is the singleton $\bbD_{\init*} = \{\initdata\}$.
Basically, $\initialtransd$ is an infinite-state "transducer" 
without registers.
The "transformations" of $\initialtransd$ are also easily defined:
\begin{itemize}
\item $\deltainit$ maps $(\initstate,\initdata)$ 
    to the "configuration" $(\emptystr,\initdata)$;
      formally, $\deltainit$ is "specified" by a pair 
      of functions $(\speca{\delta}_{\init*}, \specb{\delta}_{\init*})$, 
      where $\speca{\delta}_{\init*}$ (resp.~$\specb{\delta}_{\init*}$) 
      maps $\initstate$ to the "control state" $\emptystr$ 
      (resp.~to the identity "update" $\initdata \mapsto \initdata$);
\item for every $a\in\Sigma$, $\delta_a$ maps any "configuration"
    $(w,\initdata)$ to the "configuration" $(wa, \initdata)$; 
    this is "specified by" the pair $(\speca \delta_a, \specb \delta_a)$, 
    where $\speca{\delta_a}(w)=wa$ and $\specb{\delta_a}(w)$ is again 
    the identity "update" $\initdata \mapsto \initdata$;
\item $\deltahalt$ maps any "configuration" $(w,\initdata)$
      for which $\varphi(w)$ is defined to the pair $(\haltstate, \varphi(w))$,
      namely, $\speca{\delta}_{\halt*}(w)=\haltstate$ and 
      $\specb{\delta}_{\halt*}(w): \initdata \mapsto \varphi(w)$;
    if $\varphi(w)$ undefined, then
    $\delta_\halt(w,\initdata)$ is undefined too.
\end{itemize}
This "transducer" clearly "realizes" $\varphi$.
We also observe that the only component that depends on $\varphi$
is the "final transformation" $\delta_\halt$.

To prove that this is an "initial object", consider any "transducer" 
$A=(\States{Q'},\kappainit,(\kappa_a)_{a\in\Sigma},\kappahalt)$ 
that "realizes" the same "transduction" $\varphi$.
We can construct an "initial@@arrow" "morphism" 
$\initialmorph{A}: \initialtransd \rightharpoonup A$
by simply mapping every "configuration" $(w,\initdata)$
of $\initialtransd$ to the "configuration" 
$\kappa*_{\init* w}(\initstate,\initdata)$ 
reached by $A$ after reading $w$, 
assuming $\kappa*_{\init* w}(\initstate,\initdata)$ is defined
(otherwise, if it is not defined, we let 
$\initialmorph{A}(w,\initdata)$ be undefined too).
Formally, the "morphism" $\initialmorph{A}$ is "specified by" 
the pair $(\specainitialmorph{A}, \specbinitialmorph{A})$
of partial functions that map any $w\in Q$ ($=\Sigma^*$)
respectively to the "control state" $q_w$ of $A$ 
and to the "update" $u_w: \initdata \mapsto d_w$,
assuming $(q_w, d_w) = \kappa*_{\init* w}(\initstate,\initdata)$.

We omit the straightforward proof that $\initialmorph{A}$ 
is indeed a "transducer morphism", and we show instead that 
$\initialmorph{A}$ is the \emph{unique} possible "morphism" 
from $\initialtransd$ to $A$, essentially because the 
definition of $\initialmorph{A}$ was the only admissible one.
We prove this by considering another possible "morphism" 
$\initialmorphprime{A}: \initialtransd \rightharpoonup A$ and by verifying that 
$\initialmorph{A}(w,\initdata)=\initialmorphprime{A}(w,\initdata)$, 
using a simple induction on $|w|$.
For the base case $w=\emptystr$, we recall that 
$\delta_\init \comp \initialmorph{A} 
 = \kappainit
 = \deltainit \comp \initialmorphprime{A}$,
and hence 
\begin{align*}
    \initialmorph{A}(\emptystr,\initdata) 
    &= \initialmorph{A}(\deltainit(\initstate,\initdata)) \\
    &= \kappainit(\initstate,\initdata) \\
    &= \initialmorphprime{A}(\deltainit(\initstate,\initdata)) \\
    & = \initialmorphprime{A}(\emptystr,\initdata).
\end{align*}
For the inductive step, we assume that 
$\initialmorph{A}(w,\initdata) = \initialmorphprime{A}(w,\initdata)$, we recall that 
$\kappa_a \comp \initialmorph{A} = \delta_a \comp \initialmorph{A}$ and
$\kappa_a \comp \initialmorphprime{A} = \delta_a \comp \initialmorphprime{A}$, and we derive
\begin{align*}
    \initialmorph{A}(wa,\initdata) 
    &= \initialmorph{A}(\delta_a(w,\initdata)) \\
    &= \kappa_a(\initialmorph{A}(w,\initdata)) \\
    &= \kappa_a(\initialmorphprime{A}(w,\initdata)) \\
    &= \initialmorphprime{A}(\delta_a(w,\initdata)) \\
    &= \initialmorphprime{A}(wa,\initdata).
    \tag*{\qedhere}
\end{align*}
\end{proof}

%% file: final-transducer.tex

\subsection{The final transducer for the initial images}\label{sec:final-transducer}

The construction of a "final" "transducer" $\finaltransd$ is the most delicate
part, even after restricting to the sub-category of "images" of "initial@@arrow"
"morphisms". The goal is to extract from the "transduction" $\varphi$ some
"configurations" that ``cover'' the reachable "configurations" of every
"transducer" $A$ "realizing" $\varphi$, so as to obtain a "final morphism" 
from $\Reach(A)$ to $\finaltransd$. For this "morphism" to be unique,
the values stored in $\finaltransd$ must be as large as possible so as
to cover the values stored by any "equivalent" "transducer" at the 
corresponding "control states".

\AP
This brings us to one of the main difficulties of the "minimization" result.
Given a set $U$ of "updates", we need "factorizations" of the form
$u = g\comp v_u$ for every $u\in U$. In this case, the "update" $g$ is 
called a "common (inner) divisor" of $U$ and the "updates" $v_u$ are the 
corresponding "(outer) residuals". 
Moreover, $g$ should be a "\emph{greatest} common divisor@GCD", namely, 
every other "common divisor" $g'$ of $U$ must itself be a "divisor" of $g$. 
This requirement already appears in the minimization of
"sequential transducers" \cite{cho03}, where the longest common prefix of 
a set of "outputs" plays the role of a "greatest common divisor" in the 
"data structure" that manipulates strings by appending constants.

The additional difficulty here is that we work with an abstract, sorted monoid 
of "updates", not just with a free monoid. Thus, "greatest common divisors" need
not be unique, and different "common divisors" may induce different "residuals". 
A common way to address this issue is to assume that the monoid of
"updates" comes equipped with a canonical choice of "greatest common divisor"
and residuals. This approach is pursued, for instance, in
\cite{Maletti04,Gerdjikov18,Aristote25,Aristote24unpublished}, and, to some
extent, also in \cite{Eisner03}, where canonical factorizations of sets of
updates underpin both "minimization" and automata learning results.

Here we follow a different route. Instead of postulating canonical 
"factorizations" of "updates", whose choice may be non-obvious or unnatural 
in many "data structures", we use "strong factorizations" 
(Proposition~\ref{prop:images}) and the "diagonal fill-in property".
This allows us to restrict, without loss of generality, to "divisors" 
of "updates" that are "epis", and to show that "greatest common divisors" 
and their "residuals" ---whenever they exist--- behave canonically enough 
for the "minimization" construction to go through.



We recap the important notions that we are going to use.
\AP
We call ""vector"" a family $U$ of "updates" indexed by words $t\in\Sigma^*$.
These "updates" must share the same input "type" $\alpha$ and "output" "type" $\beta$,
and are thus of the form $U[t]: \bbD_\alpha \to \bbD_\beta$.
For convenience, we allow "vectors" to contain undefined entries,
which we denote by writing $U[t]=\bot$ for some $t\in\Sigma^*$.
We also tacitly assume that the set of indices $t$ where 
the entries of $U$ are defined is prefix-closed, namely, if
$U[t]\neq\bot$ and $t'$ is a prefix of $t$, then $U[t']\neq\bot$.
\AP
We lift "composition" from functions to "vectors" in a pointwise manner:
given an "update" $f$ and a "vector" $U$, where the "domains" of 
the "updates" $U[t]$ are all equal to the "codomain" of $f$, 
we let $""f\vcomp U @pointwise composition""$ be the "vector" 
defined by $(f \vcomp U)[t] = f\comp (U[t])$ for all $t\in\Sigma^*$, 
with the convention that $f\comp\bot=\bot$.
\AP
We call ""sub-vector"" of $U=(U[t])_{t\in\Sigma^*}$ a "vector"
of the form $\reintro{U[a-]} = (U[at])_{t\in\Sigma^*}$, for any
letter $a\in\Sigma$.

\AP
Given a "factorization" $u=g\comp v$, we say that $g$ is a ""divisor""
of $u$ and $v$ is a ""residual"" of $u$ via $g$. 
\AP
If $g$ is also an "epi", we call it an ""epi divisor"", and 
remark that in this case there is a \emph{unique} "residual" of $u$ via $g$,
which 
we denote by $\reintro{\residual*{g}{u}}$.
\AP
A ""common divisor"" of a "vector" $U$ is a "divisor" $g$ of every $u\in U$.
If $g$ is also "epi", it is called an ""epi common divisor"", and the 
""residual vector"" of $U$ via $g$ is the vector $\reintro[residual vector]{\vresidual*{g}{U}}$ 
defined by $(\vresidual{g}{U})[t] = \residual{g}{U[t]}$ for all $t\in\Sigma^*$.
\AP
A ""greatest common divisor"" (\reintro{GCD} for short) of $U$ is a "common divisor" 
of $U$ that is "divided" by every other "common divisor" of $U$. 
By convention, if $U[t]=\bot$ for all $t\in\Sigma^*$, then we let $\bot$ be its unique "GCD".
\AP
Similarly, we define an ""epi greatest common divisor"" (\reintro{EGCD}) of $U$ 
as an "epi common divisor" of $U$ that is "divided" by every other "epi common divisor" of $U$.

Note that an "EGCD" of $U$ is a special case of a "common divisor" of $U$, but it is not
necessarily greatest w.r.t.~all possible "common divisors" of $U$. 
In particular, the notions of "GCD" and "EGCD" are incomparable.
Despite this, we can still relate these two notions via the the following result, 
which exploits the existence of "strong factorizations" 
(Proposition \ref{prop:images}) and the "diagonal fill-in property":

\begin{lemma}[restate=GCDvsEGCD, name=\textbf{GCD vs EGCD}]%
	\label{lem:gcd-vs-egcd}
If $g$ is a "GCD" of a "vector" $U$, then there is a "strong factorization" $g=g'\comp g''$ of it,
where $g'$ is an "EGCD" of $U$.
\end{lemma}

\begin{proof}
Let $U$ be a "vector" and $g$ a "GCD" of $U$.
By applying Proposition \ref{prop:images} to single-state "transformations" 
---which correspond to "updates"---,
one obtains a "strong factorization" $g=g'\comp g''$.
Clearly, $g'$ is an "epi common divisor" of $U$. To prove that $g'$ is also "greatest@EGCD" among 
all "epi common divisors" of $U$, consider another "epi common divisor" $f$ of $U$.
Since $g$ is "greatest@GCD" among all "common divisors" of $U$, $f$ is a "divisor" of $g$,
so $g = f\comp f'$ for some $f'$. 
Since $f$ is an "epi" and $g''$ is a "strong mono", by the "diagonal fill-in property" 
there is a (unique) "update" $d$ such that $f\comp d = g'$. This shows that $f$ is also a "divisor"
of $g'$, and hence $g'$ is an "EGCD" of $U$.
\end{proof}

We now introduce our second assumption, to be proven later for specific classes of "updates":

\begin{assumption}[restate=GCDs,
                   name={\textbf{existence of GCDs}}]\label{ass:GCD}
Every set of "updates" of uniform input and output "types" admits a "GCD". 
\end{assumption}

Note that if Assumption \ref{ass:GCD} holds, then every "vector" 
has a "GCD", and by Lemma \ref{lem:gcd-vs-egcd} also an "EGCD".
\emph{The results that follow tacitly rely on this assumption, 
and in particular on the existence of "EGCDs" of "vectors".}

\begin{lemma}[restate=UniqueEGCD, name=\textbf{uniqueness of EGCDs}]%
	\label{lem:unique-egcd}
The "EGCDs" of any given "vector" $U$ are pairwise "isomorphic@@arrow".
\end{lemma}

\begin{proof}
Suppose that $g,g'$ are two "EGCDs" of the same "vector" $U$.
Clearly, they are "divisors" one of another, and in particular,
because $g,g'$ are "epis", they induce unique "updates" $j,j'$ 
such that $g'=g\comp j$ and $g=g'\comp j'$.
Let $d=j\comp j'$ and $d'=j'\comp j$, 
and observe that $g\comp d = g$ and $g'=g'\comp d'$.
However, because, once again, $g$ is an "epi",
there is only one "residual" of $g$ via $g$,
which is the "identity". So $d$ is the "identity", and similarly for $d'$. 
This shows that $j$ is an "isomorphism" from the "codomain" of $g$
to the "codomain" of $g'$, and $j'$ is its "inverse".
Moreover, because $g\comp j = g'$, we have that $g$ and $g'$
are "isomorphic@@arrow".
\end{proof}

By Lemma \ref{lem:unique-egcd}, for every two "EGCDs" $g,g'$ of $U$,
there is an "isomorphism@@arrow" $j$ such that $g\comp j=g'$.
In particular, $g,g'$ have the same "domain" and "isomorphic" "codomains".
Moreover, if $V = \vresidual*{g}{U}$ and $V' = \vresidual*{g'}{U}$
are the corresponding "residual vectors",
then, because $g\vcomp V = U = g'\vcomp V'$, the same $j$
above is also an "isomorphism@@arrow" between 
$V[t]$ and $V'[t]$, 
namely, $V[t] = j\comp V'[t]$ for all $t\in\Sigma^*$ .
\AP
We shall denote this latter property by writing $V \mathrel{""\viso*[j]""} V'$.

\begin{lemma}[restate=EGCDDistributivity, 
		  name=\textbf{distributivity of EGCDs}]\label{lem:egcd-distributivity}
If $U=g\vcomp V$, for some "epi" $g$ and some "vector" $V$, 
and $h$ is an "EGCD" of $V$, then $g\comp h$ is an "EGCD" of $U$.
\end{lemma}

\begin{proof}
We will reason with the following "EGCDs", which exist thanks to prior assumptions:
\begin{itemize}
\item $f$ is an "EGCD" of $U=g\vcomp V$, with $T=\vresidual{f}{U}$ associated "residual vector",
\item $h$ is an "EGCD" of $V$, with $W=\vresidual{h}{V}$ associated "residual vector".
\end{itemize}
We aim at proving that $g\comp h$ is an "EGCD" of $U$.

We start by observing that $U = g\vcomp V = g\comp h\vcomp W$, so $g\comp h$ is a
"common divisor" of $U$, and it is "epi" since $g$ and $h$ are "epis".
It remains to show that $g\comp h$ is a "\emph{greatest} epi common divisor@EGCD" of $U$.
Recalling that $f$ is an "EGCD" of $U$, we also derive that $g\comp h$ is a "divisor" of $f$.
Now, let $\ell = \residual{g\comp h}{f}$. Because every "residual" of an "epi" is an "epi",
(cf.~\cite[Proposition 7.41]{TheJoyOfCats}), $\ell$ is an "epi".
Next, from the equation
\[
    g\vcomp V ~=~ U ~=~ f\vcomp T ~=~ g\comp h\comp \ell \vcomp T,
\]
we can cancel $g$ from both sides, since it is "epi", and obtain
\[
    V ~=~ h\comp\ell \vcomp T.
\]
This shows that $h\comp\ell$ is a "divisor" of $V$, and it is "epi" because $h$ and $\ell$ are "epi".
Because $h$ is an "EGCD" of $V$, we deduce that $h\comp\ell$ is a "divisor" of $h$, and by pre-"composing" with $g$, 
we get that $g\comp h\comp\ell$ ($=g\comp h\comp \residual{(g\comp h)}{f} = f$) 
is a "divisor" of $g\comp h$.
For the last stretch, every "epi common divisor" of $U$ is a "divisor" of the "EGCD" $f$ of $U$,
and hence by transitivity it is also a "divisor" of $g\comp h$.
Since we have already established that $g\comp h$ is an "epi common divisor" of $U$,
this shows that $g\comp h$ is an "EGCD" of $U$.
\end{proof}

\begin{corollary}[restate=EGCDsOfSubvectors, 
		  name=\textbf{EGCDs of sub-vectors}]\label{cor:egcds-of-subvectors}
Let $g$ is an "EGCD" of $U$ and let $V=\vresidual{g}{U}$ be the associated "residual vector".
For every $a\in\Sigma$, if $h$ is an "EGCD" of the "sub-vector" $V[a-]$, then
$g\comp h$ is an "EGCD" of the "sub-vector" $U[a-]$.
\end{corollary}

\begin{proof}
Let $g$ is the "EGCD" of $U$, $V=\vresidual{g}{U}$, 
$U'=U[a-]$, $V'=V[a-]$, and let $h$ be the "EGCD" of $V'$.
We have that $g$ is "epi" and $U'=g\vcomp V'$,
and so, by Lemma \ref{lem:egcd-distributivity},
$g\comp h$ is the "EGCD" of $U'$.
\end{proof}

We can now turn towards the construction of the "final transducer".
\AP
We shall work with several "vectors" at once, which are conveniently
represented by ""matrices"", namely, families 
$H=(H[s,t])_{s,t\in\Sigma^*}$ of "updates" (or undefined entries) 
indexed by pairs of words $s,t\in\Sigma^*$.
\AP
Each row $H[s,-]$ of $H$ is a "vector", and thus, 
the "domains" and "codomains" of its "updates" are all the same. 
A column of $H$ is not necessarily a "vector", namely, 
its "updates" may have different "domains" or "codomains": 
we call it a ""column vector"".

\AP
Specifically, we shall consider the ""Hankel matrix"" $\reintro{\Hankel*}$
of the "transduction" $\varphi$, whose entries are defined as 
$\Hankel[s,t] = \varphi(s t)$ for every $s,t\in\Sigma^*$,
where $\varphi(s t)$ is seen as an "update" from the initial "data type" 
$\bbD_{\init*}=\{\initdata\}$ to the final "data type" $\bbD_{\halt*}$.
The "Hankel matrix" is nothing but a convenient way to describe
a "transduction", and eventually extract "EGCDs" and corresponding
"residual vectors".
Also observe that every "sub-vector" $\Hankel[s,a-]$ of a row of the 
"Hankel matrix" coincides with a ``successor'' row $\Hankel[sa,-]$.

Next, we apply Assumption \ref{ass:GCD} to each row $\Hankel[s,-]$ of 
the "Hankel matrix", obtaining an "EGCD" $g_s$ and a corresponding
"residual vector" $V_s=\vresidual{g}{\Hankel[s,-]}$ for every $s\in\Sigma^*$.
We organize these objects into 
a "column vector" $\Divisors=(g_s)_{s\in\Sigma^*}$ and 
a new "matrix" $\Residuals=(V_s[t])_{s,t\in\Sigma^*}$. 
$\Divisors$ and $\Residuals$ are called, respectively, 
the \reintro{divisor vector} and the \reintro{residual matrix} of $\Hankel$.
The $\viso$-equivalence classes of the rows of the "residual matrix" $\Residuals$
will represent the "control states" of our "final transducer", and
can be equally seen as the classes of a one-sided Myhill-Nerode equivalence.
In a similar way, the "divisor vector" $\Divisors$ will be used 
to define the "transformations" of the "final" "transducer".
These constructions crucially rely on suitable compatibility 
properties that are reminiscent of the right-invariance property of the classical 
Myhill-Nerode equivalence, and that are established in Lemma \ref{lem:invariance} below.
\AP
\phantomintro{\Divisors*}
\phantomintro{\Partials*_a}
\phantomintro{\Residuals*}

\begin{lemma}[restate=Invariance, 
		  name={\textbf{right-invariance}}]\label{lem:invariance}
Under Assumption \ref{ass:GCD}, one can find some 
"column vectors" $\reintro{\Divisors*}$ and $\reintro{\Partials*_a}$ 
(one for every letter $a\in\Sigma$), and a 
"matrix" $\reintro{\Residuals*}$ 
such that, for all $s\in\Sigma^*$ and $a\in\Sigma$:
\begin{itemize}
\item $\Divisors[s]$ is an "EGCD" of $\Hankel[s,-]$ and 
	$\Residuals[s,-] = \vresidual{\Divisors[s]}{\Hankel[s,-]}$,
\item $\Partials_a[s]$ is an "EGCD" of $\Residuals[s,a-]$ and 
	$\Residuals[sa,-] = \vresidual{\Partials_a[s]}{\Residuals[s,a-]}$,
\item $\Divisors[sa] = \Divisors[s] \comp \Partials_a[s]$.	
\end{itemize}
In particular, $\Divisors[s]$ and $\Residuals[s,-]$
are uniquely determined from $\Hankel[s,-]$, 
and $\Partials_a[s]$ is uniquely determined from $\Residuals[s,-]$. 
\end{lemma}

\begin{proof}
We construct our objects $\Divisors[s]$, $\Partials_a[s]$, and $\Residuals[s,-]$ 
by exploiting an induction on $s\in\Sigma^*$, 
while preserving the properties stated in the lemma.

For the base case of the induction, we only define $\Divisors[\emptystr]$
and $\Residuals[\emptystr,-]$, as follows: we exploit 
the existence of "EGCDs" (which follows from Assumptions \ref{ass:compactness} 
and \ref{ass:GCD} and from Lemma \ref{lem:gcd-vs-egcd})
and we let $\Divisors[\emptystr]$ be an "EGCD" of $\Hankel[\emptystr,-]$
and $\Residuals[\emptystr,-] = \vresidual{\Divisors[\emptystr]}{\Hankel[\emptystr,-]}$. 
This trivially satisfies the first property of the lemma for $s=\emptystr$.

For the inductive step, we consider a word $s\in\Sigma^*$
and a letter $a\in\Sigma$, and we define first $\Partials_a[s]$
and $\Residuals[sa,-]$, and then $\Divisors[sa]$. 
In doing so we assume as inductive hypothesis that
$\Divisors[s]$ and $\Residuals[s,-]$ are already defined.
We then define $\Partials_a[s]$ as an "EGCD" of the "sub-vector"
$\Residuals[s,a-]$, and we let $\Residuals[sa,-] = \vresidual{\Partials_a[s]}{\Residuals[s,a-]}$.
These definitions clearly satisfy the second property of the lemma.
We also let $\Divisors[sa] = \Divisors[s] \comp \Partials_a[s]$,
so as to satisfy the third property. 
By applying Corollary \ref{cor:egcds-of-subvectors} with
$U=\Hankel[s,-]$,
$g=\Divisors[s]$,
$V=\Residuals[s,-]$, and
$h=\Partials_a[s]$, 
we obtain that
$g\comp h = \Divisors[sa]$ is an "EGCD" of 
$U[a-] = \Hankel[s,a-] = \Hankel[sa,-]$,
and in particular we have
\begin{align*}
    \Divisors[sa] \vcomp \Residuals[sa,-] 
    & ~=~ \Divisors[s] \comp \Partials_a[s] \vcomp \Residuals[sa,-] 
        \tag{by definition of $\Divisors[sa]$} \\
    & ~=~ \Divisors[s] \vcomp \Residuals[s,a-]
        \tag{by definition of $\Partials_a[s]$} \\
    & ~=~ \Hankel[s,a-] = 
        \tag{by inductive hypothesis} \\
    & ~=~ \Hankel[sa,-].
        \tag{by definition of Hankel matrix}
\end{align*}
This shows that the first property is preserved
during the induction step, and completes the proof.
\end{proof}

We conclude with the construction of the "final transducer" $\finaltransd$, 
which provides the last ingredient for the existence of "algebraically minimal" 
"transducers":

\phantomintro{\finaltransd*}
\begin{proposition}[restate=FinalTransducer, 
			  name={\textbf{final transducer}}]\label{prop:final-transducer}
There is a "transducer" $\reintro{\finaltransd*}$
that is "final" in the sub-category of "images" of "initial@@arrow" "morphisms", 
namely, for every "transducer" $A$ that "realizes" $\varphi$,
there is a unique "morphism" $\finalmorph{\Reach*(A)}: \Reach(A) \to \finaltransd$.
\end{proposition}

\begin{proof}
We shall use the "column vectors" $\Divisors$ and $\Partials_a$ and the
"residual matrix" $\Residuals$ provided in Lemma \ref{lem:invariance}.
\AP
We will only consider the "vectors" of $\Residuals$ that are not constantly
$\bot$; we call these vectors ""non-trivial"".
\AP
Because we need to distinguish the "residual vectors" only
up to the equivalence $\viso$ (defined after Lemma \ref{lem:unique-egcd}), 
it is convenient to fix a ""representative"" $\repr{w}$ for all the words $w\in\Sigma^*$
that induce "non-trivial" "residual vectors" in the same $\viso$-class.
\AP
We also fix an "isomorphism" $""\isov*{w}""$ witnessing 
$\Residuals[w,-] \viso[\isov*{w}] \Residuals[\repr{w},-]$.

\smallskip
\noindent
\underline{\emph{The final transducer}.}~
We define the "transducer" 
$\finaltransd = (\space{Q},\deltainit,(\delta_a)_{a\in\Sigma},\deltahalt)$,
as follows: 
\begin{itemize}
\item $Q$ consists of the "representatives" $\repr{w}$, now seen as "control states"
      with associated "types";
      more precisely, the "type" $\type(\repr*{w})$ associated with each "state" $\repr{w}$
      is the same as the "type" of the "domain" of any "update"
      $\Residuals[\repr{w},t]$, provided it is not $\bot$
      --- recall that the non-$\bot$ "updates" of a "vector" have all the same "domain";
      in particular, we have $\Data(\repr{w}) = \Dom(\Residuals[\repr{w},\emptystr])$;
\item $\deltainit$ maps $(\initstate,\initdata)$ to $(\repr{\emptystr},d)$, 
    where $d=\Divisors[\repr{\emptystr}](\initdata)$, 
    provided that $\Residuals[\repr{\emptystr},-]$ is "non-trivial"
    --- note that 
    $\Divisors[\repr{\emptystr}] \vcomp \Residuals[\repr{\emptystr},-] = \Hankel[\repr{\emptystr},-]$,
    and hence $\Divisors[\repr{\emptystr}]$ is an "update" from the 
    singleton "domain" $\bbD_{\init*} = \{\initdata\}$ to the "domain" $\Data(\repr{\emptystr})$;
\item for every $a\in\Sigma$, $\delta_a$ maps any "configuration" 
    $(s,d)$ to the "configuration" $(\repr{sa},d')$, 
    where $d' = (\Partials_a[s]\comp \isov{sa})(d)$, provided that 
    $\Residuals[\repr{sa},-]$ is a "non-trivial" "residual vector"
    (otherwise, $\delta_a$ is undefined on $(s,d)$)
    --- note that the target "configuration" $(\repr{sa},d')$ is 
    well-defined because, by Lemma \ref{lem:invariance},
    $\Partials_a[s]$ is uniquely determined from $\Residuals[s,-]$ and $a$
    and $\isov{sa}$ is, by construction, determined by $sa$;
\item $\deltahalt$ maps any "configuration" $(s,d)$ to $(\haltstate,d')$, 
    where $d' = \Residuals[s,\emptystr](d)$
    --- note that, because $\Residuals[s,-]$ is "non-trivial", 
    $\Residuals[s,\emptystr](d)$ is defined.
\end{itemize}
A simple induction on the length of $w$ shows that,
after reading $w$, the "transducer" $\finaltransd$ reaches
the "configuration" 
\[
  \big(\repr{w}, \: \Divisors[\repr{w}](\initdata)\big)
\]
provided that $\Residuals[\repr{w},-]$ is a "non-trivial" "residual vector"
(otherwise, $\finaltransd$ has no run on $w$).
By construction, $\isov{w}$ is an "isomorphism" such that
$\Residuals[w,-] = \isov{w} \vcomp \Residuals[\repr{w},-]$.
Similarly, one verifies by induction that 
$\Divisors[\repr{w}] = \Divisors[w]\comp \isov{w}$.
Thus, the "output" produced by $\finaltransd$ after reading $w$ is 
\[
  \big( \underbrace{\Divisors[\repr{w}]}_{\Divisors[w]\comp \isov{w}} \:\comp 
        \underbrace{\Residuals[\repr{w},\emptystr]}_{\isov{w}^{-1}\comp\Residuals[w,\emptystr]} \big) (\initdata)
  ~=~ \big( \Divisors[w]\comp\Residuals[w,\emptystr] \big) (\initdata)
  ~=~ \Hankel[w,\emptystr] 
  ~=~ \varphi(w).
\]
This shows that $\finaltransd$ "realizes" exactly the "transduction" $\varphi$.

\smallskip
\noindent
\underline{\emph{Existence of final morphism}.}~
We show that $\finaltransd$ is indeed "final" for the sub-category of "initial" "images".
Consider an arbitrary "transducer" $A$ that also "realizes" $\varphi$, and let 
$\initialmorph{A}: \initialtransd \to A$ be the unique "initial@@arrow" "morphism" towards $A$
(Proposition \ref{prop:initial-transducer}). 
Further let $\initialtransd \epi[e] B \smono[m] A$ be a "strong factorization" of $\initialmorph{A}$
(Proposition \ref{prop:images}), and let $(\hat e,\check e)$ be the "specification" of the "epi" "morphism" $e$.
We also need to denote the components of the "initial transducer" $\initialtransd$ 
and those of the "initial image" $B$, so we let 
$\initialtransd = (\Sigma^*,\kappainit,(\kappa_a)_{a\in\Sigma},\kappahalt)$ and
$B = (Q_B,\chiinit,(\chi_a)_{a\in\Sigma},\chihalt)$.

We disclose some important consequences of the fact that $e$ is "epi":
\begin{enumerate}
\item Since $e$ is "epi", by Lemma \ref{lem:epis}, 
      $\speca e$ is a surjective partial function, and hence for 
      every $q'\in Q_B$, the set $\speca{e}^{-1}(q')$ is non-empty
      (this set contains "control states" of $\initialtransd$, which 
      are words over $\Sigma$).
      Hereafter, $s$ is tacitly assumed to be an arbitrary word from $\speca e^{-1}(q')$
      that induces a "non-trivial" "residual vector" $\Residuals[s,-]$.
\item For every $t\in\Sigma^*$, the "output" $\varphi(st)$ of $B$, 
      seen as an "update" from the singleton "domain" $\bbD_{\init*}$, 
      "factorizes" as $f_s\comp h_t$, where
      $f_s=\specb{\chi_{\init s}}(\initstate)$
      and
      $h_t=\specb{\chi_{t\halt}}(q')$,
      namely, $f_s$ (resp.~$h_t$) is the "update" induced by 
      $\init s$ (resp.~$t\halt$) in $B$ starting from $\initstate$ (resp.~$q'$).
      Since $e$ is a "morphism" from $\initialtransd$ to $B$,
      we also have $f_s=\specb{e}(s)$.
\item Let $V_{q'}=(h_t)_{t\in\Sigma^*}$. 
      Because $B$ computes the "transduction" $\varphi$, we have
      \[
          \Hankel[s,-] 
          ~=~ \big(\varphi(st)\big)_{t\in\Sigma^*}
          ~=~ f_s\vcomp V_{q'}.
      \]
      Moreover, by Lemma~\ref{lem:invariance},
      $\Divisors[s]$ is an "EGCD" of $\Hankel[s,-]$ and $\Residuals[s,-]$ 
      is the corresponding "residual vector". We thus have 
      \[
          \Residuals[s,-] ~=~ \residual{\Divisors[s]}{(f_s\comp V_{q'})}.
      \]
\item Now, let $g_s$ be an "EGCD" of the "vector" $V_{q'}$, and recall
      from Lemma~\ref{lem:unique-egcd} that any two "EGCDs" of $V_{q'}$
      are "isomorphic@@arrow". 
      In particular, since $V_{q'}$ depends only on $q'$, for all $s,s'\in\speca e^{-1}(q')$, 
      there exists an "isomorphism@@arrow" $i_{s,s'}$ between the "codomains"
      of $g_s$ and $g_{s'}$ such that
      \[
        g_s ~=~ g_{s'}\comp i_{s,s'}.
      \]
\item Let $W_{s,q'} = \residual{g_s}{V_{q'}}$ be the "residual vector" of $V_{q'}$ via the "EGCD" $g_s$.
      We have
      \[
          \Hankel[s,-] 
          ~=~ f_s \vcomp V_{q'} 
          ~=~ f_s \vcomp W_{s,q'}
      \]     
      and similarly for $s'\in\speca e^{-1}(q')$.
      The "isomorphism@@arrow" $i_{s,s'}$ witnessing $g_s \viso g_{s'}$
      also induces a $\viso$-equivalence between the corresponding residual vectors:
      \[
        \Residuals[s,-] \viso[i_{s,s'}] \Residuals[s',-].
      \]
      In particular, all the "vectors" $\Residuals[s,-]$ for $s\in\speca e^{-1}(q')$
      lie in the same $\viso$-class.
\end{enumerate}
We can now "specify" the "final@@arrow" "morphism" $\finalmorph{B}: B \to \finaltransd$.
For the mapping $\specafinalmorph{B}$ from the "states" 
of $B$ to the "states" of $\finaltransd$, we let
\[
    \specafinalmorph{B}(q') ~=~ 
    \begin{cases}
        \repr{w} & \text{if there is $w\in\speca e^{-1}(q')$ with $\Residuals[w,-]$ "non-trivial"} \\
        \text{undefined} & \text{otherwise.}
    \end{cases}
\]
Note that this is well-defined because by Property (5) above,
all words in $\speca e^{-1}(q')$ induce pairwise $\viso$-equivalent 
"residual vectors", and so have the same "representative".
As for the mapping $\specbfinalmorph{B}$ from "states" of $B$ to "updates",
we recall from Properties (2--4) above that, for every $s\in\speca e^{-1}(q')$, 
$g_s$ is an "EGCD" of the "vector" $V_{q'}=(h_t)_{t\in\Sigma^*}$,
which contains the "updates" $h_t=\specb{\chi_{t\halt*}}(q')$ induced by $B$ 
starting from $q'$.
Accordingly, we let 
\[
    \specbfinalmorph{B}(q') ~=~ 
    \begin{cases}
        g_{\repr{w}} & \text{if there is $w\in\speca e^{-1}(q')$ with $\Residuals[w,-]$ "non-trivial"} \\
        \text{undefined} & \text{otherwise}
    \end{cases}
\]
(this is again well-defined thanks to Property (5)).
By construction, the "transformation" $\finalmorph{B}: \States{Q_B}\to\States{Q}$ 
"specified by" the pair $(\specafinalmorph{B},\specbfinalmorph{B})$ satisfies the following property:
\begin{align*}
    \rightward{\text{if}}\phantom{\text{then }}   
    & \quad e(w,\initdata) = (q',d) \text{ and } 
      \Residuals[\repr{w},-] \text{ is "non-trivial"} \\
    \text{then } 
    & \quad \finalmorph{B}(q',d) = (\repr{w}, \: g_{\repr{w}}(d)). 
    \tag{$\star$}
\end{align*}
Below, we show that $\finalmorph{B}$ is also a "transducer morphism".
\LIPICSorARXIV{
We need to show that the solid parts of the 
diagrams in Figure \ref{fig:final_morphism}
commute (as for the dashed parts, 
we already know they commute since $e$ is a "morphism" 
from $\initialtransd$ to $B$).
\begin{figure*}[t]
\(
\input{fig-final-morphism}
\)
\caption{Diagrams describing the "final@@arrow" "morphism" $\finalmorph{B}: B \to \finaltransd$.}
\label{fig:final_morphism}
\end{figure*}
}{
We need to show that the solid parts of the 
diagrams below commute (as for the dashed parts, 
we already know that they commute since $e$ is a 
"transducer morphism" from $\initialtransd$ to $B$):
\[
\input{fig-final-morphism}
\]
}
We are going to exploit Property $(\star)$ above 
and the fact that $e$ is an "epi".
More precisely, for the left diagram, we have:
\begin{align*}
    \chiinit \:\comp \finalmorph{B} 
    & ~=~ \kappainit \:\comp e \:\comp \finalmorph{B} 
        \tag{since $e: \initialtransd \to A$} \\
    & ~=~ \deltainit.
        \tag{by ($\star$) and by definition of $\finaltransd$}
\end{align*}
For the middle diagram, we have:
\begin{align*}
    e \:\comp \chi_a \:\comp \finalmorph{B} 
    & ~=~ \kappa_a \:\comp e \:\comp \finalmorph{B} 
        \tag{since $e: \initialtransd \to A$} \\
    & ~=~ e \:\comp \finalmorph{B} \:\comp \delta_a
        \tag{by ($\star$) and by definition of $\finaltransd$} \\
    \text{and hence}\quad
    \chi_a \,\comp \finalmorph{B} 
    & ~=~ \finalmorph{B} \,\comp \delta_a.
        \tag{because $e$ is "epi"}
\end{align*}
Similarly, for the right diagram, we have:
\begin{align*}
    e \:\comp \chihalt
    & ~=~ \kappahalt
        \tag{since $e: \initialtransd \to A$} \\
    & ~=~ e \:\comp \finalmorph{B} \:\comp \deltahalt
        \tag{by ($\star$) and by definition of $\finaltransd$} \\
    \text{and hence}\quad
    \chihalt 
    & ~=~ \finalmorph{B} \,\comp \deltahalt.
        \tag{because $e$ is "epi"}
\end{align*}
This shows that $\finalmorph{B}$ is a "transducer morphism" from $B$
to $\finaltransd$. 

\smallskip
\noindent
\underline{\emph{Uniqueness of final morphism}.}~
Finally, we argue that $\finalmorph{B}$ is the unique possible
"transducer morphism" from $B$ to $\finaltransd$. 
Consider another possible "morphism"
$\finalmorphprime{B}: B \to \finaltransd$.
A simple induction on $w\in\Sigma^*$ shows that
$(e \,\comp \finalmorph{B})(w,\initdata) = (e \,\comp \finalmorphprime{B})(w,\initdata)$.
This shows that $e \,\comp \finalmorph{B} = e \,\comp \finalmorphprime{B}$,
and thus, since $e$ is "epi", $\finalmorph{B} = \finalmorphprime{B}$.
\end{proof}

%% file: fig-final-morphism.tex
\begin{tikzpicture}[baseline=(current bounding box.center)]
\matrix (M) [matrix of nodes, column sep=12mm, row sep=10mm] {
& 
|[gray,name=space1pp]| $\States*{\Sigma^*}$ && 
|[gray,name=space2pp]| $\States*{\Sigma^*}$ & 
|[gray,name=space3pp]| $\States*{\Sigma^*}$ && 
|[gray,name=space4pp]| $\States*{\Sigma^*}$ & 
\\
|[name=init]| $\States*{\{\initstate*\}}$ & 
|[nicecyan,name=space1]| $\States*{Q_B}$ && 
|[nicecyan,name=space2]| $\States*{Q_B}$ & 
|[nicecyan,name=space3]| $\States*{Q_B}$ && 
|[nicecyan,name=space4]| $\States*{Q_B}$ & 
|[name=halt]| $\States*{\{\haltstate*\}}$ 
\\
& 
|[nicered,name=space1p]| $\States*{Q}$ && 
|[nicered,name=space2p]| $\States*{Q}$ & 
|[nicered,name=space3p]| $\States*{Q}$ && 
|[nicered,name=space4p]| $\States*{Q}$ & 
\\
};
\path[gray] (init) edge [partial arrow,dashed] node [above=1mm, pos=0.4] {$\kappainit$} (space1pp);
\path[gray] (space2pp) edge [partial arrow,dashed] node [pos=0.4,above=1mm, pos=0.4] {$\kappa_a$} (space3pp);
\path[gray] (space4pp) edge [partial arrow,dashed] node [pos=0.4,above=1mm, pos=0.4] {$\kappahalt$} (halt);
\path[nicecyan] (init) edge [partial arrow] node [above=0.5mm, pos=0.4] {$\chiinit$} (space1);
\path[nicecyan] (space2) edge [partial arrow] node [pos=0.4,above=1mm, pos=0.4] {$\chi_a$} (space3);
\path[nicecyan] (space4) edge [partial arrow] node [pos=0.4,above=0.5mm, pos=0.4] {$\chihalt$} (halt);
\path[nicered] (init) edge [partial arrow] node [below=1mm, pos=0.4] {$\deltainit$} (space1p);
\path[nicered] (space2p) edge [partial arrow] node [pos=0.4,below=1mm, pos=0.4] {$\delta_a$} (space3p);
\path[nicered] (space4p) edge [partial arrow] node [pos=0.4,below=1mm, pos=0.4] {$\deltahalt$} (halt);
\path (space1) edge [partial arrow] node [right=2mm, pos=0.4] {$\finalmorph{B}$} (space1p);
\path (space2) edge [partial arrow] node [left=2mm, pos=0.4] {$\finalmorph{B}$} (space2p);
\path (space3) edge [partial arrow] node [right=2mm, pos=0.4] {$\finalmorph{B}$} (space3p);
\path (space4) edge [partial arrow] node [left=2mm, pos=0.4] {$\finalmorph{B}$} (space4p);
\path[gray] (space1pp) edge [epi arrow,dashed] node [right=2mm, pos=0.4] {$e$} (space1);
\path[gray] (space2pp) edge [epi arrow,dashed] node [left=2mm, pos=0.4] {$e$} (space2);
\path[gray] (space3pp) edge [epi arrow,dashed] node [right=2mm, pos=0.4] {$e$} (space3);
\path[gray] (space4pp) edge [epi arrow,dashed] node [left=2mm, pos=0.4] {$e$} (space4);
\end{tikzpicture}

%% file: minimization-theorems.tex

\subsection{Minimization results and effectiveness}\label{sec:main-theorem}

A first minimization result now follows directly from
Propositions \ref{prop:minimal}, \ref{prop:images}, 
\ref{prop:initial-transducer}, and \ref{prop:final-transducer}:

\begin{theorem}[restate=MinimalTransducer,
		    name={\textbf{minimal transducer}}]\label{thm:minimal-transducer}
If $\bbD$ is a "data structure" with "constrained domains" that satisfies 
Assumption \ref{ass:GCD},
then every "transduction" "realized" by a "transducer" over $\bbD$
is also "realized" by one that is "algebraically minimal".
\end{theorem}

We can also prove an analogous result for the subclass of \emph{"finite@@transducer"} "transducers",
using the additional assumption that all "constrained domains" are "finitary":

\begin{theorem}[restate=MinimalFiniteTransducer,
                name={\textbf{minimal finite transducer}}]\label{thm:minimal-finite-transducer}
If $\bbD$ is a "data structure" with "constrained domains" that satisfies 
both Assumptions \ref{ass:compactness} and \ref{ass:GCD},
then every "transduction" "realized" by a "finite transducer" over $\bbD$
is also "realized" by one that is "finite@@transducer" and "algebraically minimal".
\end{theorem}

\begin{proof}
First of all, we claim that if a "transduction" is "realized" by some "transducer" $A$
with finitely many "control states", then every "algebraically minimal"
"transducer" $\finaltransd$ also has finitely many "control states".
This is because $\Reach(A)$ is a "subobject" of $A$, hence it has no more 
"control states" than $A$, and, similarly, $\Obs(\Reach*(A))$ is a "quotient" 
of $\Reach(A)$, hence it has no more "control states" than $\Reach(A)$.
So, finiteness of "control states" transfers from $A$ to $\Obs(\Reach*(A)) \iso \finaltransd$.
As for the "constrained domains" that may appear inside $\finaltransd$,
these are finitely presentable under the additional Assumption \ref{ass:compactness}.
\end{proof}

Moreover, we can show that Assumption~\ref{ass:GCD} is also necessary for the
existence of "algebraically minimal" models, at least over "standard" "data structures".
\AP
By ""standard data structure"" we mean a "data structure" where every
"domain" $\bbD_\alpha$ and every set of "updates" $\bbD_\alpha^\beta$ 
is countable, and such that every "common divisor" of $\bbD_\alpha^\beta$ 
is "isomorphic@@arrow" to the "identity" on $\bbD_\alpha$.
\AP
When the latter condition holds, we also say that the "updates" in $\bbD_\alpha^\beta$ 
are ""jointly coprime"". 
Most "data structures" ---in particular all examples considered here--- 
are "standard". For instance, the "updates" in $\bbD_2^1$ of the 
"free term algebra" $\bbD$ are "jointly coprime", since 
already for $u_1: (x,y)\mapsto a(x,y)$ and $u_2: (x,y)\mapsto b(x,y)$
the only "common divisors" are the "identity" $(x,y)\mapsto(x,y)$ and the swap $(x,y)\mapsto(y,x)$.

\begin{theorem}[restate=GCDNecessary,
                name={\textbf{necessity of GCDs}}]\label{thm:gcd-necessary}
Assumption \ref{ass:GCD} holds over any "standard" "data structure" $\bbD$
whenever, for every "transduction" $\varphi$, the class of "transducers" 
over $\bbD$ "realizing" $\varphi$ contains an "algebraically minimal" model.
\end{theorem}

\begin{proof}
We fix a "standard data structure" $\bbD$ and assume that
every class of "transducers" over $\bbD$ "realizing" a certain
"transduction" contains an "algebraically minimal model".
Towards proving that $\bbD$ satisfies Assumption \ref{ass:GCD},
we consider some input and output "types" $\alpha$ and $\beta$
and a set $U\subseteq\bbD_\alpha^\beta$ of "updates". 

We exploit the fact that both sets $\bbD_\alpha$ and $\bbD_\alpha^\beta$
are countable to define string-based encodings of the elements of these
sets. Formally, we enumerate $\bbD_\alpha$ as $\{d_1,d_2,\dots\}$
and we encode each "data" $d_i$ by the string $\tilde d_i = \bullet^i$,
where $\bullet$ is a special symbol from the underlying alphabet.
Similarly, we enumerate $\bbD_\alpha^\beta$ as $\{u_1,u_2,\dots\}$ 
and we encode each "update" $u_j$ by the string $\bullet^j$.
We then assume that the underlying alphabet contains three other symbols, $\qsep,\usep,\dsep$,
and define the "transduction" 
$\varphi: \{\bullet,\qsep,\usep,\dsep\}^* \to \bbD_\beta$
such that:
\begin{itemize}
    \item $\varphi$ maps every string of the form $\bullet^i \qsep \usep \bullet^j$,
          with $i,j\in\bbN$, to $u_j(d_i)$;
    \item $\varphi$ maps every string of the form $\bullet^i \qsep \dsep \bullet^j$,
          with $i,j\in\bbN$ and $u_j\in U$, to $u_j(d_i)$;
    \item $\varphi$ is undefined on all other inputs, in particular on inputs $\bullet^i \qsep \dsep \bullet^j$
          such that $u_j\nin U$.
\end{itemize}

Now, for every "common divisor" $f$ of $U$, we construct a "transducer" $A_f$
that "realizes" $\varphi$, as shown in Figure \ref{fig:common-divisor-transducer}
(some "identity" "updates" along transitions are omitted for readability).
In particular, after reading a prefix $\tilde d\qsep$, $A_f$ 
reaches a distinguished "state" $q_f$ storing the data $d$ encoded by the prefix. 
From $q_f$, a $\usep$-labelled transition applies the "identity" 
on $\bbD_\alpha$ and reaches a component that consumes suffixes $\tilde u$, 
eventually producing $u(d)$. Similarly, from $q_f$, a $\dsep$-labelled transition 
applies the "update" $f$ and reaches a component that consumes suffixes $\tilde u$, 
eventually producing $u(d)$ only when $u\in U$.

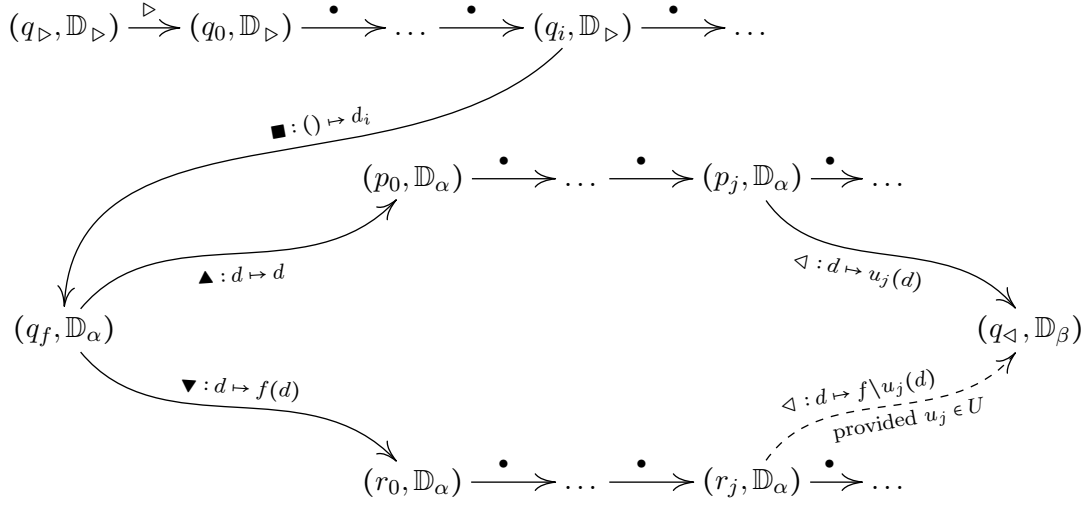
\begin{figure*}[t]
\begin{align*}
\input{fig-common-divisor-transducer}    
\end{align*}
\caption{The "transducer" $A_f$ that emits a "common divisor" $f$ of $U$ upon reading $\dsep$.}
\label{fig:common-divisor-transducer}
\end{figure*}

Next, we exploit the existence of an "algebraically minimal" "transducer" 
$\finaltransd$ "realizing" the same "transduction" $\varphi$, and from which
we will be able to extract a "GCD" of $U$.
Since each $A_f$ is pruned by design, i.e.~$A_f=\Reach(A_f)$, 
$\finaltransd$ is a "quotient" of $A_f$, and hence there is a 
"morphim" $h$ from $A_f$ to $\finaltransd$.
In particular, this "morphism" maps the unique $\dsep$-labelled transition of $A_f$ 
to a unique $\dsep$-labelled transition of $\finaltransd$, whose "update" we denote by $g$. 
Clearly, this forces $g$ to factor through every "common divisor" $f$ of $U$.

It remains to see that $g$ is a "common divisor" of $U$. 
For this, we show that the runs of $\finaltransd$ labelled by $\usep\tilde u$, for all $u\in\bbD_\alpha^\beta$, 
ensure that no non-trivial "common divisor" can be accumulated before reading $\dsep$.
Formally, we define the accumulated "update" before $\dsep$ as the "update" $g_0$ that maps
any data $d\in\bbD_\alpha$ to the data of the "configuration" reached by $\finaltransd$ 
after reading the corresponding prefix $\tilde d\qsep$.
By the commutativity properties of "transducer morphisms", we know that $g_0$
must coincide with $\specb h(q_f)$, namely, the "update" that is applied by the "morphism" $h$ 
when applied to "configurations" of $A_f$ with $q_f$ as "control state".
Moreover, by construction, $g_0$ must "divide" all the "updates" $u\in\bbD_\alpha^\beta$ 
induced by the possible continuations $\usep\tilde u$. Because the latter "updates"
are "jointly coprime", we derive that $g_0$ is "isomorphic@@arrow" to the "identity" over $\bbD_\alpha$.
Finally, knowing that $g_0\comp g$ must be a "common divisor" of $U$, we conclude that
$g$ is also a "common divisor" of $U$, and hence a "GCD" of $U$.
\end{proof}

Finally, it is natural to ask whether an "algebraically minimal" "transducer"
$\finaltransd$ can be \emph{computed} from a given "finite transducer" $A$.
Here there is no general recipe: computability depends on the class of "transducers"
at hand and on suitable effectiveness assumptions. Nevertheless, the proof of
Theorem~\ref{thm:minimal-transducer}, in particular Propositions~\ref{prop:images}
and~\ref{prop:final-transducer}, suggests the following general procedure.
\begin{enumerate}
\item For each "state" $q$ of
      $A=(Q,\deltainit,(\delta_a)_{a\in\Sigma},\deltahalt)$, compute a "GCD"
      $f_q$ of the "vector" $U_q$ of "updates" induced by maximal runs that 
      depart from $q$, namely, $U_q[t] = \specb{\delta}_{t\halt}(q)$ for all $t\in\Sigma^*$.
      This can be done, assuming an effective version of Assumption~\ref{ass:GCD},
      by viewing the $f_q$'s as unknowns in the system of equations
      \[
            f_q ~=~ \mathop{GCD}\big( \big\{ \specb{\delta}_{\halt*}(q) \big\} \cup
                                      \big\{ \specb{\delta}_a(q) \comp f_{q'}
                                             \::\: a\in\Sigma, \: q'=\speca{\delta}_a(q) \big\} \big)
            \tag{$\star$}
      \]
      and by solving it via a greatest fixpoint computation. 
      More precisely, one starts by setting all variables $f_q$ to the undefined 
      "update" $\bot$, and repeatedly assigns each variable $f_q$ the "update"
      in the right-hand side of~($\star$). Termination is guaranteed when the 
      "divisor" preorder is well-founded, since the "updates" $f_q$ can
      only decrease w.r.t.~the "divisor" preorder.
\item Use Lemma~\ref{lem:gcd-vs-egcd} to turn the "GCDs" $f_q$ into
      corresponding "EGCDs" $g_q$, and normalize the "updates" of $A$ as follows. 
      The "initial@initial transformation" "update"
      $\specb{\delta}_{\init*}(\initstate)$ is post-"composed" with $g_q$,
      where $q=\speca{\delta}_{\init*}(\initstate)$. 
      Similarly, for each $a$-labelled "transition" from $q$ to $q' = \speca{\delta}_a(q)$,
      compute the "residual" of $\specb{\delta}_a(q)$ via $g_q$, and
      then post-"compose" with $g_{q'}$, basically replacing 
      $\specb{\delta}_a(q)$ by $(\residual{g_q}{\specb{\delta}_a(q)}) \comp g_{q'}$.
      Finally, replace each "final@final transformation" "update"
      $\specb{\delta}_{\halt*}(q)$ by
      $\residual{g_q}{\specb{\delta}_{\halt*}(q)}$. 
      This yields an "equivalent" "transducer" $B$, with the same "control states" as $A$, 
      and such that, for every $q\in Q$, the "EGCD" of the "vector" $V_q$ of "updates" 
      induced by the maximal runs that depart from $q$ is "isomorphic@@arrow" to the 
      "identity".
\item Construct the "initial image" $C=\Reach(B)$, as in
      Proposition~\ref{prop:images}, by restricting the "configuration space" of
      $B$ to the "closure" of the set of reachable "configurations". 
      Equivalently, this closure is the smallest inductive invariant of
      $B=(Q,\kappainit,(\kappa_a)_{a\in\Sigma},\kappahalt)$, obtained as the
      smallest fixpoint of
      \[
        F: ~ S ~\mapsto~ \bigclS(S \:\cup\: \{\kappainit(\initstate,\initdata)\}
                                   \:\cup\: \bigcup\nolimits_{a\in\Sigma}\kappa_a(S)).
      \]
      Computing this fixpoint may require data-specific methods, such as
      adaptations of techniques from linear-algebra (cf.~\cite{Karr76}).

\item Finally, merge every two "states" $q,q'$ of $C$ such that
      $V_q \viso[j] V_{q'}$ for some "isomorphism"
      $j:\Dom(V_q[\emptystr])\to\Dom(V_{q'}[\emptystr])$. 
      The difficulty here is to find the witnessing "isomorphism" $j$, if it exists.
      When only finitely many candidate "isomorphisms" $j$ exist, one can enumerate 
      them and test which $j$ satisfies $V_q=j\vcomp V_{q'}$. 
      This test reduces to a variant of the "equivalence" problem between two 
      copies of $C$, which can often be solved by a backward fixpoint computation 
      inspired by~\cite{Karr76}
      (see also~\cite{ck86,BenediktDSW17,MuschollPuppisICALP19,SeidlManethKemper18}).

      Concretely, one starts from the equation $x=x'$ between the final outputs
      of the two copies of $C$ and propagates it backwards along pairs of
      equally labelled "transitions", using effective closure of systems of
      equations under inverse images of "updates" and finite intersections
      (Lemma~\ref{lem:constraint-closure}). This yields increasingly stronger
      systems of equations associated with pairs of "control states"; under
      Assumption~\ref{ass:compactness}, and assuming decidable entailment, the
      fixpoint can be effectively detected. One then checks whether the system
      associated with $(q,q')$ entails the equation $x=v(x')$, where $x,x'$
      denote the possible data at $q,q'$.

      Once $V_q \viso[j] V_{q'}$ is detected, one merges $q'$ into $q$ by
      redirecting every "transition" entering (resp.~exiting) $q'$ to enter 
      (resp.~exit) $q$, and by pre-"composing" its "update" with $j$ 
      (resp.~post-"composing" its "update" with $\inverse{j}$). 
      By Proposition~\ref{prop:final-transducer}, the resulting "transducer" 
      is precisely $\Obs(C)=\Obs(\Reach(B))$. This "transducer" is also "isomorphic" to
      $\Obs(\Reach(A))$ since $A$ is "equivalent" to $B$ (see remark after Proposition \ref{prop:minimal}), 
      and therefore to $\finaltransd$.
\end{enumerate} 
For these steps to be effective, Assumptions~\ref{ass:compactness} 
and~\ref{ass:GCD} must be strengthened by computability requirements. 
Some are straightforward: one typically observes that "composition" 
of "updates", "residuals" by "epi" "divisors", "GCDs" of finite sets, 
and entailment between systems of equations are computable. 
These suffice for the first two steps. The remaining steps are more subtle. 
Computing the "initial image" $C=\Reach(B)$ may already be hard.
For example, for a single-"state" "transducer" $B$ over the "polynomial register algebra",
it amounts to computing the smallest algebraic invariant of polynomial maps,
which is not feasible in general~\cite[Theorem 18]{HrushovsOuakninePoulyWorrell23}.
Likewise, there is no general method for enumerating "isomorphisms" between
"domains", as needed to test $V_q\viso V_{q'}$. 
Effective solutions are therefore model-dependent.

%% file: fig-common-divisor-transducer.tex
\begin{tikzpicture}[baseline=(current bounding box.center)]
\matrix (M) [matrix of nodes, column sep=9mm, row sep=16mm, nodes={anchor=center}] {
|[name=init]| $(\initstate,\bbD_{\init*})$ & 
|[name=q0]| $(q_0,\bbD_{\init*})$ & 
|[name=q1]| $\dots$ & |[name=qi]| $(q_i,\bbD_{\init*})$ & 
|[name=q3]| $\dots$ & & & & 
\\
& & |[name=p0]| $(p_0,\bbD_\alpha)$ & 
|[name=p1]| $\dots$ & |[name=pj]| $(p_j,\bbD_\alpha)$ & 
|[name=p3]| $\dots$ & & 
\\
|[name=qf]| $(q_f,\bbD_\alpha)$ & & & & & & 
|[name=halt]| $(\haltstate,\bbD_\beta)$ & 
\\
& & |[name=r0]| $(r_0,\bbD_\alpha)$ & 
|[name=r1]| $\dots$ & |[name=rj]| $(r_j,\bbD_\alpha)$ & 
|[name=r3]| $\dots$ & & 
\\
};
\path (init) edge [partial arrow] node [above=1mm,pos=0.4] {\scriptsize $\init$} (q0);
\path (q0) edge [partial arrow] node [above=1mm,pos=0.4] {\scriptsize $\bullet$} (q1);
\path (q1) edge [partial arrow] node [above=1mm,pos=0.4] {\scriptsize $\bullet$} (qi);
\path (qi) edge [partial arrow] node [above=1mm,pos=0.4] {\scriptsize $\bullet$} (q3);
\path (qi) edge [partial arrow,out=-135,in=90] node [above=1mm,pos=0.4,sloped] {\scriptsize $\qsep: ()\mapsto d_i$} (qf);
\path (qf) edge [partial arrow,bend right,out=30] node [below=1mm,pos=0.5,sloped] {\scriptsize $\usep: d\mapsto d$} (p0);
\path (p0) edge [partial arrow] node [above=1mm,pos=0.4] {\scriptsize $\bullet$} (p1);
\path (p1) edge [partial arrow] node [above=1mm,pos=0.4] {\scriptsize $\bullet$} (pj);
\path (pj) edge [partial arrow] node [above=1mm,pos=0.4] {\scriptsize $\bullet$} (p3);
\path (pj) edge [partial arrow,bend right,in=150] node [below=1mm,pos=0.4,sloped] {\scriptsize $\halt: d\mapsto u_j(d)$} (halt);
\path (qf) edge [partial arrow,bend left,out=-30] node [above=1mm,pos=0.5,sloped] {\scriptsize $\dsep: d\mapsto f(d)$} (r0);
\path (r0) edge [partial arrow] node [above=1mm,pos=0.4] {\scriptsize $\bullet$} (r1);
\path (r1) edge [partial arrow] node [above=1mm,pos=0.4] {\scriptsize $\bullet$} (rj);
\path (rj) edge [partial arrow] node [above=1mm,pos=0.4] {\scriptsize $\bullet$} (r3);
\path (rj) edge [partial arrow,bend left,in=-150,dashed] node [above=1mm,pos=0.4,sloped] {\scriptsize $\halt: d\mapsto \residual{f}{u_j}(d)$} node [below=0.3mm,pos=0.4,sloped] {\scriptsize \qquad\qquad provided $u_j\in U$} (halt);
\end{tikzpicture}

%% file: downward-stt.tex

In this part we use Theorems \ref{thm:minimal-transducer} and \ref{thm:minimal-finite-transducer}
to show that two classes of "transducers", namely, "downward STTs" and "upward STTs",
admit "algebraically minimal" models, which can moreover be computed from "finite transducers".

\subsection{Minimization of downward STT}\label{subsec:downwardSTT}

Recall that "downward STTs" are "transducers" over the "leaf substitution algebra",
whose "updates" transform "linear terms" by applying "substitutions" at the leaves.

The first step towards effective "minimization" is to show that Assumption \ref{ass:compactness} 
(compactness of "constrained domains") holds over the "leaf substitution algebra".
In fact, in this particular "data structure", it can be shown that all 
"constrained domains" are trivial, namely, they either coincide with 
the full "domain" over which they are defined or they are empty.
This collapse, as well as the fact that all "updates" of the "data structure"
are "epis", is a consequence of a non-overlap property of "substitutions", 
which we state below without proof:

\begin{lemma}[restate=NonOverlap,
              name={\textbf{non-overlap property}}]\label{lem:non-overlap}
For all "terms" $t$ over a $\tau$-tuple $\bar x$ of variables 
and for all $\tau$-tuples of "terms" $\bar u$ and $\bar v$,
$t\subst[\bar x]{\bar u} = t\subst[\bar x]{\bar v}$
implies
$\bar u = \bar v$.
\end{lemma}

\begin{corollary}[restate=DownwardSTTcompactness, ,
                  name={\textbf{compactness}}]\label{cor:downward-stt-compactness}
Assumption \ref{ass:compactness} holds trivially for the "leaf substitution algebra".
\end{corollary}

\begin{proof}
Given an equation $f(x)=g(x)$ in the "leaf substitution algebra" $\bbD$,
with $f$ and $g$ defined by $\alpha$-tuples of "linear terms" of "type" $\beta$, 
say $\bar u$ and $\bar v$, respectively, and given another "linear term" $t$ 
of same "type" $\alpha$,
we have that $x=t$ is a solution of the equation iff $t\subst{\bar u} = t\subst{\bar v}$.
By Lemma \ref{lem:non-overlap}, every such equation is either vacuous or unsatisfiable.
\end{proof}

\begin{corollary}[restate=DownwardSTTepi,
              name={\textbf{epi updates}}]\label{cor:downward-stt-epi}
Every "update" of the "leaf substitution algebra" is an "epi".
\end{corollary}

\begin{proof}
By Lemma \ref{lem:non-overlap} it also follows that the closure $\cl(D)$ 
of any set $D\subseteq\bbD_\tau$ coincides with the entire "domain" $\bbD_\tau$ 
or it is empty. 
By the characterization provided in Lemma \ref{lem:epis}, it follows that
every "update" of the "leaf substitution algebra" is an "epi".
\end{proof}

Another consequence of the above results is that the "initial image" $\Reach(A)$
can be computed from a given "finite transducer" $A$ by simply pruning the
unreachable "control states":

\begin{corollary}[restate=DownwardSTTimages,
              name={\textbf{effective initial images}}]\label{cor:downward-stt-images}
Given a "finite@@transducer" "downward STT" $A$, 
one can compute its "initial image" $\Reach(A)$.
\end{corollary}

As for Assumption~\ref{ass:GCD} (existence of "GCDs"), recall that "updates" 
over the "leaf substitution algebra" are specified by tuples of "linear terms". 
\AP
The "divisor" relation between such "updates" coincides with the standard notion 
of "generalization" under "first-order substitution": given two tuples
$\bar u=(u_1,\dots,u_\alpha)$ and $\bar v=(v_1,\dots,v_\alpha)$ of "linear
terms", the "update" specified by $\bar u$ "divides" the one specified by
$\bar v$ iff $\bar u$ ""generalizes"" $\bar v$, that is, iff there is a tuple
$\bar w=(w_1,\dots,w_\beta)$ of "linear terms" such that
$u_i\subst{\bar w}=v_i$ for all $i=1,\dots,\alpha$.
\AP
Accordingly, a "GCD" of a set of "updates" over the "leaf substitution algebra"
is precisely a ""least general generalization"" (\reintro{anti-unifier} for short)
of the corresponding tuples of "terms". Anti-unification has been studied since
the 70's: Plotkin and Reynolds~\cite{Plotkin70,Reynolds70} independently proved
that finite sets of "terms" admit "anti-unifiers", unique up to variable renaming. 
Their theory adapts directly to our "linear terms", by requiring the witnesses 
of "anti-unification" to be "linear" as well. Although the classical results 
concern finite sets, the same existence proof extends to infinite sets;
in that case, the "anti-unifier" is no longer computable in general, but 
can still be constructed by induction.

 
\begin{lemma}[restate=DownwardSSTgcd,
              name={\textbf{see for instance \cite[Theorem 1]{Plotkin70}}}]\label{lem:downward-stt-gcd}
Assumption \ref{ass:GCD} holds for the "leaf substitution algebra",
namely, all sets of "updates" over the "leaf substitution algebra" admit "GCDs". 
Moreover, these "GCDs" can be computed when the sets are finite.
\end{lemma}

Not only one can compute "GCDs" of finite sets of "updates": using the fixpoint method 
outlined in Section \ref{sec:main-theorem} (Step (2)), one can also compute "GCDs" of 
"vectors" of "updates" induced by the "control states" of a "finite transducer".

Summing up, the existence of "minimal" "STTs" and the possibility of computing them from given 
"finite@@transducer" "STTs", follows immediately from Corollaries \ref{cor:downward-stt-compactness}
and \ref{cor:downward-stt-images}, Lemmas \ref{lem:downward-stt-gcd} and \ref{lem:gcd-vs-egcd}, 
Theorems \ref{thm:minimal-transducer} and \ref{thm:minimal-finite-transducer}, 
and from the discussion on effectiveness provided in Section \ref{sec:main-theorem}.

\begin{theorem}[restate=MinimalDownwardSTT,
                name={\textbf{minimal downward STTs}}]\label{thm:minimal-downward-stt}
"Downward STTs" admit "algebraically minimal" "transducers". 
Moreover, "minimization" can be performed
effectively on any given "finite transducer".
\end{theorem}

We conclude by noting a similarity between the 
"leaf substitution algebra" and the algebra of strings 
with "updates" that prolong strings by appending constants. 
The latter algebra can be seen as a particular case of 
the former, by simply identifying every string $w$ with a unary 
tree that has a "substitution" variable at the leaf.
In this perspective, Theorem \ref{thm:minimal-downward-stt} can 
be seen as a natural generalization of the minimization result for 
"sequential transducers" \cite{cho03}.

%% file: upward-stt.tex

\subsection{Minimization of upward STT}\label{subsec:upwardSTT}

We now turn to the "upward@@STT" variant of "STTs", with "updates" over the "free term algebra".
Here the results are slightly more difficult, but also interesting.
Recall that in the "free term algebra", an "update" from "type" $\alpha$ to "type" 
$\beta$ maps any $\alpha$-tuple $\bar u$ of "ground terms" to the $\beta$-tuple
$\bar t\subst[\bar x]{\bar u}$ of "ground terms", where $\bar t$ only depends 
on the "update" under consideration and contains "terms" over the variables 
$\bar x=(x_1,\dots,x_\alpha)$.

We first discuss which "updates" should be allowed in the "free term algebra",
with the goal of obtaining reasonable "minimal" instances of "upward STTs".
On the one hand, "erasing updates" should be forbidden, since they may produce
"minimal" "transducers" with useless registers. Indeed, if "erasing updates"
were allowed, then any two "updates" of the form
$f:(x_1,\dots,x_\alpha)\mapsto(t_1,\dots,t_\beta)$ and
$f':(x_1,\dots,x_\alpha)\mapsto(t_1,\dots,t_\beta,t_{\beta+1})$
would be mutual "divisors", implying that one could add to an "upward STT"
arbitrarily many registers loaded with useless data, without violating "minimality". 
Note that the same issue arises in other register-based models, such as the "string register algebra".
On the other hand, non-"linear" "updates" should be allowed. Consider the set
$U$ consisting of the two "updates" $\emptytuple \mapsto a(u_1,u_2)$ and
$\emptytuple \mapsto b(u_2,u_1)$, for distinct "terms" $u_1,u_2$. If only
"linear updates" were allowed, then $\emptytuple \mapsto u_1$ and
$\emptytuple \mapsto u_2$ would be incomparable maximal "common divisors" of $U$,
so $U$ would have no "greatest common divisor", and hence,
by Theorem \ref{thm:gcd-necessary}, a "minimal" "upward STT" need not exist. 
With non-"linear updates", instead, $U$ admits a "greatest common divisor",
for example, $\emptytuple \mapsto (u_1,u_2)$; indeed the latter "update" 
can be "composed" with $(x_1,x_2)\mapsto a(x_1,x_2)$ and $(x_1,x_2)\mapsto b(x_2,x_1)$ 
to recover respectively the two original "updates" of $U$.

Based on the above observations, two variants of the "free term algebra" 
remain admissible: one with "copyful", "non-erasing" "updates", and one with 
"copyless", "non-erasing" "updates".


Next, we show that Assumption \ref{ass:compactness} (compactness of "constraints") 
holds over both variants of the "free term algebra". 
We do so by embedding the "free term algebra" into the "polynomial register algebra",
and by relying on Hilbert's finite basis theorem \cite{Hilbert1890}.

\begin{lemma}[restate=UpwardSTTcompactness,
              name={\textbf{compactness}}]\label{lem:upward-stt-compactness}
Assumption \ref{ass:compactness} holds for the "free term algebra",
both in the variant that includes "copyful" "non-erasing" "updates" 
and in the variant that restricts to "copyless" "non-erasing" "updates".
\end{lemma}

\begin{proof}
Let $\bbD$ be the "free term algebra" and let $\Sigma$ 
be the ranked alphabet over which the "ground terms" of $\bbD$ are defined. 
We fix a numeric data "type" $\tau$, 
with the objective of proving that every "constraint" of "type" $\tau$ 
is equivalent to a finite subset of it.
The plan is to embed the "domain" $\bbD_\tau$ into the 
"polynomial register algebra",
and then exploit a finite basis property for the systems of polynomial equations
that correspond to the "constraints" of "type" $\tau$.
The embedding is actually formalized as the composition of 
two standard encodings: one from the "free term algebra" 
to the "string register algebra", which maps $\tau$-tuples of "ground terms"
to $\tau$-tuples of XML-like strings, and another one from the "string register algebra" 
to the "polynomial register algebra", which maps $\tau$-tuples of strings to 
$2\tau$-tuples of numbers.

To formalize the first encoding, we introduce two unranked copies of $\Sigma$, 
containing open and closed tags, respectively:
\begin{align*}
    \opentag\Sigma ~=~ \{\opentag a \::\: a\in\Sigma\}
    \qquad\qquad
    \closetag\Sigma ~=~ \{\closetag a \::\: a\in\Sigma\}.
\end{align*}
We then define the encoding $\enc$ from "terms" over $\Sigma$ 
to strings over $\underline\Sigma\uplus\overline\Sigma$, inductively as follows:
\[
  \enc(a(t_1,\dots,t_n)) ~=~ \underline a \: \enc(t_1) \: \dots \: \enc(t_n) \: \overline{a}.
\]
We naturally extend this encoding in a pointwise manner, 
so that it maps $\tau$-tuples of "ground terms" to $\tau$-tuples of strings.
It is also convenient to further extend the encoding to "terms" over the variables
$\bar x=(x_1,\dots,x_\tau)$, by simply expanding the alphabet $\underline\Sigma\uplus\overline\Sigma$ 
with those variables and by letting $\enc(x_i)=x_i$.
For example, the "term" $a(b,c(x_1))$ is encoded as 
$\underline a\,\underline b\,\overline b\,\underline c\,x_1\,\overline c\,\overline a$.

Accordingly, every "update" $f\in\bbD_\tau^\beta$,
specified by a tuple of "terms" $\bar c=(c_1,\dots,c_\beta)$ over $\bar x$
and mapping $\bar t=(t_1,\dots,t_\tau)$ to $f(\bar t)=\bar c\subst[\bar x]{\bar t}$,
is simulated by a corresponding function $\enc(f)$ on tuples of strings, so that
$\enc(f)(\enc(\bar t)) = \enc(f(\bar t))$. More precisely, $\enc(f)$ maps every $\tau$-tuple
$\bar w=(w_1,\dots,w_\tau)$ of strings to the $\beta$-tuple
$\bar u\subst[\bar x]{\bar v}$, where $\bar u = \enc(\bar c)$
and $\bar u\subst[\bar x]{\bar w}$ denotes the "substitution" in $\bar u$
of every occurrences of variable $x_i$ by the $i$-th component of $\bar w$, 
for all $i=1,\dots,\tau$. It is easy to see that the latter function
$\enc(f)$ is actually a valid "update" over the "string register algebra".

We now define the second encoding $\enc'$ from the "string register algebra" 
to the "polynomial register algebra".
For this we let $k$ be the size of the alphabet $\opentag\Sigma \uplus \closetag\Sigma$ 
and we think of each letter as a digit from $0$ to $k-1$.
The function $\enc'$ maps every $\tau$-tuple of strings $(w_1,\dots,w_\tau)$ 
to the $2\tau$-tuple of numbers 
$(r_1,s_1,\dots,r_\tau,s_\tau)$, where each $r_i$ is the number whose base-$k$ representation 
is $w_i$, and $s_i = k^{|w_i|}$. 
As for the basic operations on tuples of strings, namely, concatenation, insertion, swap,
and duplication of registers, these are simulated by corresponding "polynomial maps", as follows:
\begin{itemize}
    \item concatenation $(\dots,x_i,x_{i+1},\dots) \mapsto (\dots,x_i x_{i+1},\dots)$ 
          is simulated by the polynomial map
          $(\dots,y_i,z_i, y_{i+1},z_{i+1},\dots) \mapsto (\dots, y_i z_{i+1} + y_{i+1}, z_i z_{i+1}, \dots)$;
    \item insertion $(\dots,x_i,x_{i+1},\dots) \mapsto (\dots,x_i, a, x_{i+1},\dots)$
          is simulated by the polynomial map
          $(\dots,y_i,z_i, y_{i+1},z_{i+1},\dots) \mapsto (\dots, y_i,z_i, r_a, k, y_{i+1},z_{i+1},\dots)$,
          where $r_a$ is the digit that corresponds to the letter $a$;
    \item swap $(\dots,x_i,x_{i+1},\dots) \mapsto (\dots,x_{i+1}, x_i,\dots)$
        is simulated by 
          $(\dots, y_i,z_i, y_{i+1},z_{i+1}, \dots) \mapsto (\dots, y_{i+1},z_{i+1}, y_i,z_i, \dots)$;
    \item duplication $(\dots, x_i, \dots) \mapsto (\dots, x_i, x_i, \dots)$
        is simulated by the map
        $(\dots, y_i,z_i, \dots) \mapsto (\dots, y_i,z_i, y_i,z_i, \dots)$.
\end{itemize}
The above correspondence is naturally extended via "composition" 
to all possible "updates" over the "string register algebra".

By composing the encodings $\enc$ and $\enc'$, we obtain an embedding $\enc\comp\enc'$
of the "free term algebra" into the "polynomial register algebra".
Via this embedding, we can transform every "constraint" of "type" $\tau$ over the "free term algebra"
to a system of polynomial equations. By Hilbert's basis theorem \cite{Hilbert1890},
we know that every such system has an equivalent finite subsystem, which can then
be transformed back to a finite "constraint" over the "free term algebra",
showing that all "constrained domains" are finitary.
\end{proof}

Together with Propositions~\ref{prop:images} and \ref{prop:initial-transducer}, 
this gives the existence of the "initial image" $\Reach(B)$ for every "transducer" $B$. 
We now show that $\Reach(B)$ is computable when $B$ is "finite@@transducer". 
The key ingredient is the classical theory of unification in the "free term algebra"~\cite{Robinson65,MartelliMontanari82},
which provides finite representations of "solutions" to systems of equations, and
therefore allows one to reason about them effectively.
We provide the important definitions below.

\AP
Let $E$ be a system of equations over variables
$\bar x=(x_1,\dots,x_\alpha)$. A ""unifier"" of $E$ is an $\alpha$-tuple%
\footnote{In \cite{Robinson65,MartelliMontanari82}, a "unifier" is defined
          as a mapping from variables to "terms". Here we use an equivalent
          definition based on a tuple of "terms", where the input variables 
          are omitted but can be recovered from the arity.}
$\bar u$ of "terms" over fresh variables $\bar y=(y_1,\dots,y_\beta)$, called
""parameters"", such that every "solution" of $E$ is of the form
$\bar u\subst[\bar y]{\bar t}$ for some $\beta$-tuple $\bar t$ of "ground terms".
\AP
A "unifier" $\bar u$ is ""most general"" (\reintro{MGU}) if every other
"unifier" $\bar v$ factors through it, namely,
$\bar v=\bar u\subst[\bar y]{\bar w}$ for some tuple $\bar w$ of "terms" over
the "parameters" of $\bar v$. 
Since an "MGU", when it exists, is unique up to renaming of "parameters", 
we shall often speak of \emph{the} "MGU" of $E$.
For example, the "MGU" of the system of equations
$a(b(x_1),x_2,x_3)=a(x_2,b(x_1),x_3)$ and
$a(c(x_1),x_2,x_3)=a(x_3,x_2,c(x_1))$ is the tuple
$\bar u = (y_1,b(y_1),c(y_1))$ over the single "parameter" $y_1$.

The next lemma is a classical result of the theory of unification 
in the "free term algebra" \cite{Robinson65,MartelliMontanari82}.
It shows that the "solutions" of a finite system of equations, if they exist,
can be effectively represented in a solved (parametric) form ---essentially an "MGU"---
and this enables deciding entailment between systems of equations:

\begin{lemma}[name={\textbf{\cite[Unification Theorem]{Robinson65}, 
                            \cite[Theorem 2.3]{MartelliMontanari82}}}]
    \label{lem:parametric-form}
One can decide whether a given finite system of equations $E$ over $\bar x$
is satisfiable (i.e.~whether $\Sol(E)\neq\emptyset$), and in this case
compute an "MGU" $\bar u$ of $E$ with "parameters" $\bar y$ such that
\[
  \Sol(E) ~=~ \big\{ \bar u\subst[\bar y]{\bar w} \::\: \bar w \text{ tuple of arbitrary "ground terms"} \big\}.
\]
Moreover, $E$ entails an equation $\bar t=\bar t'$ over $\bar x$ iff 
$\bar t\subst[\bar x]{\bar u}$ and $\bar t'\subst[\bar x]{\bar u}$
are syntactically equal.
\end{lemma}

Besides implying decidability of entailment between systems of equations,
the above lemma is also crucial for computing $\Reach(B)$ from a given
"finite transducer" $B$. 
The computability of $\Reach(B)$ indeed relies on the ideas
outlined at the end of Section \ref{sec:main-theorem}, and in particular on 
a variant of an algorithm from \cite{Karr76} for constructing inductive invariants
in the "free term algebra".
The algorithm boils down to repeatedly computing "closures" of finite unions of 
"constrained domains" and of images of "constrained domains" via "updates",
which is possible thanks to Lemma \ref{lem:parametric-form}.

\begin{lemma}[restate=UpwardSTTimages,
              name={\textbf{effective initial images}}]\label{lem:upward-stt-images}
Given a "finite@@transducer" "upward STT" $B$, 
one can compute its "initial image" $\Reach(B)$.
\end{lemma}

\begin{proof}
We recall from Propositions \ref{prop:images} and \ref{prop:initial-transducer} that 
the "transducer" $\Reach(B)$ is obtained by restricting the "configuration space" of 
$B=(Q,\deltainit,(\delta_a)_{a\in\Sigma},\deltahalt)$ to the "closure@@space" of the 
set of reachable "configurations".
This "closure@@space" can be equivalently described as the least fixpoint of the monotone operator
\[
  F(S) ~=~ \bigclS(S \;\cup\; \{\delta_{\triangleright}(q_{\triangleright},d_{\triangleright})\} 
                     \;\cup\; \bigcup\nolimits_{a\in\Sigma}\delta_a(S))
\]
Thus, it suffices to show that, in the "copyless", "non-erasing", "free term algebra", 
the Kleene iteration $S_{i+1}=F(S_i)$, starting with $S_0=\emptyset$, 
can be carried out effectively and stabilizes. 
We will see that the resulting procedure is similar to the one described in \cite{Karr76}
for computing the smallest linear inductive invariant of a linear program. 

Let us discuss more in detail the step for the Kleene iteration, which boils down to computing
$F(S)$ for a given "constrained configuration space" $S \subseteq \States{Q}$.
Recall that such $S$ is of the form $\bigcup_{q\in Q}\{q\}\times D_q$, where $D_q$ is a 
"constrained domain", which, by Lemma \ref{lem:upward-stt-compactness}, is described by a 
finite system of equations, say $D_q = \Sol(E_q)$. 
For the sake of brevity, let 
\[
\begin{aligned}
  D_{\init,q} ~=~ 
  \begin{cases}
    \{d\}     & \text{if } \deltainit(\initstate,\initdata)=(q,d) \\
    \emptyset & \text{otherwise}
  \end{cases}
\end{aligned}
\qquad\quad\text{and}\quad\qquad
\begin{aligned}
  D_{p,a,q} ~=~
  \begin{cases}
    \specb\delta_a(p)(D_p) & \text{if } \speca{\delta}_a(p)=q \\
    \emptyset              & \text{otherwise}. 
  \end{cases}
\end{aligned}
\]
We get 
\begin{align*}
  F(S) &~=~ \bigcup_{q\in Q} \{q\} \times \bigcl(D_q \:\cup\: D_{\init,q} \:\cup\: \bigcup\nolimits_{p\in Q, a\in\Sigma} D_{p,a,q}) \\
       &~=~ \bigcup_{q\in Q} \{q\} \times \bigcl(D_q \:\cup\: \cl(D_{\init,q}) \:\cup\: \bigcup\nolimits_{p\in Q, a\in\Sigma} \cl(D_{p,a,q})).
       & \tag{by Lemma \ref{lem:closure-union-push}} 
\end{align*}
Therefore, in order to compute $F(S)$ it suffices to know how to construct
\begin{enumerate}
\item the "closure" $\cl(D_1 \cup D_2)$ of the union of two "constrained domains" $D_1,D_2$,
\item the "closure" $\cl(f(D))$ of the application of an "update" $f$ to a "constrained domain" $D$.
\end{enumerate}
In the following we shall focus on these two sub-problems, and finally discuss termination of
the fixpoint computation.
For both sub-problems we shall exploit Lemma \ref{lem:parametric-form}, which shows that every "constrained domain"
admits an equivalent parametric form computable as the "most general unifier" ("MGU") of the underlying 
system of equations.

\smallskip
\noindent
\underline{\emph{Closure of unions of constrained domains}.}~
We explain how to compute $\cl(D_1 \cup D_2)$ for two "constrained domains" $D_1,D_2$
represented by finite systems of equations $E_1,E_2$, respectively. Without loss of generality, we assume that 
both systems $E_1$ and $E_2$ are satisfiable: this can be decided thanks to Lemma \ref{lem:parametric-form};
in case any of the two systems turns out to be unsatisfiable, $\cl(D_1 \cup D_2)$ is defined by the other system.
Let $\bar u_1,\bar u_2$ be the "MGUs" of $E_1,E_2$, respectively. 
By the first part of Lemma \ref{lem:parametric-form}, these "MGUs" exist and can be computed from $E_1,E_2$.

By the definition of "closure", $\cl(D_1\cup D_2)$ is determined by the (generally infinite) set of equations
that are valid on $D_1 \cup D_2$, namely, those equations that are entailed by \emph{both} systems $E_1$ and $E_2$.
Using the second part of Lemma \ref{lem:parametric-form}, entailment of an equation 
$\bar t = \bar t'$ by $E_i$ (for $i=1,2$) 
can be decided by checking syntactic equality of the tuples of terms 
$\bar t\subst[\bar x]{\bar u}$ and $\bar t'\subst[\bar x]{\bar u}$,
where $\bar u$ is the "MGU" of $E_i$. 
\AP
When this entailment holds we say that the tuples $\bar t,\bar t'$ are ""unified"" by $\bar u$.
The crux here is to avoid enumerating the infinitely many possible pairs of tuples $\bar t,\bar t'$ 
that are "unified" by $\bar u$.
\AP
To avoid this, we can restrict our attention 
to tuples that are ""maximally deconstructed"", namely, tuples 
$\bar t=(t_1,\dots,t_\beta)$ and $\bar t'=(t'_1,\dots,t'_\beta)$ 
of "terms" over $\bar x$ such that, for every $1\le j\le\beta$,
either $t_j$ or $t'_j$ is a variable from $\bar x$. 
Indeed, all pairs "unified" by $\bar u$ consist of tuples of the form $\bar w\subst[\bar y]{\bar t}$
and $\bar w[\bar y]{\bar t'}$, for some arbitrary tuple $\bar w$ over $\bar y$ and some
pair $\bar t,\bar t'$ of "maximally deconstructed" tuples "unified" by $\bar u$.
We illustrate this principle with an example:

\begin{example}
Consider again the "MGU" $\bar u = (x_1, b(x_1), c(x_1))$ over a single "parameter" $x_1$. 
The following are all the possible pairs of tuples of "terms" over $\bar x=(x_1,x_2,x_3)$
that are "unified" by $\bar u$:
\begin{itemize}
  \item $\bar t = (x_1, x_2, x_3)$ and $\bar t' = (x_1, x_2, x_3)$ (trivial pair),
  \item $\bar t = (x_2, b(x_1), x_3)$ and $\bar t' = (b(x_1), x_2, x_3)$,
  \item $\bar t = (x_3, c(x_1), x_2)$ and $\bar t' = (c(x_1), x_3, x_2)$,
  \item $\bar w\subst[\bar x]{\bar t}$ and $\bar w\subst[\bar x]{\bar t'}$,
        where $\bar w$ is any tuple of "terms" over $\bar x$ 
        and $\bar t,\bar t'$ is a pair chosen from any of the previous cases.
\end{itemize}
Up to permutations of components (which are covered by the fourth case), 
the first three cases are the only pairs of "maximally deconstructed" tuples "unified" by $\bar u$.
\end{example}

We let the reader verify that
the components of two "maximally deconstructed" tuples $\bar t,\bar t'$ "unified" by $\bar u$
are either variables from $\bar x$ or "terms" from $\bar u$, and so there are only finitely
many such pairs $\bar t,\bar t'$.
Moreover, if $\bar u_1,\bar u_2$ are the "MGUs" of two systems of equations $E_1,E_2$, 
then the "closure" of $\Sol(E_1) \cup \Sol(E_2)$ is represented by the (finite) system $E$ 
of equations $\bar t = \bar t'$, for all pairs $\bar t,\bar t'$ that are 
"maximally deconstructed" and "unified" by both $\bar u_1$ and $\bar u_2$.
In particular, one can compute from two systems of equations $E_1$ and $E_2$,
representing the "constrained domains" $D_1=\Sol(E_1)$ and $D_2=\Sol(E_2)$, 
a new system of equations $E$ representing the "constrained domain" $D = \cl(D_1 \cup D_2)$.

\smallskip
\noindent
\underline{\emph{Closure of application of an update to a constrained domain}.}~
The construction of $\cl(f(D))$ for a given "update" $f$ and a given "constrained domain" $D$
follows ideas similar to the previous case. 
More precisely, let $f\in\bbD_\alpha^\beta$ be a "copyless", "non-erasing" "update" 
specified by a $\beta$-tuple of "terms" $\bar c$ over an $\alpha$-tuple of variables
$\bar x$ (so that $f: \bar t \mapsto \bar c\subst[\bar x]{\bar t}$), 
and let $E$ be a finite system of equations over $\bar x$ such that $\Sol(E)=D$
Without loss of generality, assume $\Sol(E)\neq\emptyset$. 
By the first part of Lemma \ref{lem:parametric-form}, we can construct the "MGU" $\bar u$ of $E$,
and we can then define the tuple $\bar v = \bar c\subst[\bar x]{\bar u}$.
By construction, $\bar v$ is the parametric form of the elements in $f(D)$.

Now, let $\bar y$ be a tuple of fresh variables of the same arity as $\bar c$.
Recall that $\cl(f(D))$ is the "solution set" of the (infinite) system $E'$ 
consisting of all equations $\bar t=\bar t'$ over $\bar y$ that are valid on $f(D)$.
The tuples of "terms" $\bar t,\bar t'$ that are equated by $E'$ can be equivalently 
described as those for which the original system $E$ entails the equation 
$\bar t\subst[\bar y]{\bar c} = \bar t'\subst[\bar y]{\bar c}$.
In its turn, thanks to the second part of Lemma \ref{lem:parametric-form}, 
the latter entailment is characterized by a syntactic equality between the tuples
\[
  \underbrace{\bar t\subst[\bar y]{\bar c}\subst[\bar x]{\bar u}}_{\bar t\subst[\bar y]{\bar v}}
  \qquad\text{and}\qquad
  \underbrace{\bar t'\subst[\bar y]{\bar c}\subst[\bar x]{\bar u}}_{\bar t'\subst[\bar y]{\bar v}}
\]
namely, $\bar t$ and $\bar t'$ are "unified" by $\bar v = \bar c\subst[\bar x]{\bar u}$.
Finally, by the same arguments used in the previous case ("closure" of union of "constrained domains"),
we can restrict our attention to pairs of tuples $\bar t,\bar t'$ that are 
"maximally deconstructed". Towards a conclusion, we observe that there are only finitely many
equations $\bar t = \bar t'$ where $\bar t,\bar t'$ are "maximally deconstructed" and unified by
$\bar v$, and hence one can compute a finite system $E'$ of such equations such that
$\Sol(E') = \cl(f(\Sol(E))) = \cl(f(D))$.

\smallskip
\noindent
\underline{\emph{Termination of the fixpoint procedure}.}~
It remains to argue that the Kleene iteration $S_{i+1}=F(S_i)$ stabilizes after finitely many steps.
We start by observing that the system of equations $E_{i,q}$ associated with each step $i$
and each state $q\in Q$, and representing the current "constrained configuration space" $S_i$, 
is weakened at every step --- formally, we have that each system $E_{i,q}$ entails the next system 
$E_{i+1,q}$. This implies that the corresponding "MGU" $\bar u_{i,q}$ becomes more and more general 
--- formally, each "MGU" $\bar u_{i,q}$ is an instantiation of the next "MGU" $\bar u_{i+1,q}$, meaning
that the former can be obtained from the latter by a variable "substitution". Because the instantiation
order is well-founded, it is not possible to have, at a given "control state", an infinite chain 
of "MGUs" of strictly increasing generality. Finally, because there are only finitely many "control states", 
it follows that $S_i$ is eventually constant, allowing the fixpoint procedure to terminate.
\end{proof}

Finally, we prove the existence of "GCDs" for sets of "updates" over the
"free term algebra" with \emph{"copyless"}, "non-erasing" "updates"; later,
we show that "GCDs" may fail for the "copyful", "non-erasing" variant.
The key step is an interpolation property for the "divisor" relation: 
for any two "updates" $f,f'$ with the same "domain", there is an "update" $g$ 
such that $f$ and $f'$ both "divide" $g$, and such that $g$ "divides" every 
"update" $h$ divided by both $f$ and $f'$. 
This interpolation property relies on two crucial facts: (1) the "divisor" preorder 
can be characterized as a suitable "embedding" relation between multisets 
of "terms", and (2) overlapping occurrences of subterms are necessarily nested. 
Once interpolation is established, one constructs a "GCD" of a set $U$ of "updates" 
by considering the set $G$ of all "common divisors" of $U$: any two maximal elements 
of $G$, with respect to the "divisor" preorder, turn out to "divide" one another, 
and hence are both \emph{"greatest@GCD"} "common divisors" of $U$.

\begin{lemma}[restate=UpwardSTTgcd,
              name={\textbf{GCDs}}]\label{lem:upward-stt-gcd}
Assumption \ref{ass:GCD} holds over
the "free term algebra" with "copyless" and "non-erasing" "updates".
Moreover, "GCDs" can be computed for finite sets of "updates".
\end{lemma}

\begin{proof}
Let $\bbD$ denote the "free term algebra".
Recall that the "updates" over $\bbD$ are represented by tuples of "terms" 
with variables at the leaves, and that for "copyless", "non-erasing" "updates" 
every variable occurs exactly once. 
Also recall that if two tuples of "terms" are permutations one of another, 
then they represent "isomorphic@@arrow" "updates", which are in particular  
equivalent w.r.t.~"divisor" relation.
For this reason, in this proof it is convenient to further abstract our representations
of "updates" using \emph{multisets of "terms"}, rather than tuples. More precisely,
up to "isomorphism@@arrow", every "update" $f\in\bbD_\alpha^\beta$ is represented
by a multiset $S$ of cardinality $\beta$ that contains "terms" over a number $\alpha$ 
of variables. 
We treat the elements of a multiset, and in particular the possible duplicates 
of the same "term", as distinct objects, each associated with a "term".
\AP
A ""copy"" of a "term" $t$ inside another "term" $t'$ (resp.~inside a multiset $S$)
is an occurrence of $t$ as a subterm of $t'$ (resp.~as a subterm of some element of $S$). 
Different occurrences of the same "term" $t$ are again treated as different objects.
We claim without a proof that the "divisor" relation is characterized as follows:

\begin{myclaim}\label{claim:subterm-embedding}
The "update" represented by a multiset $S$ is a "divisor" of the "update" represented by 
another multiset $T$ 
iff there is a function $e$ that maps every element $s\in S$ to a "copy" of $s$ 
inside $T$, in such a way that 
\begin{enumerate}
  \item the images of $e$, seen as occurrences of subterms inside $T$, are pairwise 
        ""disjoint"", meaning they are not contained one in another as subterms;
  \item the images of $e$ cover all the variables that appear in $T$.
\end{enumerate}
\end{myclaim}

\noindent
\AP
A function $e$ as above is called a ""subterm embedding"" of $S$ into $T$,
and it is denoted for short by $e: S \mathrel{\reintro{\emb*}} T$. 
Below is an example:
\begin{align*}
\input{fig-subterm-embedding}
\end{align*}

When we write $S,T,\dots \:\emb\: U,V,\dots$ we mean that there are "embeddings@@subterm" 
of every multiset of the left hand-side into every multiset of the right hand-side.
We prove the following interpolation property.

\begin{myclaim}\label{claim:interpolation}
For all multisets $S,T$, there is a multiset $U$ such that $S,T\emb U$ and 
for all multisets $V$,
\[
  S,T \emb V
  \qquad\text{implies}\qquad
  S,T \emb U \emb V.
\]
\end{myclaim}

\begin{proof}
It is tempting to define $U$ as the disjoint union $S+T$ of $S$ and $T$, 
since $S,T\emb U$ would hold trivially. 
However, we cannot claim that if there are "embeddings@@subterm" 
$e_S: S\emb V$ and $e_T: T\emb V$, then the union $e_S+e_T$ of these
"embeddings@@subterm" is an "embedding@@subterm" $S+T$ into $V$, 
essentially because $e_S + e_T$ may violate the requirement of images to be "disjoint". 
Below, we see an example of a violation of this requirement
between an element $s$ that originally belonged to $S$ 
and another element $t$ that originally belonged to $T$
\AP
(we use colors to distinguish the ""origin"", i.e.~$S$ or $T$,
of each element of $S+T$ and the "embeddings@@subterm" into $V$):
\begin{align*}
\input{fig-embedding-violation}
\end{align*}
To correctly define the interpolant $U$ we need to carefully select the
elements from $S+T$. This can be done inductively using a greedy strategy,
as follows.
We start by listing the elements of $S+T$ based on a total order $<$
that refines the subterm partial order, starting from the largest element,
that is: $u_1 > u_2 > \dots > u_n$. Then, following this fixed order, 
we add each element $u_i$ to a forest, either as a new root or as a child 
of a previously added element.
In doing so, we shall prefer placing $u_i$ as a child of another element $u_j$ 
($j<i$), provided that the following conditions are satisfied:
\begin{itemize}
\item $u_j$ is a root in the current forest, 
\item $u_j$ contains an occurrence of $u_i$ as a subterm 
      that is "disjoint" from all current children of $u_j$
      (seen as other occurrences of subterms),
\item $u_j$ and $u_i$ have opposite "origins", namely,
      either $u_j\in S$ and $u_i\in T$, or, vice versa, $u_j\in T$ and $u_i\in S$.
\end{itemize}
If more than one element $u_j$ satisfies the above conditions,
then we choose any of them, say the first one according to the fixed order.
Otherwise, if there is no element $u_j$ that satisfies the above conditions, 
then we add $u_i$ as a new root of the forest.
This clearly produces a forest of height at most $2$, in which the children
of ever root have opposite "origins" than the root and represent disjoint
subterm occurrences.
Finally, we define $U$ as the multiset consisting of the roots of the forest.

It is easy to see that $S \emb U$ (and similarly $T\emb U$). Indeed,
every element $s$ of $S$ must appear in the forest either as a root,
say $r_s$, or as a child, say $c_s$; in the former case, we can map 
$s$ to the root $r_s$ itself, which clearly belongs to $U$; in the 
latter case, we can map $s$ to a "copy" inside the parent of $c_s$, 
which is a root of the forest and hence belongs to $U$.

Next, we consider a multiset $V$ admitting "embeddings@@subterm" 
$e_S: S \emb V$ and $e_T: T \emb V$, and we aim at constructing an
"embedding@@subterm" $e: U \emb V$.
We shall define the image $e(u_i)$ of each element $u_i$ of $U$ 
inductively, following again the total order on $S+T$ that we 
fixed earlier (recall that $U$ is contained in $S+T$). 
In doing so, we shall guarantee the following invariant:

\begin{quote}
\emph{$e(u_i)$ is a "copy" of $u_i$ inside $(e_S + e_T)(u_j)$
      for some $u_j\in U$ with $j\le i$ 
      (possibly $j=i$ and $e(u_i)=(e_S+e_T)(u_i)$).}
\end{quote}

Suppose, without loss of generality, that $u_i$ has "origin" in $S$.
We distinguish two cases based on the nesting relationships between the 
previously defined images $e(u_j)$, for all $u_j\in U$ with $j<i$, and 
the image $e_S(u_i)$ induced by the original "embedding@@subterm" of $S$ into $U$:
\begin{enumerate}
\item $e_S(u_i)$ is "disjoint" from all previously defined images $e(u_j)$,
      for all $u_j\in U$ with $j<i$.
      In this case, we simply let $e(u_i) = e_S(u_i)$.
\item $e_S(u_i)$ is nested inside a previously defined image $e(u_j)$,
      for some $u_j\in U$ with $j<i$.
      Without loss of generality, we can assume that $u_j$ is the first
      element of $U$ whose image $e(u_j)$ contains $e_S(u_i)$.
      In particular, by the inductive invariant, this means that
      $e(u_j) = (e_S + e_T)(u_j)$.
      First, we claim that $u_j$ must have different "origin" than $u_i$,
      namely, $u_j\in T$: indeed, if this were not the case, then
      $e_S$ would map the two elements $u_i$ and $u_j$ to 
      non-"disjoint" "copies" inside $U$, thus contradicting the
      definition of "embedding@@subterm".
      Second, because $u_i$ was added to the forest as a new root, 
      and not as a child of $u_j$, this means that $u_j$ already contained 
      subterms $u_{k_1},\dots,u_{k_m}$, for $m\ge 1$ and $j<k_1,\dots,k_m<i$, 
      that cover all occurrences of $u_i$ as a subterm of $u_j$. 
      Of course, the elements $u_{k_1},\dots,u_{k_m}$, 
      being children of $u_j$, do not belong to $U$, and have different 
      "origin" than $u_j$, so they all belong to $S - U$ (i.e.~the multiset 
      difference between $S$ and $U$).
      Also note that these "terms" $u_{k_1},\dots,u_{k_m}$, 
      being elements of $S$, must be mapped via $e_S$ to "copies" inside $U$.
      A simple Pigeonhole argument then shows that that at least one 
      such element $u_{k_\ell}$ must be mapped via $e_S$ to a "copy" 
      inside $U$, but \emph{outside} $u_j$.
      Accordingly, we define $e(u_i)$ as any occurrence of $u_i$ as
      a subterm of $e_S(u_{k_\ell})$, which is "disjoint" from $e(u_j)$.
\end{enumerate}
Based on the above arguments and constructions, the mapping $e$ 
is a valid "embedding@@subterm" of $U$ into $V$, and thus
$S,T\emb U\emb V$.
\end{proof}

It remains to prove, using the above claims, the existence of GCDs
for sets of "updates", and their computability when the sets are finite. 
Let $H$ be any set of "updates" with the same "domain", and let $F$ be 
the set of all "common divisors" of $H$.
Further let $G$ be the set of \emph{maximal} (not necessarily greatest) 
elements of $F$ w.r.t.~the "divisor" preorder, namely: $G$ contains an "update" 
$g$ iff $g\in F$ and for all $f\in F$, if $g$ is a "divisor" of $f$, 
then $f$ is a "divisor" of $g$, meaning that 
$f$ and $g$ are equivalent w.r.t.~the "divisor" preorder.
Of course $F$, and thus $G$ as well, is not empty, since the identity
"update" is clearly a "common divisor" of $H$.

Now, we head towards proving that every "update" in $G$ is a "greatest common divisor" of $H$
(this also means that the "updates" in $G$ are all equivalent w.r.t.~the "divisor" preorder).
Let $g\in G$ and $f\in F$. 
Since both $f$ and $g$ are "common divisors" of $H$,
by Claims \ref{claim:subterm-embedding} and \ref{claim:interpolation},
there is a "common divisor" $g'$ of $H$ which is "divided" by both $f$ and $g$.
However, because $g$ was chosen to be maximal in $F$, we also know that
$g'$ is a "divisor" of $g$, and hence, by transitivity, $f$ is a "divisor" of $g$.
This shows that $g$ is a "greatest common divisor" of $H$.

Finally, as concerns the computability of such a "GCD" of $H$, when $H$ is finite,
it suffices to proceed by a fixpoint computation: starting from the identity "update",
which is clearly a "common divisor" of $H$, one repeatedly extend the "update"
(or equally, one of "terms" in the multiset that represents it), while preserving
the property of being a "common divisor", until no further extension is possible.
Of course, each extension step is effective and the process will terminate in a 
finite number of steps.
Moreover, even though at each step there might be multiple possible extensions
applicable, any maximal sequence of extension steps will eventually produce the same "GCD" 
of $H$ up to "isomorphism", precisely because of Claim \ref{claim:interpolation}.
This gives an effective procedure to compute a "GCD" from a given finite set $H$
of "updates".
\end{proof}

As usual, the "minimization" result follows from the previous lemmas
and from Theorems \ref{thm:minimal-transducer} and \ref{thm:minimal-finite-transducer}. 
As for effectiveness, given a "finite@@transducer" "upward STT" $A$, one can compute the 
"EGCDs" associated with its "control states", by using 
(1) the fixpoint algorithm described at the end of Section \ref{sec:main-theorem}
and (2) the effectiveness of "GCDs" of finite sets. 
We recall that this enables a normalization of $A$ into 
an "equivalent" "transducer" $B$.
We also recall from Lemma \ref{lem:upward-stt-images} that
one can compute $C=\Reach(B)$.
Finally, "states" with $\viso$-equivalent "residual vectors"
can be detected, and then merged, because
(1) for any two "vectors" $V,V'$, we have $V \viso V'$
iff there is a permutation $\pi$ of the components of the 
tuples of the "domain" of $V$ such that 
$V[t](d_1,\dots,d_n) = V'[t](d_{\pi(1)},\dots,d_{\pi(n)})$,
for all $t\in\Sigma^*$ with $V[t]\neq\bot$ and all 
$(d_1,\dots,d_n)\in\Dom(V[t])$,
(2) the latter condition is a variant of an "equivalence" 
problem, which can be decided using the standard techniques 
described at the end Section \ref{sec:main-theorem}, and 
(3) there are only finitely many such permutations $\pi$,
which can then be verified one by one.

\begin{theorem}[restate=MinimalUpwardSTT,
                name={\textbf{minimal upward STTs}}]\label{thm:minimal-upward-stt}
"Upward STTs", with "updates" restricted to be "copyless" and "non-erasing", 
admit "algebraically minimal" "transducers".
Moreover, "minimization" can be performed
effectively on any given "finite transducer".
\end{theorem}

Concerning the possibility of a "minimization" result 
for "upward STTs" with \emph{"copyful"} "updates",
in view of Theorem \ref{thm:gcd-necessary}, it suffices to
show that "GCDs" not always exist.
Consider the updates $f,f',g \in \bbD_1^2$ and $g'\in\bbD_1^3$ defined by
\[
\begin{array}{rcclrccl}
    f:  & x & \mapsto & \big( a(x), \, b(a(x)) \big) &  \qquad\qquad
    g:  & x & \mapsto & \big( b(a(x)), \, c(a(x)) \big)
    \\
    f': & x & \mapsto & \big( a(x), \, c(a(x)) \big) &  \qquad\qquad 
    g': & x & \mapsto & \big( a(x), \, b(a(x)), \, c(a(x)) \big).
\end{array}
\]
We have that both $f$ and $f'$ are "common divisors" of $U=\{g,g'\}$,
but they are not one a "divisor" or another (this is because we forbid "erasing updates").
If a "GCD" $h$ existed for $U=\{g,g'\}$, then this would necessarily be "divided"
by $f$ and $f'$. Moreover, since neither $b(a(x))$ nor $c(a(x))$ can be obtained 
from the other by an "update", $h$ is also "divided" by $g$.
Since $g'$ is "divided" by $h$, we obtain by transitivity that $g'$ is also "divided" by $g$.
However, we see that this is not possible, because no "update" can produce
the additional term $a(x)$ of $g'$ without having either $x$ or $a(x)$ in the input
(here we are relying on the fact that the variable $x$ can take at least two different values).


\medskip
We conclude by returning to the "minimization" result for "copyless", "upward STTs",
in order to highlight a connection with resource minimization.
This setting is indeed an example where "algebraic minimization" partially coincides with resource minimization.
It is immediate to see that a "finite@@transducer", "algebraically minimal" "transducer" $\finaltransd$ 
has a minimum number of "control states" among all "equivalent" "transducers": 
this holds for all models because, by definition, 
$\finaltransd$ is a "subquotient" of every "equivalent" "transducer", and 
neither a "subobject" nor a "quotient" can have a larger number of "control states". 
In the specific case of "upward STTs", where every data is a tuple of "terms",
it is natural to think of each position of the tuple as a register, and 
thus wonder whether "algebraic minimization" achieves also register minimization.
This holds only to some extent. Indeed, an "algebraically minimal" 
"upward STT" can have a "control state" whose "constrained domain" 
forces the value of a certain register to be constant; this register can then be 
removed while preserving "equivalence".
Nevertheless, it can be shown that every "algebraically minimal" "upward STT"
optimizes, at each "control state", the number of \emph{non-constant} registers 
among all "equivalent" "transducers" with the same control structure.

%% file: fig-subterm-embedding.tex
\begin{tikzpicture}[baseline=(current bounding box.center)]
\matrix (M1) [matrix of math nodes, column sep=2mm, row sep=2mm, anchor=west] {
  |[name=S]| S & = & \Big\{~ & |[name=t1]| {} & ~,~ & |[name=t2]| {} & ~,~ & |[name=t3]| {} & ~\Big\} \\
};

\matrix (M2) [matrix of math nodes, column sep=2mm, row sep=2mm, anchor=north west] at ([yshift=-8mm]M1.south west) {
  |[name=T]| T & = & \Bigg\{~~ & |[name=t1p]| {} & |[name=tp]| {} & |[name=t3p]| {} & \quad,\qquad & |[name=t2p]| {} & \quad~~\Bigg\} \\
};

\pic at (t2p) {tri={w=1.3, h=1.3, color=gray}};
\pic[yshift=-4mm] (T2) at (t2p) {tri={w=0.5, h=0.5, color=nicecyan}};
\path (t2.center) edge [bold arrow, nicecyan, bend right=10] (T2-apex);
\pic at (t2) {tri={w=0.5, h=0.5, color=nicecyan}};
\pic[yshift=-4mm] (T3) at (t3p) {tri={w=0.5, h=0.5, color=nicecyan}};
\path (t3.center) edge [bold arrow, outline, nicecyan, bend left=10] (T3-apex);
\path (t3.center) edge [bold arrow, nicecyan, bend left=10] (T3-apex);
\pic at (t3) {tri={w=0.5, h=0.5, color=nicecyan}};
\pic at (tp) {tri={w=1.3, h=1.3, color=gray}};
\pic[yshift=-4mm] at (t3p) {tri={w=0.5, h=0.5, color=nicecyan}};
\pic[yshift=-4mm] (T1) at (t1p) {tri={w=0.5, h=0.5, color=nicecyan}};
\path (t1.center) edge [bold arrow, nicecyan, bend right=30] (T1-apex);
\pic at (t1) {tri={w=0.5, h=0.5, color=nicecyan}};
\end{tikzpicture}

%% file: fig-embedding-violation.tex
\begin{tikzpicture}[baseline=(current bounding box.center)]
\matrix (M1) [matrix of math nodes, column sep=2mm, row sep=2mm, anchor=west] {
  |[name=SU]| S\uplus T & = & \Big\{~ & |[name=t1]| {} & ~,~ & |[name=t2]| {} & ~,~ & |[name=t3]| {} & ~,~ & |[name=t4]| {} & ~,~ & |[name=t5]| {} & \Big\} 
  \\
};

\matrix (M2) [matrix of math nodes, column sep=2mm, row sep=2mm, anchor=north west] at ([yshift=-8mm]M1.south west) {
  |[name=V]| \phantom{S\uplus}~V & = & \Bigg\{~~ & |[name=t1p]| {} & |[name=tp]| {} & |[name=t3p]| {} & & \quad,\quad & & |[name=t2p]| {} & |[name=tpp]| {} & |[name=t5p]| {} & \quad\Bigg\} 
  \\
};

\node[nicecyan,yshift=4mm] at (t2) {$s$};
\node[nicered,yshift=4mm] at (t4) {$t$};
\pic at (tpp) {tri={w=1.3, h=1.3, color=gray}};
\pic[yshift=-4mm] (T2) at (t2p) {tri={w=0.5, h=0.5, color=nicecyan}};
\path (t2.center) edge [bold arrow, nicecyan, bend right=10] (T2-apex);
\pic at (t2) {tri={w=0.5, h=0.5, color=nicecyan}};
\pic[yshift=-4mm] (T3) at (t3p) {tri={w=0.3, h=0.3, color=nicered}};
\path (t3.center) edge [bold arrow, outline, nicered, bend left=10] (T3-apex);
\pic at (tp) {tri={w=1.3, h=1.3, color=gray}};
\pic[yshift=-5mm] at (t3p) {tri={w=0.3, h=0.3, color=nicered}};
\path (t3.center) edge [bold arrow, nicered, bend left=10] (T3-apex);
\pic at (t3) {tri={w=0.3, h=0.3, color=nicered}};
\pic[yshift=-4mm] (T1) at (t1p) {tri={w=0.5, h=0.5, color=nicecyan}};
\path (t1.center) edge [bold arrow, nicecyan, bend right=30] (T1-apex);
\pic at (t1) {tri={w=0.5, h=0.5, color=nicecyan}};
\pic[yshift=-5mm] (T4) at (t2p) {tri={w=0.3, h=0.3, color=nicered}};
\path (t4.center) edge [bold arrow, nicered, bend left=10] (T4-apex);
\pic at (t4) {tri={w=0.3, h=0.3, color=nicered}};
\pic[yshift=-5mm] (T5) at (t5p) {tri={w=0.3, h=0.3, color=nicered}};
\path (t5.center) edge [bold arrow, nicered, bend left=10] (T5-apex);
\pic at (t5) {tri={w=0.3, h=0.3, color=nicered}};
\end{tikzpicture}

%% file: conclusion.tex

We have presented sufficient and necessary conditions for the existence of
"algebraically minimal" models of "streaming transducers" over a general
"data structure" $\bbD$. The first condition is the possibility of enriching
$\bbD$ with "constrained domains", defined by systems of equations, providing
the closure properties needed to construct "images" of "transducer morphisms". 
The second condition is the existence of "greatest common divisors" for non-empty 
sets of "updates", providing the algebraic ingredient needed to construct the
"final" object underlying "minimization".

After developing the abstract framework, we instantiated it to concrete
"data structures" and obtained effective "minimization" results for two
variants of string-to-"term" "transducers", which construct their "outputs"
incrementally by extending "terms" either at leaves ("downward STTs") or at the
roots ("upward STTs"). In both cases, the minimization procedure can be understood
as a non-trivial generalization of Choffrut's "minimization" of "sequential
transducers"~\cite{cho03}, with anti-unification and unification playing the
respective roles needed to compute quotients and subobjects.

Beyond the "minimization" results established here, the proposed framework 
opens several directions for further investigation. First, our results suggest
possible applications to logical characterizations. Specifically, based on the
characterization of the first-order fragment of string-to-string order-preserving
transductions in~\cite{FiliotGauwinLhote16}, corresponding to "sequential
transducers" with "regular look-ahead", it seems plausible that the
"minimization" results developed here could be used to obtain analogous
characterizations for "STTs" with "regular look-ahead".

A second direction seeks a common generalization of "downward@@STT" and
"upward STTs" within a richer model that still admits "minimization". One possible
route is to allow "updates" defined by second-order "term" "substitutions", or
by suitable restricted variants of them. Along this line, the use of "GCDs" in
our framework appears to connect with open problems on anti-unification for
second-order "term" "substitution", also known as higher-order anti-unification
\cite{Pfenning1991,BaumgartnerKutsiaLevyVillaret2017}.

It would also be interesting to understand how far the present approach extends
to other "data structures", such as the "polynomial register algebra" or the
"string register algebra". In this direction, we already have preliminary
"minimization" results for "transducers" with univariate polynomial "updates",
based on Ritt's decomposition theory of polynomials~\cite{Ritt1922,RittSurvey90}.
Another related question is which equational theories can be imposed on the
"free term algebra" while preserving the existence of "GCDs". For instance, such
theories could allow strings to be represented as "terms" modulo associativity
of concatenation.

Finally, the algorithmic aspects of the theory deserve further study. In
particular, while our results give effective "minimization" procedures for the
classes considered here, the complexity of computing "minimal" "transducers"
from "finite@@transducer" ones remains to be analyzed systematically.